\documentclass[aoas,rotating,preprint,doublespacing]{imsart}

\RequirePackage{amsthm,amsmath,amsfonts,amssymb}
\RequirePackage[authoryear]{natbib}
\RequirePackage[colorlinks,citecolor=blue,urlcolor=blue]{hyperref}
\RequirePackage{graphicx}

\startlocaldefs

\usepackage{mathtools}
\usepackage{bm}

\usepackage{booktabs}
\usepackage{array}

\usepackage{algorithm}
\usepackage{algpseudocode}

\usepackage{tikz}
\usetikzlibrary{positioning,fit,calc,arrows.meta}

\usepackage{enumitem}

\usepackage{subcaption}
\usepackage{url}

\setpkgattr{abstract}{width}{\textwidth}

\DeclareMathAlphabet{\mathbbold}{U}{bbold}{m}{n}
\let\amsmathbb\mathbb
\renewcommand{\mathbb}[1]{\ifnum9<1\string#1 \mathbbold{#1}\else\amsmathbb{#1}\fi}

\newcommand{\mathbfcal}[1]{\bm{\mathcal{#1}}}
\newcommand{\indep}{\perp \!\!\! \perp}
\newcommand{\RNum}[1]{\uppercase\expandafter{\romannumeral #1\relax}}

\theoremstyle{plain}
\newtheorem{prop}{Proposition}

\theoremstyle{definition}
\newtheorem{assumption}{Assumption}

\makeatletter
\newenvironment{breakablealgorithm}
  {%
   \begin{center}
     \refstepcounter{algorithm}%
     \hrule height.8pt depth0pt \kern2pt%
     \renewcommand{\caption}[2][\relax]{%
       {\raggedright\textbf{\ALG@name~\thealgorithm} ##2\par}%
       \ifx\relax##1\relax
         \addcontentsline{loa}{algorithm}{\protect\numberline{\thealgorithm}##2}%
       \else
         \addcontentsline{loa}{algorithm}{\protect\numberline{\thealgorithm}##1}%
       \fi
       \kern2pt\hrule\kern2pt
     }
  }{%
     \kern2pt\hrule\relax%
   \end{center}
  }
\makeatother

\endlocaldefs

\begin{document}

\begin{frontmatter}

\title{Causal mediation analysis for zero-inflated longitudinal data
in the presence of treatment non-compliance and multiple mediators}
\runtitle{Causal mediation with non-compliance and multiple mediators}

\begin{aug}
\author[A]{\fnms{Saurabh}~\snm{Bhandari}\ead[label=e1]{sbhandari52@uchicago.edu}}
\author[B]{\fnms{Wreetabrata}~\snm{Kar}\ead[label=e2]{wxk5196@psu.edu}}
\author[C]{\fnms{Michael~J.}~\snm{Daniels}\ead[label=e3]{daniels@ufl.edu}}
\author[D]{\fnms{Bikram}~\snm{Karmakar}\ead[label=e4]{bkarmakar@wisc.edu}}

\address[A]{Department of Public Health Sciences,
University of Chicago\printead[presep={,\ }]{e1}}

\address[B]{Smeal College of Business, The Pennsylvania State University\printead[presep={,\ }]{e2}}

\address[C]{Department of Statistics,
University of Florida\printead[presep={,\ }]{e3}}

\address[D]{Department of Statistics,
University of Wisconsin--Madison\printead[presep={,\ }]{e4}}
\end{aug}

\begin{abstract}
Understanding whether a digital marketing campaign is effective is central to
designing effective customer engagement strategies. We analyze a large-scale,
longitudinal promotional email campaign conducted by a U.S.\ retailer to
evaluate how value-added incentives, such as free shipping, compare with
traditional price discounts in influencing customer purchasing behavior. The
analysis is complicated by non-compliance, due to not opening emails, multiple
longitudinal mediators, and zero-inflated mediators and purchase outcomes. To
address these challenges, we develop a Bayesian causal mediation framework
based on enriched Dirichlet process mixture models and estimate the causal
estimands using a scalable G-computation algorithm. We show that analyses
ignoring email-opening behavior substantially attenuate estimated effects.
Value-added incentives consistently outperform price discounts, yielding higher
estimated potential purchase amounts, with benefits accumulating over time. We
design an individualized sequential emailing strategy that optimizes expected
purchase count in the observed data.
\end{abstract}

\begin{keyword}
\kwd{Causal mediation analysis}
\kwd{Digital communication}
\kwd{EDPM}
\kwd{Longitudinal data}
\kwd{Multiple mediators}
\kwd{Treatment non-compliance}
\end{keyword}

\end{frontmatter}

\newpage

\section{Introduction}\label{sec:1}
Digital communications, such as targeted emails or text messages, are convenient mediums for companies to connect with both potential and existing customers in terms of cost, reach, and customization. The success of these communications is evaluated through the recipient's response, such as participating in a survey or making a purchase. Accurate estimation of the efficacy of such digital communications is necessary to inform effective future business strategies. However, these communication strategies may encounter `treatment non-compliance' in the form of recipients ignoring emails and messages. Further, when promotional emails are sent, and vary, over time, the recipient's response to the current email is affected by communications in the past, and may affect future communications and outcomes. These dynamics pose challenges for separately inferring the causal effect of the digital campaign from a non-compliance or time-varying-confounding bias. 

In this context, we evaluate a large observational dataset from the targeted promotional email marketing outreach of a U.S.~retailer who sent a sequence of three emails to their large customer base. Each email offers either a value-added incentive, offering free returns or free shipping, or a price discount incentive, offering a 5\% discount. We aim to address the following causal questions: How do value-added incentive emails compare with price incentive emails in their effects on customer purchasing behavior? How do these effects vary in the presence of individuals who either ignore the emails altogether or choose not to accept the discount even after opening them? What is the optimal emailing strategy for sending these two types of promotional emails over time so that it maximizes the final purchase count? What are the mediated effects of these promotional emails through intermediate customer actions? By answering these questions, our aim is to clarify how the promotion strategy causally affects product purchases and thereby guide more effective future digital communication strategies. %



There is substantial non-compliance in the dataset, with nearly five out of six customers failing to open all three promotional emails (Figure~\ref{fig:OpenerSubPopnAcrossTime} in the Supplementary Material). Unlike the traditional treatment-versus-control causal inference settings, where non-compliance typically corresponds to receiving the control condition instead of the assigned treatment, our study compares two active email strategies. Here, compliance refers to engaging with the assigned promotion by opening the email and following its offer link. Accordingly, customers who do not comply are viewed simply as not opening (and therefore not following) the assigned promotional email. As a result, a direct comparison of the two email strategies may reflect differences in customers' propensity to engage with emails rather than differences attributable to the email content itself. We adopt a principal stratification framework and classify customers into latent subgroups based on their potential email-opening behavior: active customers who would open both email types, value-attentive customers who would open only value-added emails, price-attentive customers who would open only price-discount emails, and non-active customers who would open neither. We then estimate causal effects within these principal strata and characterize the assumptions required for their identification.

In addition, we want to infer whether the longitudinal intermediate variables in the data mediate the relationship between messages sent to the customer and the final number of purchases resulting from these emails.  We hypothesize that receiving promotional emails prompts customers to open the email more promptly, motivating them to make a purchase and ultimately increasing the number of purchased items. Therefore, we consider the two variables, customers' time since the last email was opened and the time since the last purchase, as the potential time-varying mediators.%

A common challenge in analyses of promotions is that effects are generally not substantial, so a large dataset is required to infer a significant effect. Further, a notable challenge in our dataset, likely common in similar applications, is that the number of purchases is highly right-skewed, with most customers either not making any purchases or making only a few purchases. Similar non-normal structures are also seen in our mediator variables. We carefully overcome these empirical challenges using a hurdle model in a Bayesian Enriched Dirichlet Process framework and an efficient implementation of posterior computation and estimation. Finally, we report, in Section \ref{sec:5}, robust, statistically significant direct and mediated indirect effects.


\subsection{Summary of our contributions}\label{sec:11} 
 Our contributions are summarized below: 
\begin{enumerate}
    \item We provide an inferential framework to evaluate the causal effect in a promotional email campaign in the presence of multiple intermediate variables.  To do this, we define principal interventional effects based on static interventions on the email-type and randomized interventions on the mediator variables (formalized in Section \ref{sec:3.2}). We clarify the required ignorability assumptions of the treatment assignment, mediator value, and email opening behavior that are required to identify these causal effects. Some assumptions, e.g., ignorable treatment assignment, are justified in our study by how the data are generated. We conduct sensitivity analyses to assess robustness to potential violations of the other assumptions.

    \item  We propose a semi-parametric modeling framework using enriched Dirichlet process mixture (EDPM) models \citep{wade2011enriched, wade2014improving} for jointly modeling the longitudinal treatments, covariates, mediators, and the outcome. EDPM models form a rich class of models that can flexibly model multivariate dependence while allowing for tractable Bayesian posterior sampling.
    The added challenges to our problem are longitudinal structures, zero-inflated variables that are also right-skewed, and multiple intermediate variables. 
    We customize our EDPM model to incorporate these data structures. 

    \item Through our analysis, we clarify that information on email opening should be incorporated in the inference for the causal effect of the promotional email. We show that the estimated effects are substantially attenuated when email opening information is omitted. Our analysis builds four latent strata of customers based on whether and when they would open the promotional email. Since a promotional email becomes effective only when the email is opened, we also evaluate if customer behavior were modulated so that they were forced to always open the email, the expected purchase count would increase. 

    \item We provide an individualized sequential emailing strategy that optimizes the expected purchase count in the observed data. This strategy comprises simple logistic models at each time point, making them easy to implement. A practical benefit of this logistic assignment model is its interpretability for stakeholders and transparency when presented to regulators. We show empirically  that this individualized email marketing strategy improves outcomes across all latent groups compared to any global strategy. 
    
    \item A key to our analysis of 18{,}571 customers is a new algorithm for estimating the direct and indirect effects of an exposure on the outcome using the G-formula \citep{robins1986new}. Nonparametric identification results for interventional effects in the presence of multiple time-varying mediators have been established by \cite{tai2023causal} and \cite{diaz2023efficient}. Our contribution differs primarily in the estimation approach. Whereas \cite{tai2023causal} rely on parametric models that require correct model specification and \cite{diaz2023efficient} develop estimators based on semiparametric efficiency theory, we adopt a flexible Bayesian framework that accommodates complex data structures while providing a unified approach for both point estimation and uncertainty quantification of interventional effects. 
\end{enumerate}

\subsection{Organization of the paper}\label{sec:12}
 This article is structured as follows. In Section \ref{sec:2}, we introduce our case study and dataset. Section \ref{sec:3} describes our causal framework. In Section \ref{sec:4}, we present our proposed semi-parametric model specification for the observed data and provide details on parameter estimation. Sections \ref{sec:5} and \ref{sec:6} provide results from the real-data analysis and findings from simulation studies, respectively. We conclude with a brief discussion in Section \ref{sec:7}. 

\section{Description of the data and literature review}\label{sec:2}

\subsection{Data}\label{sec:21}

Our case study involves targeted e-mail promotions employed by a U.S. retailer for digital communication with customers. There are $18{,}571$ unique customers, with data recorded for each subject at three time periods.  These communications were targeted based on various demographic and compliance-related factors. Customer demographics in our analysis include gender and age. 
In the data, there are $16{,}239$ female, $482$ male, and $1{,}850$ gender unlabeled customers. 
For our case study, we redefine gender as a binary variable that equals one for subjects identifying as female and zero otherwise.  Age is recorded in five categories: under 18 years, 18--35 years, 36--49 years, 50--64 years, and 65 years or older. We recode age as a binary variable, classifying customers aged 50 years or older as older customers ($=1$) and those younger than 50 years as younger customers ($=0$). The resulting sample comprises $15{,}477$ older customers and $3{,}094$ younger customers.

For each recipient, we consider the following time-varying covariates at times $t>1$: promotional email (value-added incentive versus price discount), email-open status in the earlier time period, and the number of days since the last email was opened and the number of days since the last purchase. The final outcome of interest is the subsequent number of items purchased, that is, the order count, after the third email is sent. The longitudinal variables, "days since the last email was opened" and "days since the last purchase," serve as intermediate variables that may mediate the causal relationship between receiving promotional emails and resulting purchases.

A challenge in analyzing the dataset is the zero inflation and right-skewed distributions of the longitudinal mediators and the final outcome. A significant number of subjects do not make a purchase, while many purchase again without waiting. Thus, their observed mediators and outcome values are zero. Figure~\ref{Fig1} in the Supplementary Material illustrates this pattern in the data, showing a spike at zero for the variables "Days since email opened" and "Order count." "Days since last purchase" also exhibits a significant spike at zero at each time $t = 1,2,3$. The semi-parametric hurdle models proposed in this work are designed to accommodate the excessive number of zeros. In our context, hurdle models can be viewed as piecewise regression models (formalized in Section \ref{sec:4}).

The goal is to compare the two types of email promotions that involve sending either a value incentive or a price incentive email to make a purchase. Unlike in a treatment-to-control comparison or a drug-to-placebo comparison, here we do not have a control or an untreated population. This distinction is relevant in our context, as we also know that the email has an impact only when it is opened. Not opening the email is akin to not complying with the treatment, but it does not put the individual in the complementary group. Thus, the results should be interpreted accordingly. 
Figure~\ref{fig:OpenerSubPopnAcrossTime} in the Supplementary Material shows the empirical distribution of customers who open the email when received in each individual period or over multiple periods. At best, roughly 41\% of customers ($7{,}650$ out of $18{,}571$) open the email sent in period 2, while at the worst, 19\% ($3{,}514$ out of $18{,}571$) open all three emails. In a separate figure in the results section, we contrast the email opening patterns across the two treatment groups. 


\subsection{Review of causal mediation analysis under treatment non-compliance}\label{sec:22} 

Causal mediation analysis focuses on developing statistical methods to evaluate the role of intermediate covariates, i.e., mediators, on the causal pathway between the treatment and the outcome. In longitudinal studies, mediation analysis is particularly challenging because treatments, mediators, and confounders evolve over time and may influence one another. Our case study further involves multiple mediators and treatment non-compliance. 

Statistical methods for mediation analysis with multiple mediators are relatively limited, especially in longitudinal settings. When multiple mediators are measured concurrently, their causal ordering is often unclear. A common approach is therefore to assess their joint mediating effect without imposing an ordering. \cite{vanderweele2014mediation} introduced joint mediator effects and \cite{park2018causal} incorporated treatment non-compliance into the setting. These developments focus on cross-sectional data, whereas we work with longitudinal data.

Treatment non-compliance introduces additional complications because treatment receipt is no longer independent of unobserved characteristics. Although instrumental-variable approaches based on intent-to-treat (ITT) effects and local average treatment effects (LATE) \citep{angrist1995identification,angrist1996identification} have been proposed for mediation analysis in treatment-to-control comparison studies \citep{yamamoto2013identification,cheng2023identification}, they rely on monotonicity and exclusion restriction assumptions that are often difficult to justify.

The intent-to-treat estimates are less relevant to our study, as the goal is to design an effective marketing strategy that customizes the email to increase sales and also improve customer engagement. We observe noncompliance, e.g., not opening the promotional email for both types of emails. Thus, a LATE estimate is also not appropriate for the study. Instead, we consider the causal estimands within the different principal strata of treatment compliance, i.e., principal causal effects (PCEs) \citep{frangakis1999addressing}. At each time, there are four possible compliance classes defining the principal strata: i) opening neither type of email, ii) opening only value-added incentive or iii) price incentive email, and iv) opening either type of email. The principal stratification of the customer is latent as it depends on a customer's counterfactual email opening behavior.

In a cross-sectional data setting, PCE quantifies the total effect of treatment assignment on the outcome within each principal stratum at the time of study entry. 
In extending these definitions of principal causal estimands to a longitudinal setting for our study, a difficulty is that the number of principal strata grows exponentially over time, with $4^{t}$ principal strata at time $t$. In Section \ref{sec:3.4} we modify the standard definition of PCE in our study by only considering the principal stratum at time 1. This is a reasonable compromise as principal strata encode the customer behavior in whether they open the different emails, which typically would not change drastically over time because of behavioral stickiness \citep{wood2009habitual}.

In longitudinal settings with treatment non-compliance, treatment receipt status (different from treatment assignment) acts as a post-treatment confounder of future treatments, mediators, and outcomes, creating a recanting-witness problem \citep{avin2005identifiability}. As a result, natural direct and indirect effects generally require untestable cross-world assumptions for identification. To avoid this issue, we focus on interventional effects \citep{vanderweele2017mediation,zheng2017longitudinal,wang2023targeted,bhandari2025bayesian,bhandari2025causal,vo2026recanting,domingo2026path}, which define mediation effects through randomized interventions on the mediator distribution. These interventional effects are defined by fixing the mediator for each subject not to the level it would have been for that subject under a particular value of the exposure, but instead to a level that is randomly chosen from the conditional distribution of the mediator given a specific value of exposure and other past covariates. We adopt these interventional effects in our study.

\section{Notation and Causal Framework} \label{sec:3}
\subsection{Notation} \label{sec:3.1}
We have $n$ subjects, where, for each subject $i$, $i\in 1,\ldots, n$, the collection of pre-treatment covariates $\textbf{L}_{i,0} = \big[L_{i,0}^{(1)}, L_{i,0}^{(2)}, \ldots, L_{i,0}^{(K)}\big]$ is recorded at baseline and the longitudinal data are recorded at $T$ time periods.  For subject $i$ at time $t$, for $t = 1,\ldots,T$, $Z_{it} \in \big\{0,1\big\}$ denotes a binary treatment assignment status and $D_{it} \in \big\{0,1\big\}$ denotes the actual treatment receipt. Similarly, $M_{it} = \big[M_{it}^{(1)}, M_{it}^{(2)}, \ldots, M_{it}^{(J)}\big]$ denotes the longitudinal mediator vector where each element $M_{it}^{(j)}$, for $j = {1, \ldots, J},$ represents a mediator measured after the treatment receipt. Finally, $Y_{i}$ denotes the outcome recorded at the end of the study. Figure \ref{Fig:DAG} shows a  directed acyclic graph representing the temporal relationships among these variables.

\begin{figure}[ht]
\centering
\resizebox{0.8\textwidth}{!}{%
\begin{tikzpicture}[
    every node/.style={circle, draw=black, thick, minimum size=1.4cm, font=\footnotesize, align=center},
    arrow/.style={->, thick, >=Stealth}
]

\node (L0) at (-4.5,0) {$L\textsubscript{i,0}$};
\node (Zt1) at (-1.5,0) {$Z_{i,t}$};
\node (Mt1) at (1.5,0) {$\big[M_{i,t}^{(1)}, \ldots, M_{i,t}^{(J)}\big]$};
\node (Zt2) at (4.5,0) {$Z_{i,t+1}$};
\node (Mt2) at (8,0) {$\big[M_{i,t+1}^{(1)}, \ldots, M_{i,t+1}^{(J)}\big]$};
\node (Yi) at (12.5,0) {$Y_i$};

\node (Dt1) at (2,3) {$D_{i,t}$};
\node (Dt2) at (7,3) {$D_{i,t+1}$};

\draw[arrow, bend left=30] (L0) to +(1.5,1);  
\draw[arrow, bend right=30] (L0) to +(+1.5,-1); 
\draw[arrow] (L0) -- +(1.5,0);                 

\draw[arrow] (Mt2) -- +(2.5,0);
\draw[arrow, bend left=-35] (Mt2) to (Yi);
\node[draw=none, fill=none, inner sep=0pt] at (10.75, 0) {$\cdots$};
\draw[arrow] (11, 0) -- +(Yi);
\node[draw=none, fill=none, inner sep=0pt] at (-2.5, 0) {$\cdots$};

\draw[arrow] (Mt1) -- (Zt2);
\draw[arrow, bend left=25] (Mt1) to (Dt2);
\draw[arrow, bend left=-35] (Mt1) to (Mt2);
\draw[arrow, bend left=-35] (Mt1) to (Yi);

\draw[arrow] (Dt1) -- (Mt1);
\draw[arrow] (Dt1) -- (Zt2);
\draw[arrow] (Dt1) -- (Mt2);
\draw[arrow, bend left=45] (Dt1) to (Yi);
\draw[arrow] (Dt1) -- (Dt2);

\draw[arrow] (Dt2) -- (Mt2);
\draw[arrow] (Dt2) -- (Yi);

\draw[arrow] (Zt1) -- (Mt1);
\draw[arrow, color = blue, bend left=25] (Zt1) to (Dt1);
\draw[arrow, color = blue, bend left=60] (Zt1) to (Dt2);
\draw[arrow, bend left=-50] (Zt1) to (Zt2);
\draw[arrow, bend left=-45] (Zt1) to (Mt2);
\draw[arrow, bend left=-45] (Zt1) to (Yi);

\draw[arrow] (Zt2) -- (Mt2);
\draw[arrow, color = blue, bend left=25] (Zt2) to (Dt2);
\draw[arrow, bend left=-45] (Zt2) to (Yi);
\end{tikzpicture}}
\caption{A directed acyclic graph representing the temporal relationships among baseline, longitudinal, and outcome variables. The blue curved arrows illustrate the recanting witness issue, where longitudinal treatment receipt indicators, acting as confounders of the relationships between future variables, are influenced by past treatment assignment variables.}
\label{Fig:DAG}
\end{figure}
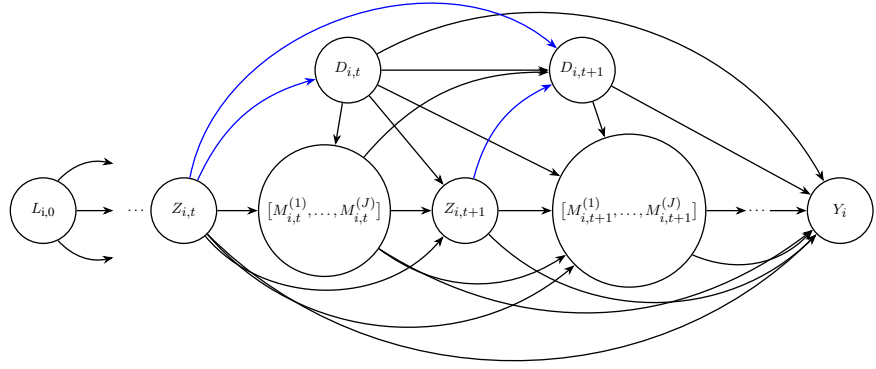

 The observed data is denoted by the set $\mathcal{O} = \big\{\mathcal{O}_{i}\big\}_{i=1}^{n} \equiv \big\{\textbf{L}_{i,0}, Z_{it}, D_{it},  M_{it}, Y_{i}  \big\}_{i=1}^{n}$. We assume that  $\mathcal{O}_{i}$'s are IID realizations of the underlying population distribution. For longitudinal variables, we use the bold letters to denote the vector of random variables representing both the history and the current observation, e.g., $\textbf{Z}_{it}\equiv \big\{Z_{i1}, \ldots, Z_{it}\big\}$. For notational clarity and consistency with our
empirical study, we set $T =3$ throughout the article. However, the proposed methods are valid for any finite time length $T$, and the extension is straightforward.

 In our case study, $Z_{it} = 1$ (or 0) indicates that individual $i$ received an email offering a value-added incentive (or a price discount incentive) at time $t$. This treatment definition marks an important departure from causal inference in a treatment-to-control study. In contrast, our analysis compares alternative intervention strategies, each representing a distinct treatment regime. Consequently, the causal contrasts of interest are \emph{intervention-to-intervention} comparisons. This distinction is crucial for interpreting the resulting causal estimates and their policy implications. Finally, at each time $t$, we define $D_{it} = 1$ if individual $i$ opens the promotional email and $D_{it} = 0$ otherwise, and $Y_i$ denotes customer $i$'s total number of purchases at the study period $T$.

\subsection{Potential outcomes}\label{sec:3.2}

Our potential outcomes framework is established based on joint intervention on the treatment assignment history $\textbf{Z}_{i,T}$, treatment receipt history $\textbf{D}_{i,T}$, and the mediator history $\textbf{M}_{i,T} = \big(\textbf{M}_{i,T}^{(1)},\ldots,\textbf{M}_{i,T}^{(J)}\big)$ with $M_{i,t}^{(j)}$ denoting the $j$th mediator at time $t$. Throughout, we use the shorthand $\textbf{m}_{i,T} \equiv \big(\textbf{m}_{i,T}^{(1)},\ldots,\textbf{m}_{i,T}^{(J)}\big)$ for a realized mediator-history vector.  Let $Y_{i}\Bigl(\textbf{z}_{i,T},\textbf{d}_{i,T},\textbf{m}_{i,T}\Bigr) \equiv Y_{i}\Bigl(\textbf{z}_{i,T},\textbf{d}_{i,T},\textbf{m}_{i,T}^{(1)}, \ldots, \textbf{m}_{i,T}^{(J)}\Bigr)$ denote the value of the cross-sectional outcome $Y_{i}$ had the treatment assignment history $\textbf{Z}_{i,T}$, the treatment receipt history $\textbf{D}_{i,T}$, and the mediator history $\textbf{M}_{i,T}$ been set---possibly, contrary to the fact---to $\textbf{z}_{i,T}$, $\textbf{d}_{i,T}$ and $\textbf{m}_{i,T}$, respectively.

 Suppose that $\textbf{z}_{t}$ and $\textbf{z}_{*,t}$ are two possible exposure regimes. At each time $t$, these exposure regimes define potential values of treatment receipt statuses, $D_{i,t}\big(\textbf{z}_{t}\big)$ and $D_{i,t}\big(\textbf{z}_{*,t}\big)$, respectively. Now, for the $J$ mediators at time $t \in \{1,\ldots,T \}$, let: 
        \begin{equation}\label{Eq::joint_mediator_distn}
            \resizebox{\textwidth}{!}{$\displaystyle
            \begin{aligned}
               \mathcal{G}^{\textbf{z}_{*,t},\textbf{d}_{t}(z_{*,t})}_{(m_{t}^{(1)}, \ldots, m_{t}^{(J)} |\textbf{d}_t, \textbf{m}_{t-1}^{(1)},\ldots, \textbf{m}_{t-1}^{(J)}, \boldsymbol{\ell}_{0} )}  
                & =  f\biggl(\Bigl(M_{t}^{(1)}\bigl(\textbf{z}_{*,t}, \textbf{d}_t(z_{*,t})\bigr),\ldots,M_{t}^{(J)}\bigl(\textbf{z}_{*,t}, \textbf{d}_t(z_{*,t}) \bigr)\Bigr) = \bigl(m_{t}^{(1)},\ldots, m_{t}^{(J)}\bigr)
                \big| \textbf{D}_{t}(z_{*,t}) =  \textbf{d}_t(z_{*,t}),
                \\&
                \Bigl(\textbf{M}_{t-1}^{(1)}\bigl(\textbf{z}_{*,t-1}, \textbf{d}_{t-1}(z_{*,t-1})\bigr), \ldots, \textbf{M}_{t-1}^{(J)}\bigl(\textbf{z}_{*,t-1}, \textbf{d}_{t-1}(z_{*,t-1})\bigr)\Bigr) = \bigl( \textbf{m}_{t-1}^{(1)}, \ldots,\textbf{m}_{t-1}^{(J)}\bigr),L_0 = \ell_0\biggr) 
            \end{aligned}
            $} 
        \end{equation}
        denote the joint conditional density of the counterfactual mediators, $M_{t}^{(1)}\bigl(\textbf{z}_{*,t}, \textbf{d}_t(z_{*,t})\bigr), \ldots,$ $ M_{t}^{(J)}\bigl(\textbf{z}_{*,t}, \textbf{d}_t(z_{*,t}) \bigr)$, in the world where the treatment assignment is set to $\textbf{Z}_{t} = \textbf{z}_{*,t}$ and the treatment receipt is set to  $\textbf{D}_{t}(z_{*,t}) = \textbf{d}_{t}(z_{*,t}) \coloneqq \big\{ d_{s}(z_{*,s})\big\}_{s = 1}^{t}$. This conditional density provides a random draw $\big(\text{say, } \bigl(m_{t}^{(1)}, \ldots, m_{t}^{(J)}\bigr)\big)$ of $\bigl(M_{t}^{(1)}, \ldots, M_{t}^{(J)}\bigr) $ within each stratum $(\textbf{d}_t,\textbf{m}_{t-1}, \boldsymbol{\ell}_{0})$ at each time $t$. We compress the notation to denote $\mathcal{G}^{\textbf{z}_{*,t},\textbf{d}_{t}(z_{*,t})}_{m_{t}} = \mathcal{G}^{\textbf{z}_{*,t},\textbf{d}_{t}(z_{*,t})}_{(m_{t}^{(1)}, \ldots, m_{t}^{(J)} |\textbf{d}_t, \textbf{m}_{t-1}^{(1)},\ldots, \textbf{m}_{t-1}^{(J)}, \boldsymbol{\ell}_{0} )}$ and $\mathbfcal{G}^{\textbf{z}_{*,t},\textbf{d}_{t}(z_{*,t})}_{m_{t}}= \Bigl\{ \mathcal{G}^{z_{*,1},d_{1}(z_{*,1})}_{m_{1}}, \ldots , \mathcal{G}^{\textbf{z}_{*,t},\textbf{d}_{t}(z_{*,t})}_{m_{t}}\Bigr\}$ for convenience. Finally, consider a joint intervention to statically set $\textbf{Z}_{T} = \textbf{z}_{T}$ and $\textbf{D}_{t} = \textbf{d}_{T}(z_T)$ in the population and randomly draw an observation $ \big(\textbf{m}_{T}^{(1)}, \ldots,\textbf{m}_{T}^{(J)}\big)$ of  $\big(\textbf{M}_{T}^{(1)}, \ldots,\textbf{M}_{T}^{(J)}\big) \sim \mathbfcal{G}^{\textbf{z}_{*,T},\textbf{d}_{T}(z_{*,T})}_{m_{T}}$. We define the potential outcome resulting from this intervention as $Y\Bigl(\textbf{z}_{T}, \textbf{d}_{T}(z_T),\mathbfcal{G}^{\textbf{z}_{*,T},\textbf{d}_{T}(z_{*,T})}_{m_{T}}\Bigr).$
        
    Note that we use  $Y\Bigl(\textbf{z}_{T}, \textbf{d}_{T}(z_T),\mathbfcal{G}^{\textbf{z}_{*,T},\textbf{d}_{T}(z_{*,T})}_{m_{T}}\Bigr)$ instead of $Y\Bigl(\textbf{z}_{T}, \textbf{d}_{T}(z_T),\textbf{m}_{T}^{(1)}, \ldots,\textbf{m}_{T}^{(J)}\Bigr)$ to denote the potential outcome resulting from the joint intervention. This is done, with a slight abuse of notation, to highlight that the random interventions in our setting define causal parameters by enforcing mediator distributions rather than deterministically enforcing particular mediator values. Although $\mathbfcal{G}^{\textbf{z}_{*,t},\textbf{d}_{t}(z_{*,t})}_{m_{t}}$ represents a conditional density, we sometimes use it interchangeably with the potential mediator $\textbf{M}_{t}\bigl(\textbf{z}_{*,t}, \textbf{d}_t(z_{*,t})\bigr)$ in this article. The specific usage should be discerned based on the context. For instance, in the equality  $\mathbfcal{G}^{\textbf{z}_{*,t},\textbf{d}_{t}(z_{*,t})}_{m_{t}} = \textbf{m}_t$, the observation $\textbf{m}_t$ should be understood as a realized value of the potential mediator $\textbf{M}_{t}\bigl(\textbf{z}_{*,t}, \textbf{d}_t(z_{*,t})\bigr)$, drawn from $\mathbfcal{G}^{\textbf{z}_{*,t},\textbf{d}_{t}(z_{*,t})}_{m_{t}}$.  

In our case study, the randomized intervention framework is particularly appealing because the two mediators---time since the last email was opened and time since the last purchase---are driven by complex customer behaviors that cannot be deterministically controlled by the retailer. 
Randomized interventions draw mediator values from their conditional distributions under a specified email regime, preserving the natural heterogeneity in customer engagement patterns. 
The interventional direct effect measures the expected change in purchases when switching from price-discount to value-added emails while holding the mediator distribution fixed at that observed under price-discount emails, thereby isolating the effect of email content itself. In parallel, the interventional indirect effect quantifies the portion of the effect that operates through changes in customer engagement and purchasing timing induced by the email type. Together, these quantities inform marketing strategy: the direct effect addresses the impact of email content, while the indirect effect reveals whether the email type influences purchases by altering customer engagement dynamics.

\subsection{Principal stratification framework}\label{sec:3.3}    

Each subject has two potential treatment receipt values at $t =1$, $D_1(0)$ and $D_1(1)$, corresponding to the receipt status under each possible treatment assignment. The observed treatment receipt status, therefore, satisfies $D_1=\mathbb{I}\big\{Z_1=0\big\}D_1(0)+\mathbb{I}\big\{Z_1=1\big\}D_1(1)$.
We define the  principal stratum at $t =1$ as $U_1=\bigl(D_1(1),D_1(0)\bigr)$, yielding four latent compliance classes: $(1,1)$, $(1,0)$, $(0,1)$, and $(0,0)$. We refer to these strata as active, non-price value-attentive, price-attentive, and non-active customers, respectively. Active and non-active customers always open or always ignore the first promotional email regardless of promotion type, whereas price-attentive customers open the email only when it contains a price discount, and non-price value-attentive customers open it only when it contains a value-added offer (such as free delivery on their next purchase).

\subsection{Causal estimands}\label{sec:3.4}    
Suppose $\mathbf{z}_{T}$ and $\mathbf{z}_{*,T}$ represent the generic treatment regime vectors across $T$ periods. For example, when $T=3$, $\mathbf{z}_{T} = \{1,1,1\}$ and $\mathbf{z}_{*,T} = \{0,0,0\}$ denote two contrasting treatment regimes. We define the principal causal effect (PCE) as the conditional expected difference given principal strata only at the beginning of the study, $U_1$, that is: 
        \begin{equation}
        \resizebox{0.9\textwidth}{!}{$\displaystyle
            \begin{aligned}
                 \text{PCE}_{(d_{1}(1), d_{1}(0) )} 
                &=  \mathbb{E}\Bigl[Y\Bigl(\textbf{z}_{T}, \textbf{d}_{T}\big(\textbf{z}_{T}\big),\mathbfcal{G}^{\textbf{z}_{T},\textbf{d}_{T}(\textbf{z}_{T})}_{m_{T}}\Bigr) 
                - Y\Bigl(\textbf{z}_{*,T}, \textbf{d}_{T}\big(\textbf{z}_{*,T}\big),\mathbfcal{G}^{\textbf{z}_{*,T},\textbf{d}_{T}(\textbf{z}_{*,T})}_{m_{T}}\Bigr) 
                \,\big|\, U_1 = \bigl(d_{1}(1), d_{1}(0)\bigr) \Bigr].
            \end{aligned}
            $} 
        \end{equation}
Here, $(d_{1}(1), d_{1}(0))$, with a slight abuse of notation, denotes the realization of $(D_{1}(1), D_{1}(0))$ at time $t=1$. In other words, the pair $(d_{1}(1), d_{1}(0))$ represents a realization of the joint potential treatment receipt statuses under contrasting treatment assignments at the beginning of the study. The PCE can be decomposed in terms of principal interventional direct effect $\bigl(\text{PIDE}_{(d_{1}(1), d_{1}(0) )}\bigr)$ and principal joint interventional indirect effect $\bigl(\text{PJIIE}_{(d_{1}(1), d_{1}(0) )}\bigr)$. The former quantifies the direct effect of the intervention on the outcome, while the latter quantifies the effect traversing through the mediators. This decomposition can be expressed as: 
          \begin{align}
            \small
             \text{PCE}_{(d_{1}(1), d_{1}(0) )}
             &= \mathbb{E}\Bigl[Y\Bigl(\textbf{z}_{T}, \textbf{d}_{T}\big(\textbf{z}_{T}\big),\mathbfcal{G}^{\textbf{z}_{T},\textbf{d}_{T}(\textbf{z}_{T})}_{m_{T}}\Bigr)
            - Y\Bigl(\textbf{z}_{*,T}, \textbf{d}_{T}\big(\textbf{z}_{*,T}\big),\mathbfcal{G}^{\textbf{z}_{*,T},\textbf{d}_{T}(\textbf{z}_{*,T})}_{m_{T}}\Bigr)\,\big|\, U_1 = \bigl(d_{1}(1), d_{1}(0)\bigr) \Bigr]
            \nonumber\\
            & = 
            \mathbb{E}\Bigl[Y\Bigl(\textbf{z}_{T}, \textbf{d}_{T}\big(\textbf{z}_{T}\big),\mathbfcal{G}^{\textbf{z}_{*,T},\textbf{d}_{T}(\textbf{z}_{*,T})}_{m_{T}}\Bigr)
            - Y\Bigl(\textbf{z}_{*,T}, \textbf{d}_{T}\big(\textbf{z}_{*,T}\big),\mathbfcal{G}^{\textbf{z}_{*,T},\textbf{d}_{T}(\textbf{z}_{*,T})}_{m_{T}}\Bigr)\,\big|\, U_1 = \bigl(d_{1}(1), d_{1}(0)\bigr) \Bigr]
            \nonumber\\  & +
            \mathbb{E}\Bigl[Y\Bigl(\textbf{z}_{T}, \textbf{d}_{T}\big(\textbf{z}_{T}\big),\mathbfcal{G}^{\textbf{z}_{T},\textbf{d}_{T}(\textbf{z}_{T})}_{m_{T}}\Bigr)
            - Y\Bigl(\textbf{z}_{T}, \textbf{d}_{T}\big(\textbf{z}_{T}\big),\mathbfcal{G}^{\textbf{z}_{*,T},\textbf{d}_{T}(\textbf{z}_{*,T})}_{m_{T}}\Bigr)\,\big|\, U_1 = \bigl(d_{1}(1), d_{1}(0)\bigr) \Bigr]
            \nonumber\\
            &=  \text{PIDE}_{(d_{1}(1), d_{1}(0) )} + \text{PJIIE}_{(d_{1}(1), d_{1}(0) )} .\label{Eq::PCEdecompose}
            \end{align}
            
Define $\theta(\textbf{z}_{T},\textbf{z}_{*,T}) \coloneqq \mathbb{E}\Bigl[Y\Bigl(\textbf{z}_{T}, \textbf{d}_{T}(\textbf{z}_T),\mathbfcal{G}^{\textbf{z}_{*,T},\textbf{d}_{T}(\textbf{z}_{*,T})}_{m_{T}}\Bigr)\,\big|\, U_1 = \bigl(d_{1}(1), d_{1}(0)\bigr) \Bigr]$ for the exposure regimes $\textbf{z}_{T}$ and $\textbf{z}_{*,T}$.  In our case study, $\theta(\textbf{z}_{T},\textbf{z}_{*,T})$ represents the expected purchase count, within the stratum characterized by $(d_1(1), d_1(0))$, when customers receive promotional emails according to regime $\textbf{z}_{T}$, but their intermediate behaviors, i.e., the time since the last email was opened and the time since the last purchase, follow the distribution that would arise under the alternative regime $\textbf{z}_{*,T}$. Analogously, $\theta(\mathbf{z}_{T}, \mathbf{z}_{T})$ and $\theta(\mathbf{z}_{*,T}, \mathbf{z}_{*,T})$ denote the expected purchase counts generated under $\mathbf{z}_{T}$ and $\mathbf{z}_{*,T}$, respectively.


We can write the PCE and its decomposition effects in (\ref{Eq::PCEdecompose}) as functions of $\theta(\textbf{z}_{T},\textbf{z}_{*,T})$ as $\text{PCE}_{\bigl(d_1(1),d_1(0)\bigr)} = \theta(\textbf{z}_{T},\textbf{z}_{T}) - \theta(\textbf{z}_{*,T},\textbf{z}_{*,T})$, $\text{PIDE}_{\bigl(d_1(1),d_1(0)\bigr)} = \theta(\textbf{z}_{T},\textbf{z}_{*,T}) - \theta(\textbf{z}_{*,T},\textbf{z}_{*,T})$, and $\text{PJIIE}_{\bigl(d_1(1),d_1(0)\bigr)} =  \theta(\textbf{z}_{T},\textbf{z}_{T}) - \theta(\textbf{z}_{T},\textbf{z}_{*,T})$. For illustration, consider the treatment regimes $\textbf{z}_{T} = \{1,1,1\}$ (indicating value-added incentive emails throughout the study) and $\textbf{z}_{*,T} = \{0,0,0\}$ (indicating price discount incentive emails throughout the study). Then, PCE measures the overall effect of sending value-added incentive emails on purchase count relative to sending price discount incentive emails. PIDE captures the direct effect of switching from price discount emails to value-added incentive emails, while PJIIE reflects the effect of this switch operating through changes in the time since the last purchase and the time since the last email was opened. 

\subsection{Identification assumptions}\label{sec:3.5} 
We assume consistency of observed data with the corresponding potential variables to rule out the possibility of multiple versions of the interventional variables \citep{imai2010general,10.1214/10-STS321}. This assumption generally holds when the intervention variables are correctly defined. Thus, our first assumption is as follows.

\begin{assumption}\label{assump::1}
    \textbf{Consistency}: (i) $D_{t}\bigl(\textbf{z}_{t}^{'}\bigr) = D_{t}$ given $\textbf{Z}_{t} = \textbf{z}_{t}^{'}$ for $\textbf{z}_{t}^{'} \in \{\textbf{z}_{t},\textbf{z}_{*,t} \}$, (ii) $M_{t}\bigl(\textbf{z}_{t}^{'},  \textbf{d}_{t}\bigr) = M_{t}$ given $\textbf{Z}_{t} = \textbf{z}_{t}^{'}$, and $\textbf{D}_{t} = \textbf{d}_{t}$ for $\textbf{z}_{t}^{'} \in \{\textbf{z}_{t},\textbf{z}_{*,t} \}$, and (iii) $Y(\textbf{z}_{T},\textbf{d}_{T},  \textbf{m}_{T}) = Y$ given $\textbf{Z}_{T} = \textbf{z}_{T}$, $\textbf{D}_{T} = \textbf{d}_{T}$, and $\textbf{M}_{T} = \textbf{m}_{T}$. 
\end{assumption}
\begin{assumption}\label{assump::2}
    \textbf{Ignorability of treatment assignment}:     There are no unmeasured confounders of the relationship between the current treatment assignment status and all current and future counterfactuals, conditional on the observed covariate history. In other words, for any $\textbf{z}^{'}_{t} \in \{\textbf{z}_{t}, \textbf{z}_{*,t} \}$
    \begin{enumerate}
        \item $ D_{w}(z^{'}_{w})  \indep Z_{s} \;\big|\; \textbf{Z}_{s-1},\textbf{D}_{s-1}, \textbf{M}_{s-1}, \textbf{L}_0 $ for all $w \geq s, s \in \{2, \ldots, T\}$, 
        \item $ M_{w}(\textbf{z}^{'}_{w}, D_{w}(z^{'}_{w}))  \indep Z_{t} \;\big|\; \textbf{Z}_{t-1},\textbf{D}_{t-1}, \textbf{M}_{t-1}, \textbf{L}_0 $ for all $w \geq t$, and  
        \item $ Y\bigl(\textbf{z}_{T},\textbf{d}_{T},  \textbf{m}_{T}\bigr)  \indep Z_{t} \;\big|\; \textbf{Z}_{t-1},\textbf{D}_{t-1}, \textbf{M}_{t-1}, \textbf{L}_0 $ for all $t \leq T$.
    \end{enumerate}   
\end{assumption}

Assumption \ref{assump::2} requires that, given a subject's observed covariate history, whether the email they receive at a given time contains a promotional price discount is {\it conditionally independent} of (1) their current and future potential email open status, (2) current and future potential time since the email was opened and since the last purchase, and (3) their future potential purchases resulting from the promotional email. In our case study, the marketing manager makes the decision on the promotional email type to be sent based on the customer's recorded history. Thus, the ignorability of the treatment assignment assumption is fairly likely to hold in our case study.

\begin{assumption}\label{assump::3}
\textbf{Ignorability of the cross-world treatment receipt:} 
The counterfactual treatment receipt status at $t=1$ is independent of the future counterfactual treatment receipt status, mediator, and outcome values in the worlds where the exposure assignments for the former and the latter differ. I.e., for all $z_{1},z_{*,1} \in \{0,1\}$  and for all $w \geq 1, s \geq 2$

    \begin{enumerate}
        \item $D_{s}(z_{s}) \indep  D_1(1-z_{1}) \;\big|\;  D_1(z_{1}),  \textbf{D}_{s-1}(\textbf{z}_{s-1}), \textbf{L}_0$, 
        \item $\mathcal{G}^{\textbf{z}_{*,w},\textbf{d}_{w}(z_{*,w})}_{m_{w}}   \indep   D_1(1-z_{*,1})  \;\big|\;  D_1(z_{*,1}),  \textbf{D}_{2:w}(\textbf{z}_{*,2:w}), \mathbfcal{G}^{\textbf{z}_{*,w-1},\textbf{d}_{w-1}(z_{*,w-1})}_{m_{w-1}},  \textbf{L}_0$, and 
        \item $Y(\textbf{z}_{T},\textbf{d}_{T}, \textbf{m}_{T})  \indep  D_1(1-z_{1})  \;\big|\;  D_1(z_{1}),\textbf{D}_{2:T}(\textbf{z}_{2:T}),\mathbfcal{G}^{\textbf{z}_{*,T},\textbf{d}_{T}(z_{*,T})}_{m_{T}}, \textbf{L}_0$.
    \end{enumerate}
\end{assumption}
 
 In our case study, Assumption \ref{assump::3} requires that a subject's (1) future email open status, (2) future time (in days) since the last email was opened, time since the last purchase, and (3) future order count resulting from receiving an email with a promotional price discount are {\it conditionally independent} of the subject's initial email open status due to receiving an email with a non-price offer. This assumption cannot be verified empirically. However, it is plausible in our case study because the incentive is revealed only after an email is opened. Thus, the initial decision to open an email is driven primarily by habitual engagement behaviors, such as sender familiarity and subject-line salience, rather than the incentive itself \citep{wood2009habitual}. Conditioning on $D_1(z_1)$ and baseline covariates therefore captures much of the stable propensity to engage, leaving only residual dependence due to unmeasured persistent factors.
 
Assumption~\ref{assump::3} is related to the more familiar \textit{principal ignorability assumption} \citep{jo2009use}, which is commonly invoked for nonparametric identification of principal causal effects under treatment noncompliance, and here we adapt to the mediation analysis setting with longitudinal data. In the cross-sectional case without mediators, the principal ignorability assumption states that principal stratum membership is conditionally independent of the potential outcomes given observed covariates. Section~\ref{sec:S3} of the Supplementary Material provides a sensitivity analysis framework to assess the robustness of the results to potential violations of Assumption \ref{assump::3}.3.

Section~\ref{sec:22} above notes that defining random interventional analogues of natural direct and indirect effects avoids reliance on untestable cross-world independence assumptions for identification. However, Assumption~\ref{assump::3} requires conditional independence across two parallel worlds of potential treatment receipt statuses. This apparent discrepancy warrants clarification. While replacing natural effects with interventional effects removes the need for cross-world independence assumptions in identifying mediation parameters under perfect treatment compliance, our setting is more complex. Specifically, our causal estimands are conditioned on principal stratum membership at $t=1$ to account for treatment non-compliance. This conditioning leads to a cross-world independence specific to non-compliance estimands such that Assumption~\ref{assump::3} becomes key for identifying principal causal estimands in the longitudinal non-compliance framework for mediation analysis. 

\begin{assumption}\label{assump::4}
  \textbf{Ignorability of the mediator}:
  There are no unmeasured confounders of the relationship between the current mediators and all current and future potential treatment receipts and the potential outcomes, conditional on the observed covariate history.
  \begin{enumerate}
        \item $D_{s}(\textbf{z}^{'}_s) \indep M_{t} \;\big|\;  \textbf{Z}_{s}, \textbf{D}_{s-1}, \textbf{M}_{s-1}, \textbf{L}_0  $ for all $ s \geq 2, \textbf{z}^{'}_s \in \{\textbf{z}_s,\textbf{z}_{*,s} \}$. 
        \item $Y\bigl(\textbf{z}_{T},\textbf{d}_{T},  \textbf{m}_{T}\bigr) \indep M_{t} \;\big|\; \textbf{Z}_{t},\textbf{D}_{t}, \textbf{M}_{t-1}, \textbf{L}_0$ for all $ t \leq T$.
    \end{enumerate}
\end{assumption}
Assumption \ref{assump::4} requires that, given a subject's observed covariate history, their current time since the last email was opened and the last purchase are {\it conditionally independent} of  1) their future potential email open status and 2) subsequent potential number of purchases.

We also assume the positivity assumptions in the email sent, email opened, and mediator values on the support of the observed history. 
Let $f_{U\mid V}(u\mid v)$ be the conditional probability (resp., density function) of discrete (resp., continuous) $U=u$ given $V=v$. Formally:

\begin{assumption}\label{assump::5}
    {\bf Positivity:}
    For all $d_t, z_t$ and $m_t$ and history $\textbf{z}_{t-1}, \textbf{d}_{t-1},\textbf{m}_{t-1},\ell_0$ we have
    \begin{align*}
        f_{Z_t|\textbf{Z}_{t-1}, \textbf{D}_{t-1},\textbf{M}_{t-1}, L_0}(z_t|\textbf{z}_{t-1}, \textbf{d}_{t-1},\textbf{m}_{t-1},\ell_0) &> 0, \\
        f_{D_t|\textbf{Z}_{t}, \textbf{D}_{t-1},\textbf{M}_{t-1}, L_0}(d_t|\textbf{z}_{t}, \textbf{d}_{t-1},\textbf{m}_{t-1},\ell_0) &> 0, \quad\text{and} \\
        f_{M_t|\textbf{Z}_{t}, \textbf{D}_{t},\textbf{M}_{t-1}, \ell_0}(m_t|\textbf{z}_{t}, \textbf{d}_{t},\textbf{m}_{t-1},\ell_0) &> 0 .
    \end{align*}
\end{assumption}

\subsection{Non-parametric identification of $\theta(\textbf{z},\textbf{z}_{*})$}\label{sec:3.6}
\begin{prop}\label{prop1}
    Suppose that Assumptions 1--5 hold. For any two treatment regimes $\textbf{z}$, $\textbf{z}_{*}$, $\theta(\textbf{z},\textbf{z}_{*}) $ is identified using the following G-formula: 
\begin{equation}\label{Eq:PropGformula_NoRE}
\resizebox{\textwidth}{!}{$\displaystyle
        \begin{aligned}
            \theta(\textbf{z},&\textbf{z}_{*})
            =    \int_{\ell_{i0}} \int_{\textbf{d}_{i,2:T}} \int_{\textbf{m}^{(1)}_{i,T},\ldots, \textbf{m}^{(J)}_{i,T}}  \mathbb{E}\Bigl[ Y\big|\textbf{Z}_{i,T} = \textbf{z}_{i,T}, \textbf{D}_{i,T} = \textbf{d}_{i,T}, \textbf{M}^{(1)}_{i,T} = \textbf{m}^{(1)}_{i,T}, \ldots, \textbf{M}^{(J)}_{i,T} =  \textbf{m}^{(J)}_{i,T}, L_{i0} = \ell_{i0}\Bigr] \times
            \\&
            \Bigl\{ \prod_{t=1}^{T}  f\bigl( M_{i,t}^{(1)} = m_{i,t}^{(1)},\ldots, M_{i,t}^{(J)} = m_{i,t}^{(J)}\big|\textbf{Z}_{i,t} = \textbf{z}_{*,i,t}, \textbf{D}_{i,t} = \textbf{d}_{i,t}, \textbf{M}^{(1)}_{i,t-1} =  \textbf{m}^{(1)}_{i,t-1}, \ldots, \textbf{M}^{(J)}_{i,t-1} =\textbf{m}^{(J)}_{i,t-1}, L_{i0} =\ell_{i0}
             \bigr)
            \Bigr\} \times \\&
            \Bigl\{ \prod_{s=2}^{T} f\Bigl(d_{i,s}\big|  \textbf{Z}_{i,s} = \textbf{z}_{i,s}, \textbf{D}_{i,s-1} = \textbf{d}_{i,s-1},\textbf{M}^{(1)}_{i,s-1} =  \textbf{m}^{(1)}_{i,s-1}, \ldots, \textbf{M}^{(J)}_{i,s-1} =\textbf{m}^{(J)}_{i,s-1},  L_{i0} = \ell_{i0}\Bigr)\Bigr\} \times 
            f_{L_{i0}}(\ell_{i0})
             d(\textbf{m}^{(1)}_{i,T},\ldots, \textbf{m}^{(J)}_{i,T}) d(\textbf{d}_{i,2:T}) d\ell_{i0}.
        \end{aligned}
$}
\end{equation}
\end{prop}

The proof is given in Section~\ref{sec:S2} of the Supplementary Material.

\section{Model specification and estimation}\label{sec:4}
\subsection{Joint modeling with enriched Dirichlet process mixture (EDPM) models}\label{sec:4.1}  
Throughout the rest of the article, we fix the number of mediators at $J = 2$ as in our case study. However, the specification below is valid and can be easily extended for any finite $J > 2$. To estimate the causal parameters, we estimate the joint distribution $F\big(Y, \textbf{M}^{(2)}_{T}, \textbf{M}^{(1)}_{T}, \textbf{D}_{T}, \textbf{Z}_{T}, \textbf{L}_{0}  \big)$. We propose modeling the joint distribution 
using an EDPM \citep{wade2011enriched, wade2014improving} as below:
\begin{equation}\label{eq:EDPMmodel}
\resizebox{\textwidth}{!}{$\displaystyle
    \begin{aligned}
       Y_i \big| \textbf{M}^{(2)}_{i,T}, \textbf{M}^{(1)}_{i,T}, \textbf{D}_{i,T}, \textbf{Z}_{i,T}, \textbf{L}_{i,0}; \boldsymbol{\beta}_i & \sim F_y\big(. \big| \textbf{m}^{(2)}_{i,T}, \textbf{m}^{(1)}_{i,T}, \textbf{d}_{i,T}, \textbf{z}_{i,T}, \boldsymbol{\ell}_{i,0}; \boldsymbol{\beta}_i \big)
       \\
        \big(M^{(2)}_{i,t}, M^{(1)}_{i,t}\big)\big|\textbf{M}^{(2)}_{i,t-1}, \textbf{M}^{(1)}_{i,t-1}, \textbf{D}_{i,t}, \textbf{Z}_{i,t}, \textbf{L}_{i,0},b_i^{M^{(2)}}, b_i^{M^{(1)}}; \boldsymbol{\theta}^{M^{(2)}}_{i,t}, \boldsymbol{\theta}^{M^{(1)}}_{i,t}
        & \sim F_{(m^{(2)}_{t}, m^{(1)}_{t})} \big(.\big|\textbf{m}^{(2)}_{i,t-1}, \textbf{m}^{(1)}_{i,t-1}, \textbf{d}_{i,t}, \textbf{z}_{i,t}, \boldsymbol{\ell}_{i,0}, b_i^{M^{(2)}}, b_i^{M^{(1)}}; \boldsymbol{\theta}^{M^{(2)}}_{i,t}, \boldsymbol{\theta}^{M^{(1)}}_{i,t} \big)
        \\
        D_{i,t}\big|\textbf{M}^{(2)}_{i,t-1}, \textbf{M}^{(1)}_{i,t-1}, \textbf{D}_{i,t-1}, \textbf{Z}_{i,t}, \textbf{L}_{i,0}, b_i^{D} ; \boldsymbol{\theta}^{D}_{i,t} &\sim F_{d_{t}}\big(. \big|\textbf{m}^{(2)}_{i,t-1}, \textbf{m}^{(1)}_{i,t-1}, \textbf{d}_{i,t-1}, \textbf{z}_{i,t}, \boldsymbol{\ell}_{i,0}, b_i^{D}; \boldsymbol{\theta}^{D}_{i,t} \big)
        \\
        Z_{i,t}\big|\textbf{M}^{(2)}_{i,t-1}, \textbf{M}^{(1)}_{i,t-1}, \textbf{D}_{i,t-1}, \textbf{Z}_{i,t-1}, \textbf{L}_{i,0}, b_i^{Z}; \boldsymbol{\theta}^{Z}_{i,t} &\sim F_{z_{t}}\big(. \big|\textbf{m}^{(2)}_{i,t-1}, \textbf{m}^{(1)}_{i,t-1}, \textbf{d}_{i,t-1}, \textbf{z}_{i,t-1}, \boldsymbol{\ell}_{i,0}, b_i^{Z}; \boldsymbol{\theta}^{Z}_{i,t} \big)
        \\
        \textbf{L}_{i,0}^{(k)};\boldsymbol{\theta}^{L_0^{(k)}}_{i} &\sim F_{\ell_{0}^{(k)}}\big(.\big| \boldsymbol{\theta}^{L_0^{(k)}}_{i} \big), \quad k = 1, \ldots, K \\
    \end{aligned}
$}
\end{equation}

Here, $t = 1, \ldots, T$, and $\boldsymbol{\theta}_{i} = \big(\boldsymbol{\theta}^{M^{(1)}}_{i,t}, \boldsymbol{\theta}^{M^{(2)}}_{i,t}, \boldsymbol{\theta}^D_{i,t}, \boldsymbol{\theta}^Z_{i,t},  \big\{\boldsymbol{\theta}^{L_0^{(k)}}_{i}\big\}_{k=1}^{K}\big)$. We draw $(\boldsymbol{\beta}_i, \boldsymbol{\theta}_{i})$ from $H$ where $
         H$ is an EDP with parameters $\big(\alpha^{\beta}, \boldsymbol{\alpha}^{\theta\mid\beta}, H_{0\beta}, H_{0\theta|\beta}\big)$.

 The EDPM model utilizes a nested two-level clustering structure that is well-suited for our application involving a large heterogeneous customer base. Its hierarchical structure allows the model to capture heterogeneity in customer behavior at two distinct levels: the outer level clusters differentiate groups of customers by how their purchase counts respond to promotional emails and intermediate behaviors, while the inner level subclusters capture finer-grained heterogeneity in the distributions of time-varying covariates and baseline characteristics. 


 At the outer level ($\beta$-level), customers are grouped into clusters that share common outcome model parameters denoted by $\boldsymbol{\beta}$, which appear in the models for the outcome in \eqref{eq:EDPMmodel}.  Thus, customers with the same $\beta$-values have the same distributional relationships between their covariate histories and purchase counts. Within the cluster of customers with the same $\beta$-values, the inner level ($\theta$-level) further partitions customers into subclusters that share parameters denoted by $\boldsymbol{\theta}$ governing the distributions of mediators, treatment receipt, treatment assignment, and baseline covariates. This allows that customers who respond similarly to promotional emails in terms of purchases (i.e., are in the same $\beta$-cluster) may nonetheless differ in their email-opening behaviors or in how quickly they make purchases after receiving emails (i.e., maybe in different $\theta$-subclusters).

 The EDPM contains two sets of concentration parameters. The scalar parameter $\alpha^{\beta}$ governs the number of outer ($\beta$-level) clusters. Conditional on the outer clustering, the vector $\boldsymbol{\alpha}^{\theta\mid\beta}$ contains one concentration parameter for each outer cluster. Correspondingly, $\alpha_r^{\theta\mid\beta}$ controls the number of inner ($\theta$-level) subclusters nested within outer cluster $r$. Lower values of either $\alpha^{\beta}$ or the entries of $\boldsymbol{\alpha}^{\theta\mid\beta}$ correspond to fewer clusters at the respective level. By the square-breaking construction of the EDP, we have $H = \sum_{r=1}^{\infty} \sum_{s=1}^{\infty}\gamma_{r} \gamma_{s|r} \delta_{(\beta_r^{*},\,\theta_{s|r}^{*})},$ where
$ \gamma_r = \gamma_r^{'}\prod_{t <r} (1- \gamma_t^{'}), \gamma_t^{'} \sim \operatorname{Beta}(1, \alpha^{\beta}), \beta^{*}_r \stackrel{\text{iid}}{\sim} H_{0\beta},$ for the outer cluster level weights and
$\gamma_{s|r} = \gamma_{s|r}^{'}\prod_{t <s} (1- \gamma_{t|r}^{'}), \gamma_{t|r}^{'} \sim \operatorname{Beta}(1, \alpha_r^{\theta\mid\beta})$, and $\theta^{*}_{s|r} \stackrel{\text{iid}}{\sim} H_{0\theta|\beta}$ for inner cluster level weights. Following \citet{burns2023truncation} and \citet{daniels2023bayesian}, we implement the EDPM in this article via a truncation approximation. Sections~\ref{sec:S4} and~\ref{sec:S5} of the Supplementary Material give details of the approximation and provide the associated posterior computation.

Within each $\beta-$level cluster, we assume the following generalized linear model (GLM) for \\ $F_{y}\big(. \big| \textbf{m}^{(2)}_{i,T}, \textbf{m}^{(1)}_{i,T}, \textbf{d}_{i,T}, \textbf{z}_{i,T}, \boldsymbol{\ell}_{i,0}; \boldsymbol{\beta}_i \big)$ that accounts for the zero-inflation present in $Y_i$:  
\begin{equation}\label{Eq::LocalYRegression}
    Y_i \big| \mathbb{X}^{Y}_i; \boldsymbol{\beta}^{Y}_i, \sigma^{2,Y}_i,\pi^{Y}_{i}  = \pi^{Y}_{i}\mathbb{I}\{Y_i =0\} + \big(1-\pi^{Y}_{i}\big)N\bigl( \mathbb{X}^{Y}_i\boldsymbol{\beta}^{Y}_i, \sigma^{2,Y}_i\bigr) \mathbb{I}\{Y_i >0\}. 
\end{equation}
This Gaussian hurdle specification originates from \citet{cragg1971statistical}, who paired a probit hurdle with a normal component for the positive values in modeling the demand for durable goods. Although $Y_i$ is a count, \citet{yen1999gaussian} shows that Gaussian and count-data hurdle models yield nearly indistinguishable inferences for zero-inflated consumption outcomes,  while the Gaussian hurdle model also allows faster and scalable posterior computation.
In (\ref{Eq::LocalYRegression}), $\mathbb{X}^{Y}_i = \big(1,\textbf{L}_{i,0}, Z_{i,1}, D_{i,1}, M^{(1)}_{i,1}, M^{(2)}_{i,1}, \ldots,Z_{i,T}, D_{i,T}, M^{(1)}_{i,T}, M^{(2)}_{i,T}\big)$ is the $i^{th}$ row of the design matrix $\mathbb{X}^{Y}$ involving $\textbf{L}_{i,0}$, $\textbf{Z}_T, \textbf{D}_T, \textbf{M}^{(1)}_{T}$ and $\textbf{M}^{(2)}_{T}$ for the local outcome regression. For all $t = 1, \ldots, T$, we assume the following GLMs for 
\begin{align*}
&F_{(m^{(2)}_{t}, m^{(1)}_{t})}\big(.\big|\textbf{m}^{(2)}_{i,t-1}, \textbf{m}^{(1)}_{i,t-1}, \textbf{d}_{i,t}, \textbf{z}_{i,t}, \boldsymbol{\ell}_{i,0}, b_i^{M^{(2)}}, b_i^{M^{(1)}}; \boldsymbol{\theta}^{M^{(2)}}_{i,t}, \boldsymbol{\theta}^{M^{(1)}}_{i,t} \big), \\
&F_{d_{t}}\big(. \big|\textbf{m}^{(2)}_{i,t-1}, \textbf{m}^{(1)}_{i,t-1}, \textbf{d}_{i,t-1}, \textbf{z}_{i,t}, \boldsymbol{\ell}_{i,0}, b_i^{D}; \boldsymbol{\theta}^{D}_{i,t} \big), \quad\text{and} \\
&F_{z_{t}}\big(. \big|\textbf{m}^{(2)}_{i,t-1}, \textbf{m}^{(1)}_{i,t-1}, \textbf{d}_{i,t-1}, \textbf{z}_{i,t-1}, \boldsymbol{\ell}_{i,0}, b_i^{Z}; \boldsymbol{\theta}^{Z}_{i,t} \big):
\end{align*}
\begin{equation}\label{Eq::LocalMDZRegression}
\resizebox{\textwidth}{!}{$\displaystyle
\begin{aligned}
         M^{(1)}_{i,t}\big|\mathbb{X}^{M}_{i,t}, b_i^{M^{(1)}};   \boldsymbol{\theta}^{M^{(1)}}_{i,t},\sigma^{2,M^{(1)}}_{i,t}, \pi^{M^{(1)}}_{i,t}  &=  \pi^{M^{(1)}}_{i,t}\mathbb{I}\{M^{(1)}_{i,t} =0\} + \big(1- \pi^{M^{(1)}}_{i,t}\big)N\bigl(\mathbb{X}^{M}_{i,t}\boldsymbol{\theta}^{M^{(1)}}_{i,t} + b_i^{M^{(1)}}, \sigma^{2,M^{(1)}}_{i,t}\bigr) \mathbb{I}\{M^{(1)}_{i,t} >0\},  \\
         M^{(2)}_{i,t}\big|\mathbb{X}^{M}_{i,t}, b_i^{M^{(2)}};   \boldsymbol{\theta}^{M^{(2)}}_{i,t},\sigma^{2,M^{(2)}}_{i,t}, \pi^{M^{(2)}}_{i,t}  &=  \pi^{M^{(2)}}_{i,t}\mathbb{I}\{M^{(2)}_{i,t} =0\} + \big(1- \pi^{M^{(2)}}_{i,t}\big)N\bigl(\mathbb{X}^{M}_{i,t}\boldsymbol{\theta}^{M^{(2)}}_{i,t} + b_i^{M^{(2)}}, \sigma^{2,M^{(2)}}_{i,t}\bigr) \mathbb{I}\{M^{(2)}_{i,t} >0\},  \\
      D_{i,t}\big|\mathbb{X}^{D}_{i,t}, b_i^{D};   \boldsymbol{\theta}^{D}_{i,t}  &\sim \operatorname{Bernoulli}\bigl(p_{i,t}^D\bigr), \quad \operatorname{probit}\big(p_{i,t}^D\big) = \mathbb{X}^{D}_{i,t} \boldsymbol{\theta}^{D}_{i,t} +b_i^{D}, \text{ and} \\
      Z_{i,t}\big|\mathbb{X}^{Z}_{i,t}, b_i^{Z};   \boldsymbol{\theta}^{Z}_{i,t}  &\sim \operatorname{Bernoulli}\bigl(p_{i,t}^Z\bigr), \quad \operatorname{probit}\big(p_{i,t}^Z\big) = \mathbb{X}^{Z}_{i,t} \boldsymbol{\theta}^{Z}_{i,t} +b_i^{Z}.
\end{aligned}
$}
\end{equation}

 In (\ref{Eq::LocalMDZRegression}), $\mathbb{X}^{M}_{i,t} = \big(1,\textbf{L}_{i,0}, Z_{i,1}, D_{i,1}, \ldots, Z_{i,t}, D_{i,t}\big)$, $\mathbb{X}^{D}_{i,t} = \big(1,\textbf{L}_{i,0}, Z_{i,1}, \ldots, Z_{i,t}\big)$, and $\mathbb{X}^{Z}_{i,t} = \big(1, \textbf{L}_{i,0}\big)$ denote the $i^{th}$ rows of the design matrices governing the mediator, treatment receipt, and treatment assignment models at time $t$, respectively. Here, $\mathbf{L}_0$ denotes the baseline covariates, $\mathbf{Z}_t$ the treatment assignment history, and $\mathbf{D}_t$ the treatment receipt history. The model has the following specifics:
 (i) conditional on $\mathbf{L}_0$ and subject-specific random intercepts, $b_i^D, b_i^{M^{(1)}}, b_i^{M^{(2)}}$, the mediator and treatment receipt models depend on lag-one mediator and treatment values
 rather than the complete covariate history in strict temporal order, (ii) the local treatment assignment model for $Z_{i,t}$ depends on baseline covariates and a subject-specific random intercept, $b_i^Z$, implying that treatment assignment is conditionally independent of the past  within clusters given $\boldsymbol{\ell}_{i,0}$.  Although marginally, the treatment assignment and the past may be dependent because cluster membership is informed by the longitudinal history through the $\theta$-level likelihood. Marginalizing over the latent cluster therefore induces dependence of $Z_{i,t}$ on $\textbf{Z}_{i,t-1}$, $\textbf{D}_{i,t-1}$, and $\textbf{M}_{i,t-1}$, consistent with Assumption~\ref{assump::2}, (iii) the local density of the outcome model
 (\ref{Eq::LocalYRegression}) is specified conditional on baseline covariates $\boldsymbol{\ell}_{i,0}$, treatment assignment history $\mathbf{z}_{i,T}$, treatment receipt history $\mathbf{d}_{i,T}$, and mediator history $\big(\mathbf{m}^{(1)}_{i,T},\mathbf{m}^{(2)}_{i,T}\big)$, and (iv) the baseline covariates $\boldsymbol{\ell}_{i,0} = [\ell_{i,0}^{(1)}, \ldots, \ell_{i,0}^{(K)}]$ are modeled as locally independent within clusters.

 The inclusion of all longitudinal covariates in the local outcome regression allows us to better capture the effects of treatment and mediators on the outcome. This modeling choice highlights a key strength of our framework: the flexibility to specify local regression models with different sets of covariates within a unified Bayesian framework.  In general, the framework may be tailored according to the case study.  The random intercepts $b_i^{M^{(2)}}$, $b_i^{M^{(1)}}$, $b_i^{D}$, and $b_i^{Z}$  are modeled as independent and centered normal distributions. These random effects are not included in the EDP prior so they do
not depend on clusters. The baseline covariates are modeled to be locally independent, with $L_{i,0}^{(k)}; \mu^{L_{0}^{(k)}}_{i}, \sigma^{2,L_{0}^{(k)}}_{i} \sim N\big(\mu^{L_{0}^{(k)}}_{i}, \sigma^{2,L_{0}^{(k)}}_{i}\big)$  for continuous and $L_{i,0}^{(k)}; \pi^{L_{0}^{(k)}}_{i} \sim Bernoulli\big(\pi^{L_{0}^{(k)}}_{i}\big)$ for binary covariates.

 Note that all variables are globally dependent and may exhibit complex non-linear relationships, even when parametric models are assumed within each cluster. This flexible modeling of conditional distributions is a key advantage of the EDPM model. It accommodates non-normality and multi-modality of errors at the global level while simultaneously supporting simple parametric GLMs locally within a cluster. The conditional densities that appear in the nonparametric identification of the causal parameter $\theta(\mathbf{z}, \mathbf{z}_*)$ in Equation~(\ref{Eq:PropGformula_NoRE}) can be expressed as infinite mixture models under the EDPM specification in Equation~(\ref{eq:EDPMmodel}). Section~\ref{sec:S4} of the Supplementary Material provides details on EDP base measures, prior specifications, and on the derivation of densities in (\ref{Eq:PropGformula_NoRE}).

\subsection{Posterior sampling and g-computation}\label{sec:4.2}
We fit the observed-data models described in Section~\ref{sec:4.1} using MCMC implemented in \texttt{Nimble} \citep{de2017programming,de2020nimble}. After discarding the first $B$ iterations as burn-in, we retain $Q$ posterior samples (after thinning by a factor of 10), indexed by $q=1,\ldots,Q$. At each retained iteration, we obtain posterior draws of the parameters governing $L_{i,0}$, $Z_{i,t}$, $D_{i,t}$, $M^{(1)}_{i,t}$, $M^{(2)}_{i,t}$, and $Y_i$.

To facilitate posterior computation under the enriched Dirichlet process mixture model, we adopt the finite truncation approximation of \cite{burns2023truncation}, replacing the infinite outer and inner mixtures by truncation at finite levels. This approximation yields closed-form posterior updates for the quantities in identification equation  (\ref{Eq:PropGformula_NoRE}) and enables efficient blocked Gibbs sampling.  Without this approximation, posterior inference would require drawing samples from posterior predictive distributions within the G-computation framework. See details in Section~\ref{sec:S5} of the Supplementary Material.

Given the $Q$ posterior draws, we estimate $\theta(\mathbf{z},\mathbf{z}_{*})$ using G-computation with Monte Carlo integration. Algorithm~\ref{Alg1:Gcomp} in the Supplementary Material provides the complete procedure. Briefly, at each retained iteration $q$, we generate $C^*$ Monte Carlo samples of the baseline covariates and random effects from their posterior predictive distributions. For each sample, we fix the principal stratum  and recursively simulate the counterfactual treatment and mediator trajectories under the intervention regimes of interest using the conditional distributions implied by the truncated EDPM. We then evaluate the conditional expectation of the outcome at time $T$ and approximate $\theta(\mathbf{z},\mathbf{z}_{*})^{(q)}$ by averaging over the $C^*$ Monte Carlo samples.

\section{Case study results} \label{sec:5}
Recall that there are $18,571$ customers who received emails over three time periods. In each time period, they either received a value-added incentive email (V) or a price-discount incentive email (P). Thus, we have treatment regimes of the form `\{P,P,P\}' indicating price-discount incentive email in each of the three periods, etc. The number of days since the last email was opened
and the number of days since the last purchase are the two intermediate variables of interest for the effect of email incentives on purchases. 
Throughout this section, we report posterior means and $95\%$ Bayesian credible intervals to summarize our direct, indirect, and total interventional effect estimates. Our inference is based on 1,000 posterior samples from four MCMC chains. After discarding 17,500 iterations as burn-in and thinning by a factor of 10, each chain retains 250 iterations. 
Section~\ref{sec:S3} of the Supplementary Material reports on sensitivity analysis to potential violations of Assumption~\ref{assump::3}.3.

\subsection{Impact of adjustment for compliance with opening email}\label{sec:51} 

A naive analysis could ignore the longitudinal treatment assignment, mediators, compliance with opening the email, or all of the above. Ignoring the fact that the treatment only becomes effective when individuals open the email would lead to a biased understanding of the treatment effect. We compare our inference with a naive analysis that accounts for all features of the data except the email-opening action to quantify the influence of adjustment for compliance with opening.

The naive analysis uses Gaussian hurdle models for the zero-inflated mediators and the final outcome, similar to our inference model. The hurdle models consist of (i) a zero component governed by a Bernoulli hurdle probability and (ii) a continuous Gaussian component whose mean depends additively on all temporally preceding variables. All model parameters---including regression coefficients, variance components, and hurdle probabilities---are assigned weakly informative priors. Section~\ref{sec:S1} of the Supplementary Material gives the full details on the naive causal estimands and the corresponding parametric Bayesian model. Using the proposed inference and the naive inference, we compare the treatment regimes \{V,V,V\} and \{P,P,P\}, i.e., we compare all value-added incentive emails vs all price-discount incentive emails. The results are reported in Table \ref{tab:naive_vs_proposed}. 

\begin{table}[h!]
\caption{Effect of different promotional emails on purchase; comparing three subsequent value-added offer emails vs three subsequent price-discount offer emails. The principal strata are based on customers' compliance with opening the emails.}
\centering
\resizebox{\textwidth}{!}{%
\begin{tabular}{l|ccccc}
        \hline
                & \multicolumn{1}{c}{Ignoring email-} & \multicolumn{4}{c}{Principal   Strata of Customer Type}                                                                                                  \\ \cline{3-6} 
                & \multicolumn{1}{c}{opening information}  & \multicolumn{1}{c}{Active} & \multicolumn{1}{c}{Value-attentive} & \multicolumn{1}{c}{Price-attentive} & \multicolumn{1}{c}{Non-active} \\ \hline
Direct effect   & 2.76  &   3.91  &   3.98   &  3.87  &   3.94   \\
                & ($-$7.41, 5.06) & (2.45, 5.23)    & (2.46, 5.23)   & (2.54, 5.34)    & (2.57, 5.25)  \\
Indirect effect & $-$0.06   & -1.15   & -1.42   & -1.22    & -1.48    \\
                & ($-$2.39, 2.53)  & (-2.22, -0.33)  & (-2.45, -0.68)  & (-2.31, -0.45) & (-2.53, -0.80)   \\ \hline
\end{tabular}%
}
\label{tab:naive_vs_proposed}
\end{table} 

  Throughout the table, the direct effect estimates are positive and the mediated indirect effect estimates are negative. A positive posterior mean for the direct effect indicates that, given the observed data, if all individuals were emailed \{V,V,V\} versus if all were emailed \{P,P,P\}, there is an expected positive direct impact on the number of purchases. Table \ref{tab:naive_vs_proposed} shows that the direct effect estimate from the proposed inference is significant and approximately $1.5$ times the magnitude of the non-significant naive estimate when adjusting for email-opening information. 

  All indirect effect estimates in Table \ref{tab:naive_vs_proposed} are negative. A negative estimate indicates that a longer receipt time for a value incentive email, compared to that for the price incentive email, acts as a deterrent to purchasing items. In Section \ref{sec:54}, we investigate this indirect effect further. The naive method infers a negligible indirect effect, whereas adjusting for email-opening information yields statistically significant effect estimates.

 The clear differences in email opening rates between value- vs.~price-incentive emails seen in Figure \ref{fig:recipients_vs_openers_sampleSize} provide an explanation for these differences between naive and proposed inferences. In Figure \ref{fig:recipients_vs_openers_sampleSize}, the email-opening rates for value-incentive emails are lower than those for price-incentive emails. The largest difference is at time $t= 3$, where these rates differ by more than $5\%$.

\begin{figure}
    \centering
    \includegraphics[width=.75\linewidth]{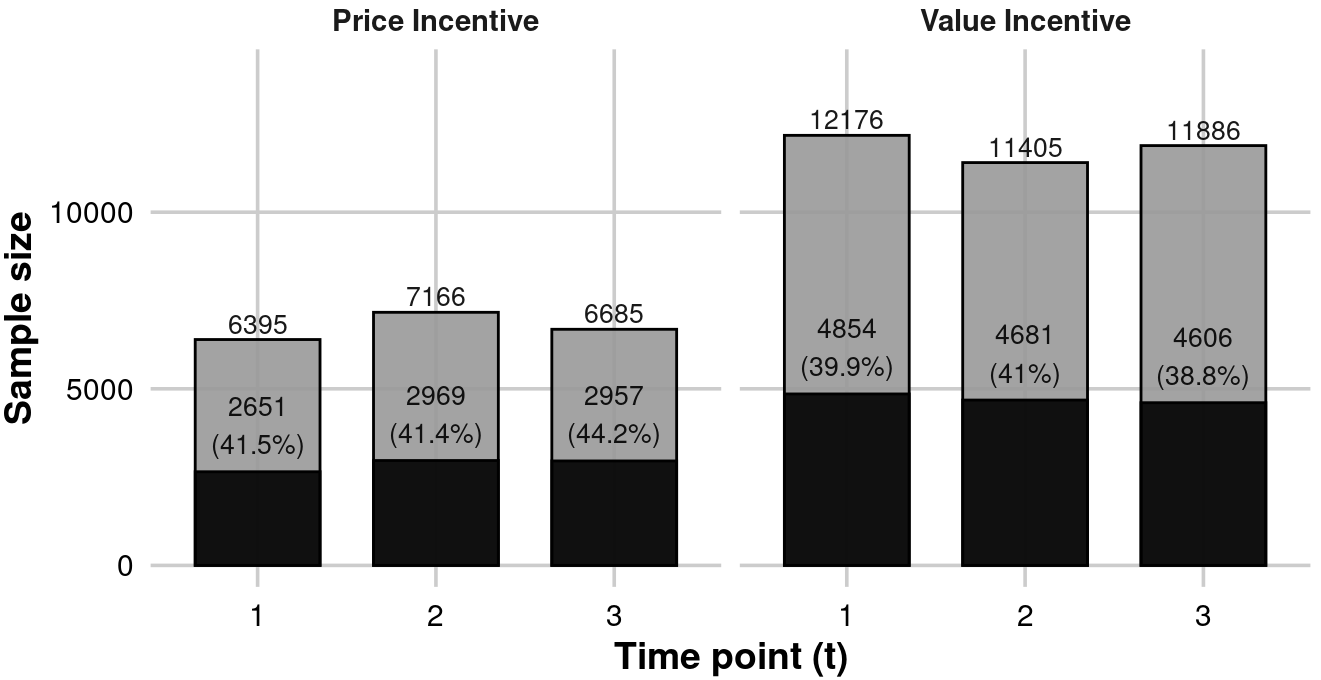}
    \caption{The number of recipients who received the promotional emails and who opened emails over time, stratified by incentive type. The number above the gray bar is the number of email recipients, and the number above the black bar is the number of email openers. 
     }
    \label{fig:recipients_vs_openers_sampleSize}
\end{figure}

  Table \ref{tab:naive_vs_proposed} also shows inference separately within our four principal strata. Recall that we stratify individuals into four strata based on their latent behavior of whether they would open either type of email (active customer), only value-added offer email (value-attentive customer), only price-discount offer email (price-attentive customer), or neither type of email (non-active customer) at time $t= 1$. The estimates are very similar across these four principal strata.

\subsection{Cumulative effects of value-incentive emails across the campaign}\label{sec:52} 

The results in Table~\ref{tab:naive_vs_proposed} demonstrate that sending value-added incentive emails throughout the campaign yields a net positive impact on purchases relative to price-discount emails. From a practical standpoint, the company may wish to understand how the cumulative benefit of value-added emails evolves as they are introduced sequentially into the campaign. Such information is critical for resource allocation: if most of the benefit accrues from the first value-added email, subsequent emails may be allocated to price discounts, and conversely, if benefits accumulate across all time points, a consistent value-added strategy is warranted. To address this question, we estimate principal interventional effects by comparing three treatment regimes---$\textbf{Z}_{T} = \text{\{V,P,P\}}$, $\textbf{Z}_{T} = \text{\{V,V,P\}}$, and $\textbf{Z}_{T} = \text{\{V,V,V\}}$---against the baseline regime $\textbf{Z}_{*,T} = \text{\{P,P,P\}}$.

  Table~\ref{Table:EffectEstimatesvsBaselineRegime} presents the estimated direct, indirect, and total effects for each regime contrast across the four principal strata. The direct effect estimates are uniformly positive and statistically significant, increasing monotonically as more value-added emails are introduced: approximately two additional purchases for \{V,P,P\}, three for \{V,V,P\}, and about four for \{V,V,V\} relative to \{P,P,P\}. 
  In contrast, the indirect effect estimates are negative, and they are statistically significantly negative only for the \{V,V,V\} vs. \{P,P,P\} comparison. The magnitudes of the indirect effect estimates' posterior means increase in absolute value as more value-added emails are sent. 
  Finally, the total effect estimates, which combine direct and indirect effects, remain positive and significant across all comparisons, ranging from approximately $1.46$ to $1.83$ additional purchases on average for \{V,P,P\} to more than 2.45 additional purchases on average for \{V,V,V\}. Notably, all effect estimates are stable across principal strata, indicating that the benefits of value-added emails are not confined to a particular compliance subgroup.

\begin{table}[!h]
\caption{Causal effect estimates (posterior mean and $95\%$ Bayesian credible interval) across treatment regime contrasts}
\label{Table:EffectEstimatesvsBaselineRegime}
\centering
\resizebox{\textwidth}{!}{%
\begin{tabular}{@{}lcrcrrr@{}}
\hline
Principal strata & Effects &  \{V,P,P\} vs. \{P,P,P\} &  \{V,V,P\} vs. \{P,P,P\}   &  \{V,V,V\} vs. \{P,P,P\} \\
 at $t=1$ &   &  &  & \\
\hline
{Active}  & {Direct} & 1.86 (0.82, 3.02) & 3.22 (2.19, 4.49) & 3.91 (2.45, 5.23)\\
{customers}  & {Indirect}  & $-$0.03 ($-$0.58, 0.80) & $-$0.55 ($-$1.25, 0.24) & $-$1.15 ($-$2.22, $-$0.33)\\
          & {Total}   & 1.83 (0.69, 3.39) & 2.67 (1.39, 4.19) & 2.76 (1.15, 3.83)\\[6pt]
{Non-price}  & {Direct} & 1.83 (0.92, 2.76) & 3.20 (2.08, 4.28) & 3.98 (2.46, 5.23)\\
 {value-attentive}  & {Indirect}  & $-$0.27 ($-$0.85, 0.67) & $-$0.76 ($-$1.45, 0.20) & $-$1.42 ($-$2.45, $-$0.68)\\
  {customers}  & {Total}   & 1.57 (0.50, 2.97) & 2.44 (1.18, 3.86) & 2.56 (0.99, 3.59)\\[6pt]
{Price-attentive} & {Direct} & 1.87 (0.77, 3.40) & 3.24 (2.05, 4.82) & 3.87 (2.54, 5.34)\\
{customers}  & {Indirect}  & $-$0.16 ($-$0.77, 0.50) & $-$0.67 ($-$1.40, 0.03) & $-$1.22 ($-$2.31, $-$0.45)\\
          & {Total}   & 1.71 (0.59, 3.29) & 2.57 (1.29, 4.23) & 2.65 (1.18, 3.68)\\[6pt]
{Non-active}  & {Direct} & 1.85 (0.83, 2.98) & 3.22 (2.08, 4.51) & 3.94 (2.57, 5.25)\\
{customers}   & {Indirect}  & $-$0.39 ($-$0.93, 0.36) & $-$0.88 ($-$1.58, $-$0.09) & $-$1.48 ($-$2.53, $-$0.80)\\
          & {Total}   & 1.46 (0.46, 2.73) & 2.33 (1.13, 3.79) & 2.45 (0.98, 3.42)\\[6pt]
\hline
\end{tabular}%
}
\end{table}

\subsection{Impact of incentive types sent over time}\label{sec:53} 

There is no apparent cost to simply emailing a value-added or a price-discount offer. However, there may be differences in the costs required to implement these offers. Therefore, identifying which email regimes increase the expected purchase count can help the company make cost-benefit-informed decisions about which promotional emails to send in the future. To investigate this, we compare the expected number of purchases across the eight possible email sequences that the company could send. Note that these are still counterfactual quantities. 
Figure~\ref{fig:effect_across_regimes} presents the estimates and corresponding credible intervals for the expected number of purchases at the end of the email campaign. 
\begin{figure}[htbp]
    \centering
    \includegraphics[width=.7\linewidth]{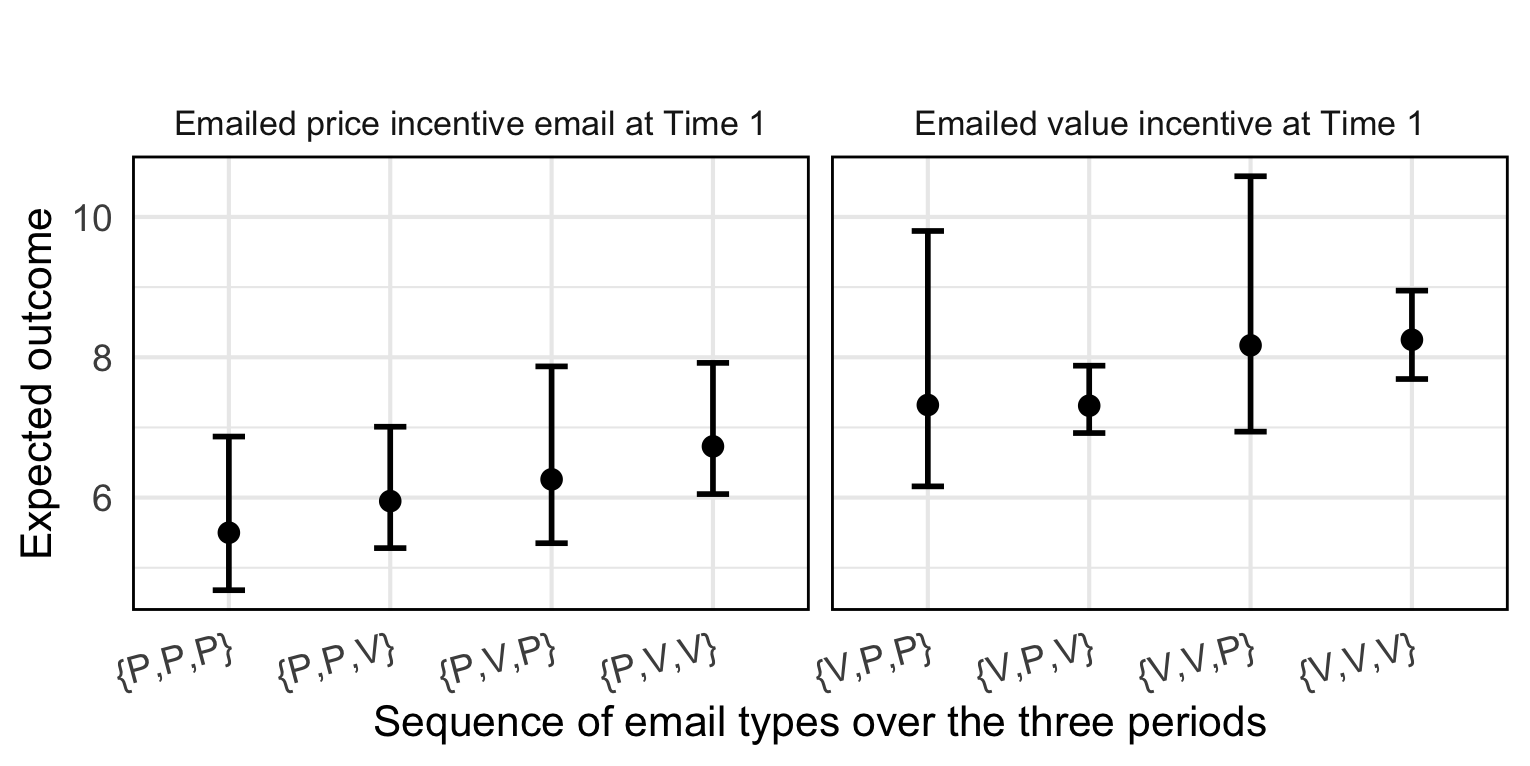}
    \caption{The expected number of purchases at time $t=3$, i.e., after the end of the email campaign, following different types of emails.} 
    \label{fig:effect_across_regimes}
\end{figure}
 Figure~\ref{fig:effect_across_regimes} shows that, overall, value-added offers are more appealing to customers than price-discount offers. Receiving a value-added email, rather than a price-discount email, at any of the three time points increases the expected number of purchases. For example, the email sequence \{V,V,V\} yields approximately $1.5$ times as many expected purchases as the sequence \{P,P,P\}, increasing from $5.50$ (95\% CI: $(4.69, 6.91)$) to $8.25$ (95\% CI: $(7.69, 8.95)$). All numerical values corresponding to Figure~\ref{fig:effect_across_regimes} are provided in Table~\ref{Table:AppendixThetaZZstar} in the Supplementary Material. 

 More broadly, within each panel of Figure~\ref{fig:effect_across_regimes}, the expected number of purchases exhibits a general upward trend as the number of value-added emails in the sequence increases, regardless of the temporal ordering. This pattern suggests that the benefit of value-added incentives is approximately cumulative: substituting an additional price-discount email with a value-added email tends to increase the expected purchase count, regardless of when the substitution occurs in the sequence. In fact, a value-added offer email sent at time $t=1$ has the largest estimated impact. 


\subsection{Examining the dynamics of indirect effects across mediator levels} \label{sec:54}

Next, we investigate how the two intermediate variables impact the indirect effects. Earlier, Table \ref{tab:naive_vs_proposed} showed a negative estimate of mediated indirect effect when comparing \{V,V,V\} to \{P,P,P\}. Under our potential outcomes framework, mediator values are drawn from their counterfactual conditional distributions, see Equation (\ref{Eq::joint_mediator_distn}). Our mediators encode how quickly customers open the promotional email and how quickly they make a purchase. In our problem, the company may have opportunities to influence customers by adjusting the mediator values. 

  To gain deeper insight into how the indirect effects evolve as the mediators increase, i.e., as the wait times in opening an email and making a purchase extend, we re-estimate the causal effects for the two target treatment regime contrasts (i) $\mathbf{Z} = \{\text{V,V,V}\}$ versus $\mathbf{Z}_{*} = \{\text{V,V,V}\}$ and (ii) $\mathbf{Z} = \{\text{P,P,P}\}$ versus $\mathbf{Z}_{*} = \{\text{P,P,P}\}$ by fixing the mediators at prespecified values rather than drawing them randomly from their mediator distributions. Specifically, for both target regime contrasts, we jointly fix the two mediator values in the $\mathbf{Z}$ arm at increasing delay times of 1, 8, 16, and 24 days, while fixing both mediators in the $\mathbf{Z}_{*}$ arm at 0 days. In other words, instead of evaluating the effects under the as-is value of the mediator (random draw from $\mathbfcal{G}^{\textbf{z}_{*,t},\textbf{d}_{t}(z_{*,t})}_{m_{t}}$), we impose a series of static interventions that fix the mediators at prespecified values to examine how causal estimates vary as the mediators increase. Contrasting between $\mathbf{Z}$ and $\mathbf{Z}_*$ isolates the interventional effect of the mediator and its interaction with the treatment.

\begin{figure}[h!]
    \centering
    \includegraphics[width=.7\linewidth]{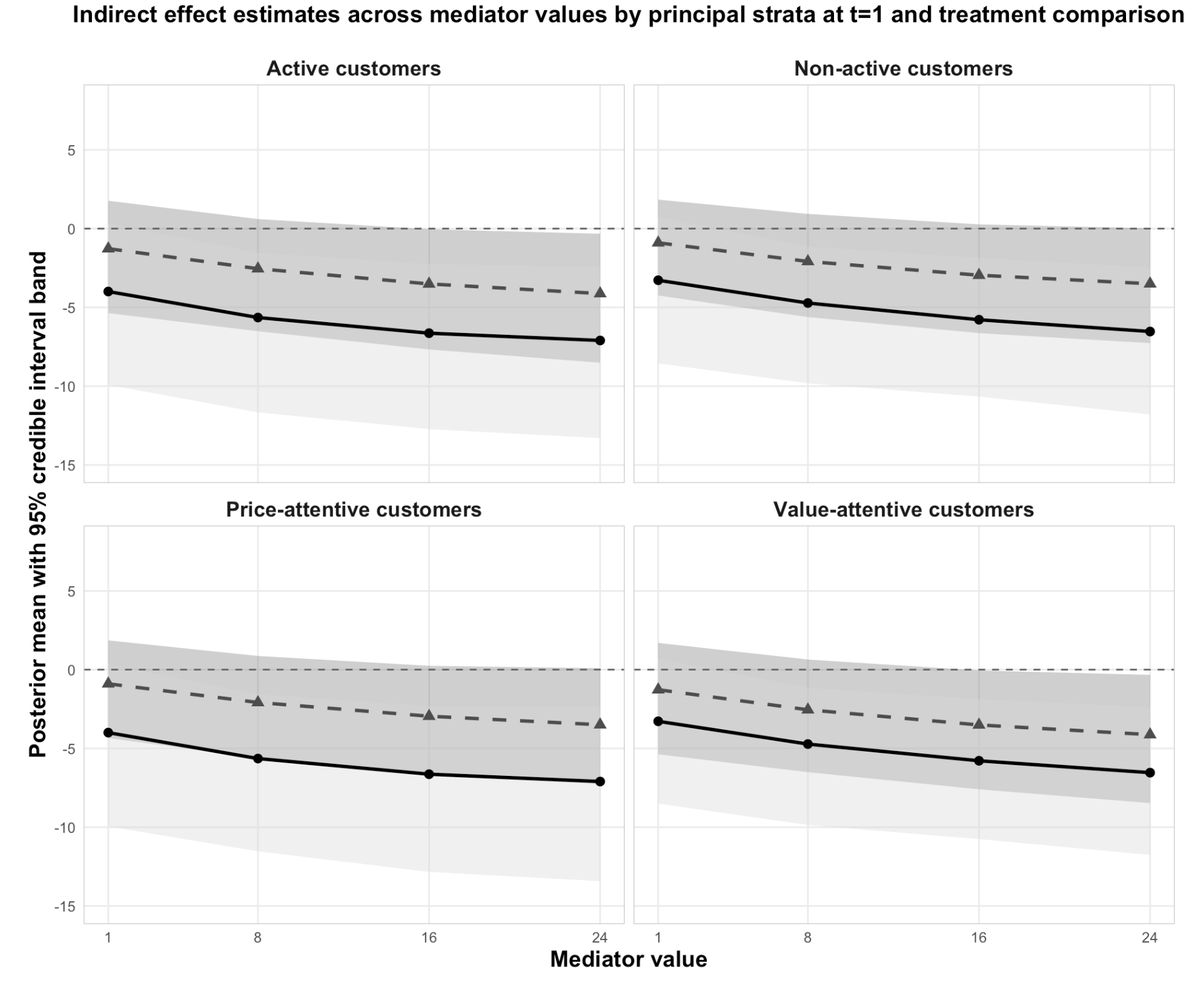}
    \caption{Indirect effects comparing different mediator values in the $\mathbf{Z}$ arm, with both mediators in the $\mathbf{Z}_{*}$ arm fixed at zero. The darker band (dashed line, triangles) corresponds to the \{P,P,P\} vs \{P,P,P\} regime; lighter band (solid line, circles) corresponds to the \{V,V,V\} vs \{V,V,V\} regime. }
    \label{fig:IndirectEffectEstimatesAcrossM181624VVVvsPPP.png}
\end{figure}

Figure~\ref{fig:IndirectEffectEstimatesAcrossM181624VVVvsPPP.png} plots the trends in indirect effects under static interventions on the mediators as functions of the mediator values. We observe that increasing the mediator values leads to a decline in the estimated indirect effects. These negative estimates likely reflect the fact that 
customers tend to become less engaged with the brand with longer waits between their purchases.

Consider the two comparisons (i) $\mathbf{Z} = \{\text{V,V,V}\}$ versus $\mathbf{Z}_{*} = \{\text{V,V,V}\}$ and (ii) $\mathbf{Z} = \{\text{P,P,P}\}$ versus $\mathbf{Z}_{*} = \{\text{P,P,P}\}$. The indirect effect is more negative for the former comparison. The results indicate that the waning interest or reduced purchase intent is more pronounced with value-added promotional emails. Section~\ref{sec:S7} of the Supplementary Material expands on these results and shows, among all regimes, that the smallest and the largest negative indirect effects are seen for \{P,P,P\} and \{V,V,V\} respectively. A comparison of the indirect effects across all eight regimes at $T=3$ is provided in Figure~\ref{fig:IndirectEffectEstimatesAcrossM181624.png} in the Supplementary Material.


\subsection{Optimal emailing sequence customized to individual behavior}\label{sec:55}

The findings so far established the direct and indirect effects of different incentive emails sent over time. They further quantified the benefits of sending emails early.

 A general goal here is to automate which sequence of emails should be sent to the customers. Figure \ref{fig:effect_across_regimes} shows that the globally optimal email sequence type is \{V,V,V\}. If this email is sent to the customers, the value-attentive, price-attentive, non-active, and active customers make an estimated 8.05, 8.25, 8.05, and 8.25 purchases at the end of the email campaign, respectively.  %

\begin{figure}[!h]
    \centering
    \includegraphics[width = .92\linewidth]{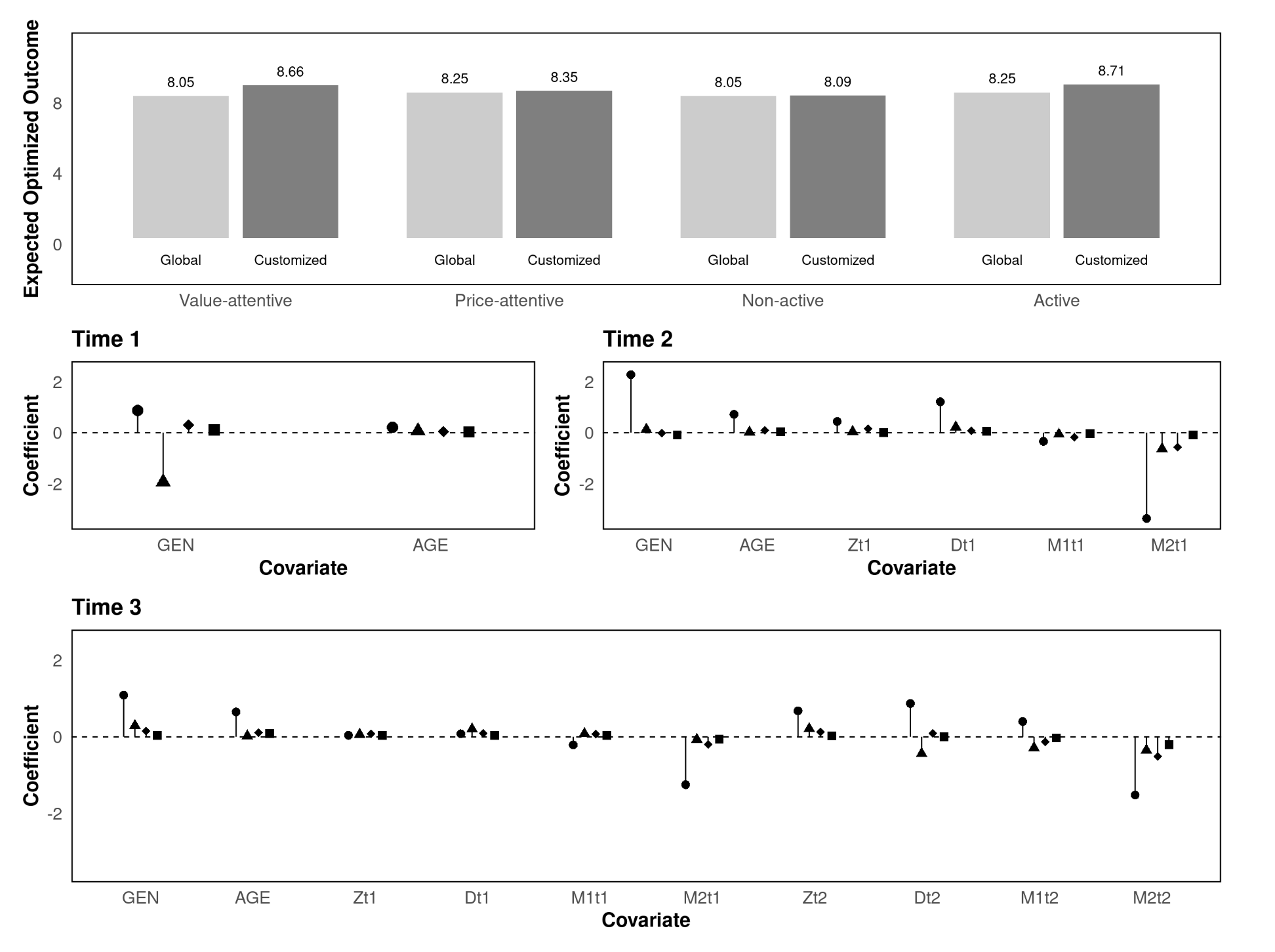}
    \caption{Optimizing email sequence for the expected number of purchases. The top panel compares the expected number of purchases under the global optimal \{V,V,V\} with those under the customized regime. The bottom panels show the coefficients for the logistic assignment model to value incentive email vs price incentive email at times 1, 2 and 3 of the covariates available till then. The symbols dot, triangle, rhombus, and square are for value-attentive, price-attentive, non-active, and active customers, respectively. 
    }
    \label{Fig:OptimEmails}
\end{figure}

 The globally optimal email sequence is deterministic and does not account for heterogeneity across customers. Thus, there is likely an opportunity to further increase purchases by designing a more individualized email campaign strategy. Such an individualized strategy could still be deterministic, assigning exactly one email sequence based on the observed characteristics of the customer. However, if the customer base changes---consequently the effects change---a deterministic strategy prevents us from evaluating these causal effects because of the violation of the treatment-overlap assumption required for causal inference.

 Thus, we design a stochastic customized email campaign that assigns probabilities to which incentive is sent based on the observed characteristics of the individuals. We fit a logistic model at each time point. The assignment model takes as input the baseline covariates gender and age for assignment at time $t=1$, 
the baseline covariates along with the email sent, the email opened, and the mediator values observed at time $t = 1$ for assignment at time $t=2$, and so on.

  The results reported in Figure \ref{Fig:OptimEmails} show a consistent increase in the estimated purchase count under the customized strategy over the globally optimal fixed strategy across the principal strata. The increase is most pronounced for value-attentive customers whose estimated number of purchases increases from 8.05 under the global strategy \{V,V,V\} to 8.66 under the designed customized regime, and is least pronounced for non-active customers where it  increases from 8.05 to 8.09.

  Figure \ref{Fig:OptimEmails} also shows the coefficients for logistic assignment models at each time point. It is clear from Figure \ref{Fig:OptimEmails} that the coefficients differ across the principal strata. Notably, the coefficients differ more for value-attentive customers compared to the other three customer types, indicating that value-attentive customers are targeted differently. For example, when a value-attentive customer opens their emails at either time $t=1$ or $t=2$ (i.e., $D_{i1}=1$ or $D_{i2}=1$), they are more likely to be assigned a value-incentive email in the following time period. Additionally, for all customers, waiting longer to purchase, i.e., larger mediator values, typically results in a higher chance of being assigned a price incentive email. This chance is also much higher for value-attentive customers than for others. 
  By calculating the coefficients of the optimal assignment model, we provide the company with a readily implementable email marketing strategy.

\section{Simulation studies}\label{sec:6} 

We conduct simulation studies to evaluate the performance of the proposed EDPM model 
by (i) assessing how the EDPM model performs compared to a correctly specified parametric model, and (ii) examining the robustness of the EDPM model under parametric misspecification. Mimicking the case study, the two different data-generating models used in our analysis are based on parameter estimates obtained from fitting parametric models to the empirical dataset described in Section~\ref{sec:21}. For each scenario, we generate $R = 500$ replicated datasets and fit the proposed EDPM model and compare it against a single-component Bayesian parametric model. Section~\ref{sec:S6} of the Supplementary Material details the data-generating mechanism, simulation scenarios, and estimation procedures.

\subsection{Results under correct model specification}\label{sec:61} 

The top half of Table~\ref{Table:SimTrueParamModel} summarizes the results at $n = 15{,}000$ for Scenario~1, where the data-generating mechanism matches the parametric model used for estimation. The corresponding results at $n = 12{,}000$, reported in Table~\ref{Table:SimTrueParamModel_n12000} in the Supplementary Material,  yield the same substantive conclusions. The results show that both approaches recover the principal interventional direct, indirect, and total effects well across all principal strata. The correctly specified parametric model is efficient, with negligible bias and consistently small MSE. The proposed flexible EDPM model remains competitive as its estimation biases are generally small relative to the magnitude of the underlying causal effects 
across direct, indirect, and total effects.

\begin{table}[ht]
\caption{Simulation results at sample size $n = 15{,}000$ for comparing the bias, MSE, and coverage of 95\% credible intervals under a Bayesian parametric model versus EDPM model. The Bayesian parametric model is \textit{correctly specified} for the single-component parametric data-generating model, but \textit{misspecified} for the ten-component parametric data-generating model}
\label{Table:SimTrueParamModel}
\resizebox{\textwidth}{!}{
\begin{tabular}{@{}lccccccccccc@{}}
\hline
Principal &  &  &  & \multicolumn{4}{c}{Bayesian parametric model} & \multicolumn{4}{c}{EDPM model} \\
\cline{5-8} \cline{9-12}
strata & n & Effects & Ground truth & Bias & $95\%$ & Average CI & MSE & Bias & $95\%$ & Average CI & MSE \\
at $t=1$ &  &  &  &  & coverage & width &  &  & coverage & width &  \\
\hline
\multicolumn{12}{c}{Scenario 1: Single component parametric data generating model} \\
\hline
Active & 15000 & Direct & 1.55 & -0.0040 & 0.99 & 0.53 & 0.0092 &  0.0168   &    0.95  &   0.74  &          0.0377  \\
customers &  & Indirect & -0.08 & 0.0437 & 1.00 & 1.04 & 0.0021 &  0.0211    &   0.94   &    0.13 &  0.0013  \\
 &  & Total & 1.47 & 0.0397 & 1.00 & 1.04 & 0.0108 &  0.0379   &    0.95  &   0.74 & 0.0386  \\
\cline{1-12}
Non-price & 15000 & Direct & 1.57 & -0.0054 & 0.99 & 0.54 & 0.0095 &  -0.0108  &     0.94  &    0.87      &      0.0512  \\
value-attentive &  & Indirect & -0.08 & 0.0400 & 1.00 & 1.16 & 0.0019 &  0.0245  &     0.91    &   0.17  &     0.0025  \\
customers &  & Total & 1.49 & 0.0345 & 1.00 & 1.14 & 0.0107 &  0.0137  &     0.94   &   0.87 &  0.0509  \\
\cline{1-12}
Price-attentive & 15000 & Direct & 1.53 & -0.0033 & 1.00 & 0.53 & 0.0092 &  0.0422  &     0.94  &   0.88 & 0.0597  \\
customers &  & Indirect & -0.08 & 0.0471 & 1.00 & 1.05 & 0.0023 &  0.0205   &    0.94  &   0.16  &  0.0017  \\
 &  & Total & 1.45 & 0.0439 & 1.00 & 1.04 & 0.0111 & 0.0627 &      0.94  &   0.87 &  0.0611  \\
\cline{1-12}
Non-active & 15000 & Direct & 1.55 & -0.0049 & 0.99 & 0.53 & 0.0093 &  0.0138   &    0.96   &  0.74 &           0.0374  \\
customers &  & Indirect & -0.08 & 0.0440 & 1.00 & 1.02 & 0.0021 &  0.0248   &    0.92  &     0.13 &  0.0014  \\
 &  & Total & 1.47 & 0.0391 & 1.00 & 1.03 & 0.0108 &  0.0386   &    0.94   &   0.73 &  0.0383 \\
\hline
\multicolumn{12}{c}{Scenario 2: Ten component parametric data generating model} \\
\hline
Active & 15000 & Direct & 1.73 & $-0.3416$ & 0.02 & 0.35 & 0.1244 & $0.2679$ & 0.65 & 0.73 & 0.1254 \\
customers &  & Indirect & 0.03 & $-0.1809$ & 0.00 & 0.10 & 0.0331 & $-0.1713$ & 0.58 & 0.45 & 0.0656 \\
 &  & Total & 1.76 & $-0.5225$ & 0.00 & 0.36 & 0.2809 & $0.0967$ & 0.84 & 0.80 & 0.0861 \\
\cline{1-12}
Non-price & 15000 & Direct & 1.74 & $-0.3410$ & 0.02 & 0.36 & 0.1242 & $0.2235$ & 0.77 & 0.82 & 0.1151 \\
value-attentive &  & Indirect & 0.04 & $-0.1932$ & 0.00 & 0.11 & 0.0379 & $-0.1154$ & 0.71 & 0.58 & 0.0726 \\
customers &  & Total & 1.78 & $-0.5341$ & 0.00 & 0.36 & 0.2936 & $0.1081$ & 0.85 & 0.90 & 0.1173 \\
\cline{1-12}
Price-attentive & 15000 & Direct & 1.71 & $-0.3407$ & 0.03 & 0.35 & 0.1238 & $0.3107$ & 0.65 & 0.83 & 0.1624 \\
customers &  & Indirect & 0.02 & $-0.1710$ & 0.00 & 0.09 & 0.0295 & $-0.1772$ & 0.56 & 0.46 & 0.0726 \\
 &  & Total & 1.74 & $-0.5117$ & 0.00 & 0.36 & 0.2697 & $0.1335$ & 0.82 & 0.90 & 0.1079 \\
\cline{1-12}
Non-active & 15000 & Direct & 1.73 & $-0.3420$ & 0.03 & 0.35 & 0.1247 & $0.2638$ & 0.64 & 0.72 & 0.1225 \\
customers &  & Indirect & 0.03 & $-0.1816$ & 0.00 & 0.10 & 0.0334 & $-0.1194$ & 0.71 & 0.53 & 0.0618 \\
 &  & Total & 1.76 & $-0.5236$ & 0.00 & 0.36 & 0.2821 & $0.1444$ & 0.82 & 0.82 & 0.1042 \\
\hline
\end{tabular}}
\end{table}

 The main difference between the two approaches lies in the trade-off between efficiency and flexibility. Under this correct specification, the parametric model achieves a smaller MSE for direct and total effects, reflecting the efficiency advantage of a low-dimensional correctly specified model. The EDPM incurs a modest variance penalty from modeling the unknown joint distribution through a mixture representation, but this does not materially affect inference. For indirect effects, EDPM often shows slightly smaller bias and narrower intervals, suggesting that its flexible mixture structure can adapt to localized features of the mediator distribution.

 The parametric model shows empirical coverage at or near one across nearly all effects and strata, indicating conservative inference with mild over-coverage. In contrast, the EDPM achieves coverage closer to the nominal $95\%$ level, while producing somewhat wider intervals for direct and total effects. Overall, the parametric model offers greater efficiency under correct specification and the EDPM provides better-calibrated uncertainty quantification.

 Thus, the flexible EDPM closely tracks the correctly specified parametric model. 
 Figure~\ref{Fig:Sim_TRUE_EDPM_num_outer_clusters} in the Supplementary Material further shows this adaptive behavior: across replicated datasets, the model favors a richer latent structure with multiple effective outer clusters for parametric data-generating models. 

\subsection{Results under model misspecification}\label{sec:62}

Next, simulation Scenario~2 specifies an $A=10$ component finite mixture, so that the single-component Bayesian parametric model is misspecified; details in Section~\ref{sec:S6} of the Supplementary Material.
This setting reflects substantial latent heterogeneity in the distributions of the outcome, mediators, and treatment-related processes, and therefore provides a solid ground for evaluating the EDPM model.

The bottom half of Table~\ref{Table:SimTrueParamModel}
shows a clear contrast between the two approaches at $n = 15{,}000$. The corresponding results at $n = 12{,}000$ are reported in Table~\ref{Table:SimMisspecifiedParamModel_10comp_n12000} in the Supplementary Material.
 The misspecified parametric model performs poorly across all principal strata and all causal effects. Direct and total effects are consistently biased downward by roughly 20\% of their true magnitude, while indirect effects are also biased relative to their much smaller scale. These biases persist across both sample sizes ($n=12{,}000$ and $n=15{,}000$), indicating that larger samples do not mitigate structural misspecification. At the same time, the model produces narrow credible intervals with empirical coverage close to zero for nearly all estimands.

 In contrast, the EDPM results in substantially smaller biases and better calibrated intervals across all strata compared to the parametric model. For direct and total effects, absolute bias is consistently smaller, with corresponding reductions in MSE for total effects and comparable MSE for direct effects despite somewhat wider intervals. 
 However, it is notable that the EDPM model provides frequentist coverage below the nominal $95\%$ level, even though it is uniformly better calibrated than the near-zero coverage of the parametric alternative across all causal parameters and sample sizes. This is expected because while Bayesian parametric models are known to provide approximately nominal coverages of finite dimensional parameters in large samples (by the Bernstein--von Mises theorem), Bayesian nonparametric models do not conform to frequentist inferential properties without curated prior choices or posterior adjustments (see \cite{castillo2015bernstein, yiu2025semiparametric}). We do not pursue such frequentist goals as our modeling and inference for our case study rely on Bayesian point estimates and credible intervals.

 Overall, these results show that the EDPM provides meaningful protection against model misspecification by adaptively capturing latent heterogeneity that the single-component parametric model cannot represent. 
 Figure~\ref{Fig:Sim_MISSPECIFIED_EDPM_num_outer_clusters} in the Supplementary Material further supports this interpretation by showing that the EDPM adaptively recruits multiple effective outer clusters to represent the underlying heterogeneity, rather than imposing a fixed parametric form. Table~\ref{Table:SimMisspecifiedParamModel} in the Supplementary Material further reports analysis under a three-component mixture mechanism and shows that the EDPM achieves near-nominal coverage for direct and total effects, whereas the parametric model continues to exhibit substantial bias and undercoverage. 

\section{Discussion}\label{sec:7}

Our goal was to estimate the efficacy of a longitudinal digital communication campaign run on $18{,}571$ customers and design a campaign strategy that could maximize sales. 
For this purpose, we build a Bayesian semi-parametric framework for estimating principal interventional direct and indirect effects in longitudinal settings with treatment noncompliance and multiple time-varying mediators. The proposed approach combines enriched Dirichlet process mixture (EDPM) models with a G-computation step to flexibly estimate causal mediation effects within principal strata of treatment compliance. 
The post-estimation G-computation step, which constitutes the primary computational burden for causal effect estimation, is highly parallelizable across posterior draws and principal strata, making the inference step scalable to large posterior samples.

 
 Our analysis yields several findings with direct implications for the retailer's email marketing strategy. First, sending value-added incentive emails leads to higher expected purchase counts compared to price-discount emails, and this advantage is cumulative as the direct effect on purchases grows monotonically as additional time points receive value-added emails. Second, the causal effect estimates are stable across the four principal strata at $t=1$, indicating that the relative benefit of value-added emails is robust to heterogeneity in initial email-opening behavior. Finally, the individualized sequential emailing strategy developed in Section~\ref{sec:55} demonstrates that a logistic model for determining promotion assignment at each time point can improve expected outcomes relative to global emailing strategies across all latent compliance groups.

We highlight a few practical future directions here. Modern marketing campaigns are often multi-channel, involving email, SMS, mobile notifications, and on-site advertising, and extending our framework to accommodate channel-specific treatment vectors would allow evaluation of dynamic cross-channel intervention strategies and customer fatigue effects. Incorporating multimodal learning frameworks that jointly model behavioral, transactional, and channel interaction data could further improve personalization by capturing complex dependencies across heterogeneous customer information sources. In addition, the analysis of longer campaign horizons with open
populations can expand the scope of our framework. While our dataset has no dropout, real campaigns often have customers who may enter or leave the system over time, requiring new formulations
of principal strata and longitudinal mediation effects under staggered enrollment. Finally, although the EDPM substantially reduces bias arising from nuisance-model misspecification relative to Bayesian parametric models, its credible intervals are not guaranteed to achieve nominal frequentist coverage under misspecification. Therefore, incorporating a semiparametric correction \citep{yiu2025semiparametric} within our framework may help address this limitation.

\section*{Code availability}

All code for the data analysis and simulation studies is available at
\begingroup
\def\UrlBreaks{\do\/\do-\do.\do:}\sloppy
\url{https://github.com/SBstats/Bayesian-Nonparametric-Causal-Mediation-DiscreteTime-Noncompliance}.
\endgroup

\begin{funding}
Michael J. Daniels was supported by NIH grant R01 HL166324.
\end{funding}

\bibliographystyle{imsart-nameyear}
\bibliography{bibliography}

\clearpage

\section*{Supplementary Material}

\setcounter{section}{0}
\setcounter{table}{0}
\setcounter{figure}{0}
\setcounter{equation}{0}
\renewcommand{\thesection}{S\arabic{section}}
\renewcommand{\thesubsection}{S\arabic{section}.\arabic{subsection}}
\renewcommand{\thesubsubsection}{S\arabic{section}.\arabic{subsection}.\arabic{subsubsection}}
\renewcommand{\thetable}{S\arabic{table}}
\renewcommand{\thefigure}{S\arabic{figure}}
\renewcommand{\theequation}{S\arabic{equation}}
\renewcommand{\thealgorithm}{S\arabic{algorithm}}
\renewcommand{\theHsection}{S\arabic{section}}
\renewcommand{\theHtable}{S\arabic{table}}
\renewcommand{\theHfigure}{S\arabic{figure}}
\renewcommand{\theHequation}{S\arabic{equation}}

\section{Details on the Bayesian parametric model used for the naive analysis of causal effects}\label{sec:S1}

For the naive analysis of the causal effects in Section~\ref{sec:51} of the main text, we assume that the actual treatment receipt status of participants is unknown. Consequently, the observed data are given by
\[
\big\{\mathcal{O}^{\text{naive}}_{i}\big\}_{i=1}^{n} \equiv \big\{L_{i0}, Z_{it}, M^{(1)}_{it}, M^{(2)}_{it}, Y_{i}\big\}_{i=1}^{n}, \quad t \in \{1, \ldots, T\}.
\]

\subsection{Naive causal estimands}
Suppose $\mathbf{z}_{T}$ and $\mathbf{z}_{*,T}$ represent the treatment regime vectors across $T$ periods. For this analysis, we set $\mathbf{z}_{T} = \{1,1,1\}$ and $\mathbf{z}_{*,T} = \{0,0,0\}$ with $T=3$. We define the naive direct, indirect, and total effects as: 
        \begin{equation}\label{Eq:naive_effects}
            \begin{aligned}
                 \text{Direct Effect: } DE^{\text{naive}} 
                &=  \mathbb{E}\Bigl[Y\Bigl(\textbf{z}_{T}, \mathbfcal{G}^{\textbf{z}_{*,T}}_{m_{T}}\Bigr) 
                - Y\Bigl(\textbf{z}_{*,T}, \mathbfcal{G}^{\textbf{z}_{*,T}}_{m_{T}}\Bigr) \Bigr]
                \\
                \text{Indirect Effect: } IE^{\text{naive}} 
                &=  \mathbb{E}\Bigl[Y\Bigl(\textbf{z}_{T}, \mathbfcal{G}^{\textbf{z}_{T}}_{m_{T}}\Bigr) 
                - Y\Bigl(\textbf{z}_{T}, \mathbfcal{G}^{\textbf{z}_{*,T}}_{m_{T}}\Bigr) 
                \Bigr], \text{ and }
                \\
                \text{Total Effect: } TE^{\text{naive}} 
                &=  \mathbb{E}\Bigl[Y\Bigl(\textbf{z}_{T}, \mathbfcal{G}^{\textbf{z}_{T}}_{m_{T}}\Bigr) 
                - Y\Bigl(\textbf{z}_{*,T}, \mathbfcal{G}^{\textbf{z}_{*,T}}_{m_{T}}\Bigr) \Bigr].
            \end{aligned}
        \end{equation}
In this representation, 
\begin{equation}\label{Eq::naive_joint_mediator_distn}
            \thinmuskip=0mu\relax\medmuskip=0mu\relax\thickmuskip=0mu\relax
            \begin{aligned}
                & \mathcal{G}^{\textbf{z}_{*,t}}_{(m_{t}^{(1)},  m_{t}^{(2)} | \textbf{m}_{t-1}^{(1)},\textbf{m}_{t-1}^{(2)}, \boldsymbol{\ell}_{0} )}  
                \\& 
                =  f\biggl(\Bigl(M_{t}^{(1)}\bigl(\textbf{z}_{*,t}\bigr), M_{t}^{(2)}\bigl(\textbf{z}_{*,t} \bigr)\Bigr) = \bigl(m_{t}^{(1)}, m_{t}^{(2)}\bigr)
                \big| 
                \Bigl(\textbf{M}_{t-1}^{(1)}\bigl(\textbf{z}_{*,t-1}\bigr),  \textbf{M}_{t-1}^{(J)}\bigl(\textbf{z}_{*,t-1}\bigr)\Bigr) = \bigl( \textbf{m}_{t-1}^{(1)}, \textbf{m}_{t-1}^{(2)}\bigr),L_0 = \ell_0\biggr) 
            \end{aligned} 
        \end{equation}
        denotes the joint conditional density of the counterfactual mediators, $\Big(M_{t}^{(1)}\bigl(\textbf{z}_{*,t}\bigr), $ $ M_{t}^{(J)}\bigl(\textbf{z}_{*,t}\bigr)\Big)$, in the world where the treatment assignment is set to $\textbf{Z}_{t} = \textbf{z}_{*,t}$ for $t =1,2,3$. This conditional density is defined analogously to the joint conditional density of the counterfactual mediators $\mathcal{G}^{\mathbf{z}_{*,t},\mathbf{d}_{t}(z_{*,t})}_{(m_{t}^{(1)}, \ldots, m_{t}^{(J)} \mid \mathbf{d}_t, \mathbf{m}_{t-1}^{(1)}, \ldots, \mathbf{m}_{t-1}^{(J)}, \boldsymbol{\ell}_{0} )}$  introduced in Section~\ref{sec:3.2} of the main text, except that the naive formulation does not condition on treatment receipt statuses $\mathbf{d}_t$. Now, let 
        \begin{equation}
            \begin{aligned}
                \theta^{\text{naive}}(\textbf{z}_{T},\textbf{z}_{*,T}) = \mathbb{E}\Bigl[Y\Bigl(\textbf{z}_{T}, \mathbfcal{G}^{\textbf{z}_{*,T}}_{m_{T}}\Bigr) \Bigr]
            \end{aligned}
        \end{equation}
        for the exposure regimes $\textbf{z}_{T}$ and $\textbf{z}_{*,T}$. We can write the naive causal effects in (\ref{Eq:naive_effects}) as functions of $\theta(\textbf{z}_{T},\textbf{z}_{*,T})$ as follows:
        
\begin{equation}\label{Eq::naive_effects_defn}
    \begin{aligned}
        \text{TE}^{\text{naive}} &= \theta^{\text{naive}}(\textbf{z}_{T},\textbf{z}_{T}) - \theta^{\text{naive}}(\textbf{z}_{*,T},\textbf{z}_{*,T}), \\
        \text{DE}^{\text{naive}} &= \theta^{\text{naive}}(\textbf{z}_{T},\textbf{z}_{*,T}) - \theta^{\text{naive}}(\textbf{z}_{*,T},\textbf{z}_{*,T}), \text{ and } \\
        \text{IE}^{\text{naive}} &=  \theta^{\text{naive}}(\textbf{z}_{T},\textbf{z}_{T}) - \theta^{\text{naive}}(\textbf{z}_{T},\textbf{z}_{*,T}).
    \end{aligned}
\end{equation}      

The estimation of $\theta^{\text{naive}}(\textbf{z}_{T},\textbf{z}_{*,T})$--- and consequently of the naive direct, indirect, and total effects---reduces to the standard mediation analysis problem with multiple mediators. We follow the same estimation procedure prescribed in the main text, but we do not require the additional identification assumptions (like Assumption 3, Assumption 4.1, and so on, in the main text) or the Monte Carlo steps (like Step 3 and Step 4(b) of Algorithm~\ref{Alg1:Gcomp}) that arise from modeling potential treatment receipt statuses and principal strata. In other words, the naive analysis avoids the complications introduced by counterfactual treatment receipt and principal stratification.

\subsection{Bayesian parametric model}

We fit the following parametric Bayesian model to $\big\{\mathcal{O}^{\text{naive}}_{i}\big\}_{i=1}^{n}$:
\begin{equation}\label{Eq:NaiveParamModel}
\resizebox{\textwidth}{!}{$\displaystyle
    \begin{aligned}
          & Y_{i}\big| \textbf{M}_{i,T}^{(1)}, \textbf{M}_{i,T}^{(2)}, \textbf{Z}_{i,T}, L_{i0} \sim
         \begin{cases}
             0, & \text{with probability } \pi_Y, \\[4pt]
             N\big(\boldsymbol{\gamma}_{M^{(1)}}^\top \textbf{m}^{(1)}_{i,T}
             + \boldsymbol{\gamma}_{M^{(2)}}^\top \textbf{m}^{(2)}_{i,T}
             + \boldsymbol{\gamma}_{Z}^\top \textbf{z}_{i,T}
             + \boldsymbol{\gamma}_{L_{0}}\ell_{i0}, \, \sigma^{2}_{Y}\big), & \text{with probability } 1 - \pi_Y,
         \end{cases}
         \\[6pt]
        & M_{i,t}^{(2)} \big| \textbf{M}_{i,t}^{(1)}, \textbf{M}_{i,t-1}^{(2)}, \textbf{Z}_{i,t}, L_{i0} \sim
         \begin{cases}
             0, & \text{with probability } \pi_{M_{t}^{(2)}}, \\[4pt]
             N\big(\boldsymbol{\eta}_{t,M^{(1)}}^\top \textbf{m}^{(1)}_{i,t}
             + \boldsymbol{\eta}_{t-1,M^{(2)}}^\top \textbf{m}^{(2)}_{i,t-1}
             + \boldsymbol{\eta}_{t,Z}^\top \textbf{z}_{i,t}
             + \boldsymbol{\eta}_{t,L_{0}}\ell_{i0}, \, \sigma^{2}_{M_{t}^{(2)}}\big), & \text{with probability } 1 - \pi_{M_{t}^{(2)}},
         \end{cases}
         \\[6pt]
        & M_{i,t}^{(1)} \big| \textbf{M}_{i,t-1}^{(1)}, \textbf{M}_{i,t-1}^{(2)}, \textbf{Z}_{i,t}, L_{i0} \sim
         \begin{cases}
             0, & \text{with probability } \pi_{M_{t}^{(1)}}, \\[4pt]
             N\big(\boldsymbol{\xi}_{t,M^{(1)}}^\top \textbf{m}^{(1)}_{i,t-1}
             + \boldsymbol{\xi}_{t,M^{(2)}}^\top \textbf{m}^{(2)}_{i,t-1}
             + \boldsymbol{\xi}_{t,Z}^\top \textbf{z}_{i,t}
             + \boldsymbol{\xi}_{t,L_{0}}\ell_{i0}, \, \sigma^{2}_{M_{t}^{(1)}}\big), & \text{with probability } 1 - \pi_{M_{t}^{(1)}}.
         \end{cases}
    \end{aligned}
$}
\end{equation}
for $t = 1, 2, 3$. 

 The model in~\eqref{Eq:NaiveParamModel} is fitted to the empirical dataset using \texttt{Nimble}. Parameter estimation is performed via MCMC, running 20{,}000 iterations and discarding the first 10{,}000 as burn-in. The remaining 10{,}000 iterations are thinned by a factor of 10, yielding 1{,}000 posterior samples for inference. At each retained iteration, we apply a G-computation algorithm---analogous to Algorithm~\ref{Alg1:Gcomp} but adapted to the parametric setting in the absence of treatment receipt status---using 10{,}000 Monte Carlo samples (at each MCMC iteration) to compute the principal interventional direct, indirect, and total effects in (\ref{Eq::naive_effects_defn}). This procedure yields 1{,}000 posterior draws (one per retained MCMC iteration) of the principal interventional effects, from which posterior means and 95\% Bayesian credible intervals are obtained. These are referred to as the ''naive'' effect estimates.

\section{Non-parametric identification}\label{sec:S2}

\subsection{Proof of Proposition~\ref{prop1} in the main text}
Proposition~\ref{prop1} in the main text states that for any two treatment regimes $\textbf{z}$, $\textbf{z}_{*}$, $\theta(\textbf{z},\textbf{z}_{*}) $ is identified using Assumptions 1-5 as:
\begin{equation}\label{Eq:PropGformula_NoRE}
\resizebox{0.975\textwidth}{!}{$\displaystyle
        \begin{aligned}
            \theta(\textbf{z},\textbf{z}_{*})
            &=    \int_{\ell_{i0}} \int_{\textbf{d}_{i,2:T}} \int_{\textbf{m}^{(1)}_{i,T},\ldots, \textbf{m}^{(J)}_{i,T}}  \mathbb{E}\Bigl[ Y\big|\textbf{Z}_{i,T} = \textbf{z}_{i,T}, \textbf{D}_{i,T} = \textbf{d}_{i,T}, \textbf{M}^{(1)}_{i,T} = \textbf{m}^{(1)}_{i,T}, \ldots, \textbf{M}^{(J)}_{i,T} =  \textbf{m}^{(J)}_{i,T}, L_{i0} = \ell_{i0}\Bigr] \times
            \\&
            \Bigl\{ \prod_{t=1}^{T}  f\bigl( M_{i,t}^{(1)} = m_{i,t}^{(1)},\ldots, M_{i,t}^{(J)} = m_{i,t}^{(J)}\big|\textbf{Z}_{i,t} = \textbf{z}_{*,i,t}, \textbf{D}_{i,t} = \textbf{d}_{i,t}, \textbf{M}^{(1)}_{i,t-1} =  \textbf{m}^{(1)}_{i,t-1}, \ldots, \textbf{M}^{(J)}_{i,t-1} =\textbf{m}^{(J)}_{i,t-1}, L_{i0} =\ell_{i0}
             \bigr)
            \Bigr\} \times \\&
            \Bigl\{ \prod_{s=2}^{T} f\Bigl(d_{i,s}\big|  \textbf{Z}_{i,s} = \textbf{z}_{i,s}, \textbf{D}_{i,s-1} = \textbf{d}_{i,s-1},\textbf{M}^{(1)}_{i,s-1} =  \textbf{m}^{(1)}_{i,s-1}, \ldots, \textbf{M}^{(J)}_{i,s-1} =\textbf{m}^{(J)}_{i,s-1},  L_{i0} = \ell_{i0}\Bigr)\Bigr\} \times \\&
            f_{L_{i0}}(\ell_{i0})
             d(\textbf{m}^{(1)}_{i,T},\ldots, \textbf{m}^{(J)}_{i,T}) d(\textbf{d}_{i,2:T}) d\ell_{i0}.
        \end{aligned}
$}
\end{equation}

For conciseness, let $\textbf{b}_i = \big\{b_I^{D},b_i^{M^{(1)}}, \ldots, b_i^{M^{(J)}}  \big\}$ denote the vector of random effects for subject $i$. Equation (\ref{Eq:PropGformula_NoRE}) can then be re-expressed to incorporate these random effects as follows:
\begin{equation} \label{Eq::non-paramEstimatorWithRE}
\resizebox{\textwidth}{!}{$\displaystyle
        \begin{aligned}
            \theta(\textbf{z},\textbf{z}_{*})
            &=  \int_{\textbf{b}_{i}}  \int_{\ell_{i0}} \int_{\textbf{d}_{i,2:T}} \int_{\textbf{m}^{(1)}_{i,T},\ldots, \textbf{m}^{(J)}_{i,T}}  \mathbb{E}\Bigl[ Y\big|\textbf{Z}_{i,T} = \textbf{z}_{i,T}, \textbf{D}_{i,T} = \textbf{d}_{i,T}, \textbf{M}^{(1)}_{i,T} = \textbf{m}^{(1)}_{i,T}, \ldots, \textbf{M}^{(J)}_{i,T} =  \textbf{m}^{(J)}_{i,T}, L_{i0} = \ell_{i0}\Bigr] \times
            \\&
            \Bigl\{ \prod_{t=1}^{T}  f\bigl( M_{i,t}^{(1)} = m_{i,t}^{(1)},\ldots, M_{i,t}^{(J)} = m_{i,t}^{(J)}\big|\textbf{Z}_{i,t} = \textbf{z}_{*,i,t}, \textbf{D}_{i,t} = \textbf{d}_{i,t}, \textbf{M}^{(1)}_{i,t-1} =  \textbf{m}^{(1)}_{i,t-1}, \ldots, \textbf{M}^{(J)}_{i,t-1} =\textbf{m}^{(J)}_{i,t-1}, L_{i0} =\ell_{i0},
            \\&
            b_{i}^{M^{(1)}}, \ldots, b_{i}^{M^{(J)}} \bigr)
            \Bigr\} \times
            \Bigl\{ \prod_{s=2}^{T} f\Bigl(d_{i,s}\big|  \textbf{Z}_{i,s} = \textbf{z}_{i,s}, \textbf{D}_{i,s-1} = \textbf{d}_{i,s-1},\textbf{M}^{(1)}_{i,s-1} =  \textbf{m}^{(1)}_{i,s-1}, \ldots, \textbf{M}^{(J)}_{i,s-1} =\textbf{m}^{(J)}_{i,s-1},  L_{i0} = \ell_{i0}, b_{i}^{D} \Bigr)\Bigr\} \times \\&
            f_{L_{i0}}(\ell_{i0}) \times f(\textbf{b}_i)
            \\&
            d\textbf{b}_{i} d(\textbf{m}^{(1)}_{i,T},\ldots, \textbf{m}^{(J)}_{i,T}) d(\textbf{d}_{i,2:T}) d\ell_{i0} d\textbf{b}_i.
        \end{aligned}
$}
    \end{equation}

\begin{proof}
\footnotesize
     By the definition of $\theta(\textbf{z},\textbf{z}_{*})$ , we have:
    \begin{equation}
       \scriptsize
       \thinmuskip=0mu\relax\medmuskip=0mu\relax\thickmuskip=0mu\relax
       \begin{aligned}
        \theta(\textbf{z},\textbf{z}_{*})
         &= \mathbb{E}\Bigl[Y\Bigl(\textbf{z}_{T}, \textbf{d}_{T}(z_{T}),\mathbfcal{G}^{\textbf{z}_{*,T},\textbf{d}_{T}(z_{*,T})}_{m_{T}}\Bigr)\big| U_1 = \bigl(d_{1}(1), d_{1}(0)\bigr) \Bigr]
        \\
        \\
        &=\int_{\ell_0} \int_{\textbf{d}_{2:T}}\int_{\bigl(\textbf{m}^{(1)}_T, \ldots,\textbf{m}^{(J)}_T\bigr)}  \mathbb{E}\Bigl[Y\Bigl(\textbf{z}_{T}, \textbf{d}_{T}(z_{T}),\mathbfcal{G}^{\textbf{z}_{*,T},\textbf{d}_{T}(z_{*,T})}_{m_{T}}\Bigr)\big| U_1 = \bigl(d_{1}(1), d_{1}(0)\bigr), \textbf{D}_{2:T}(z_{2:T}) = \textbf{d}_{2:T}, \\&
         \mathbfcal{G}^{\textbf{z}_{*,T},\textbf{d}_{T}(z_{*,T})}_{m^{(1)}} =  \textbf{m}_T^{(1)}, \ldots,   \mathbfcal{G}^{\textbf{z}_{*,T},\textbf{d}_{T}(z_{*,T})}_{m^{(J)}} =  \textbf{m}_T^{(J)}, L_0 = \ell_0 \Bigr]  \times \\& 
            \Bigl\{ \prod_{t=1}^{T}   f\Bigl(\mathcal{G}^{\textbf{z}_{*,t},\textbf{d}_{T}(z_{*,T})}_{m^{(1)}} =  m_t^{(1)}, \ldots, \mathcal{G}^{\textbf{z}_{*,t},\textbf{d}_{T}(z_{*,T})}_{m^{(J)}} =  m_t^{(J)} \big| U_1 = \bigl(d_{1}(1), d_{1}(0)\bigr),  \textbf{D}_{2:t}(z_{*,2:t}) = \textbf{d}_{2:t}, \\& \mathbfcal{G}^{\textbf{z}_{*,t-1},\textbf{d}_{t-1}(z_{*,t-1})}_{m^{(1)}} =  \textbf{m}_{t-1}^{(1)},  \ldots,  
            \mathbfcal{G}^{\textbf{z}_{*,t-1},\textbf{d}_{t-1}(z_{*,t-1})}_{m^{(J)}} =  \textbf{m}_{t-1}^{(J)},  
             L_0 = \ell_0 \Bigr)  \Bigr\} \times 
             \\&
             \Bigl\{ \prod_{s=2}^{T} f\Bigl(D_s(z_s) = d_{s} \big| U_1 = \bigl(d_{1}(1), d_{1}(0)\bigr),  \textbf{D}_{s-1}(z_{s-1}) = \textbf{d}_{s-1}, \\& 
            \mathbfcal{G}^{\textbf{z}_{*,s-1},\textbf{d}_{s-1}(z_{*,s-1})}_{m^{(1)}} =  \textbf{m}_{s-1}^{(1)},  \ldots, \mathbfcal{G}^{\textbf{z}_{*,s-1},\textbf{d}_{s-1}(z_{*,s-1})}_{m^{(J)}} =  \textbf{m}_{s-1}^{(J)},
             L_0 = \ell_0 \Bigr)\Bigr\} \\&
            \times f_{L_0}(\ell_0)  d\bigl(\textbf{m}^{(1)}_T, \ldots,\textbf{m}^{(J)}_T\bigr) d\textbf{d}_{2:T} d\ell_0  
    \end{aligned}
   \end{equation}
where the second equality follows from the law of total probability $\big(f(A) = \int_{B} f(A|B)f(B)\big)$ iteratively for all $t \in \{1, \ldots, T \}$ and $s \in \{2, \ldots, T \}$. Below, we demonstrate the steps for $T = 2$, and the same arguments can be used for any $T>2$ by induction. Observe that when $T=2$, the previous equation becomes:

{\allowdisplaybreaks\scriptsize
\thinmuskip=0mu\relax\medmuskip=0mu\relax\thickmuskip=0mu\relax
    \begin{align*}
        &= \int_{\ell_0} \int_{\bigl(m^{(1)}_1, \ldots, m^{(J)}_1\bigr)} \int_{d_2}\int_{\bigl(m^{(1)}_2, \ldots, m^{(J)}_2\bigr)} \\&
        \mathbb{E}\Bigl[Y\Bigl(\textbf{z}_{2}, \textbf{d}_{2}(z_{2}),\mathbfcal{G}^{\textbf{z}_{*,2},\textbf{d}_{2}(z_{*,2})}_{m_{2}}\Bigr)\big| U_1 = \bigl(d_{1}(1), d_{1}(0)\bigr),  D_{2}(z_{2}) = d_{2},  \mathbfcal{G}^{\textbf{z}_{*,2},\textbf{d}_{2}(z_{*,2})}_{m^{(1)}} =  \textbf{m}_2^{(1)}, \ldots, \mathbfcal{G}^{\textbf{z}_{*,2},\textbf{d}_{2}(z_{*,2})}_{m^{(J)}} =  \textbf{m}_2^{(J)}, L_0 = \ell_0 \Bigr]  \times 
        \\& 
       f\Bigl(\mathcal{G}^{\textbf{z}_{*,2},\textbf{d}_{2}(z_{*,2})}_{m^{(1)}} =  m_2^{(1)}, \ldots, \mathcal{G}^{\textbf{z}_{*,2},\textbf{d}_{2}(z_{*,2})}_{m^{(J)}} =  m_2^{(J)} \big|U_1 = \bigl(d_{1}(1), d_{1}(0)\bigr),  D_{2}(z_{*,2}) = d_{2},  
        \mathcal{G}^{z_{*,1},d_{1}(z_{*,1})}_{m^{(1)}} =  m_{1}^{(1)}, \ldots, \mathcal{G}^{z_{*,1},d_{1}(z_{*,1})}_{m^{(J)}} =  m_{1}^{(J)},  
        \\& 
        L_0 = \ell_0 \Bigr)  \times
        f\Bigl(D_2(z_2) = d_{2}\big| U_1 = \bigl(d_{1}(1), d_{1}(0)\bigr),  \mathbfcal{G}^{z_{*,1},d_{1}(z_{*,1})}_{m^{(1)}} =  m_{1}^{(1)}, \ldots, \mathbfcal{G}^{z_{*,1},d_{1}(z_{*,1})}_{m^{(J)}} =  m_{1}^{(J)},  L_0 = \ell_0 \Bigr) \times 
        \\& 
        f\Bigl(\mathcal{G}^{z_{*,1},\textbf{d}_{1}(z_{*,1})}_{m^{(1)}} =  m_1^{(1)}, \ldots, \mathcal{G}^{z_{*,1},\textbf{d}_{1}(z_{*,1})}_{m^{(J)}} =  m_1^{(J)} \big|U_1 = \bigl(d_{1}(1), d_{1}(0)\bigr),  L_0 = \ell_0 \Bigr) \times  
        f_{L_0}(\ell_0) 
        \\&
        d\bigl(m^{(1)}_2, \ldots, m^{(J)}_2\bigr) d(d_2) d\bigl(m^{(1)}_1, \ldots, m^{(J)}_1\bigr)  d\ell_0
        \\&
        \{\text{ Settting } \mathcal{G}^{z_{*,0},d_{0}(z_{*,0})}_{m^{(j)}} =  m_{0}^{(j)} = \emptyset \}
        \\
        \\
        &= \int_{\ell_0} \int_{\bigl(m^{(1)}_1, \ldots, m^{(J)}_1\bigr)} \int_{d_2}\int_{\bigl(m^{(1)}_2, \ldots, m^{(J)}_2\bigr)} \\&
        \mathbb{E}\Bigl[Y\Bigl(\textbf{z}_{2}, \textbf{d}_{2}(z_{2}),\mathbfcal{G}^{\textbf{z}_{*,2},\textbf{d}_{2}(z_{*,2})}_{m_{2}}\Bigr)\big| D_{1}(z_{1}), D_{1}(1-z_{1}),  D_{2}(z_{2}) = d_{2},  \mathbfcal{G}^{\textbf{z}_{*,2},\textbf{d}_{2}(z_{*,2})}_{m^{(1)}} =  \textbf{m}_2^{(1)}, \ldots, \mathbfcal{G}^{\textbf{z}_{*,2},\textbf{d}_{2}(z_{*,2})}_{m^{(J)}} =  \textbf{m}_2^{(J)}, L_0 = \ell_0 \Bigr]  \times 
        \\& 
       f\Bigl(\mathcal{G}^{\textbf{z}_{*,2},\textbf{d}_{2}(z_{*,2})}_{m^{(1)}} =  m_2^{(1)}, \ldots, \mathcal{G}^{\textbf{z}_{*,2},\textbf{d}_{2}(z_{*,2})}_{m^{(J)}} =  m_2^{(J)} \big|D_{1}(z_{*,1}), D_{1}(1-z_{*,1}),  D_{2}(z_{*,2}) = d_{2},  
        \mathcal{G}^{z_{*,1},d_{1}(z_{*,1})}_{m^{(1)}} =  m_{1}^{(1)}, \ldots, \mathcal{G}^{z_{*,1},d_{1}(z_{*,1})}_{m^{(J)}} =  m_{1}^{(J)},  
        \\& 
        L_0 = \ell_0 \Bigr)  \times
        f\Bigl(D_2(z_2) = d_{2}\big| D_{1}(z_{1}), D_{1}(1-z_{1}),  \mathbfcal{G}^{z_{*,1},d_{1}(z_{*,1})}_{m^{(1)}} =  m_{1}^{(1)}, \ldots, \mathbfcal{G}^{z_{*,1},d_{1}(z_{*,1})}_{m^{(J)}} =  m_{1}^{(J)},  L_0 = \ell_0 \Bigr) \times 
        \\& 
        f\Bigl(\mathcal{G}^{z_{*,1},\textbf{d}_{1}(z_{*,1})}_{m^{(1)}} =  m_1^{(1)}, \ldots, \mathcal{G}^{z_{*,1},\textbf{d}_{1}(z_{*,1})}_{m^{(J)}} =  m_1^{(J)} \big|D_{1}(z_{*,1}), D_{1}(1-z_{*,1}),  L_0 = \ell_0 \Bigr) \times  
        f_{L_0}(\ell_0) 
        \\&
        d\bigl(m^{(1)}_2, \ldots, m^{(J)}_2\bigr) d(d_2) d\bigl(m^{(1)}_1, \ldots, m^{(J)}_1\bigr)  d\ell_0  \\&
        \Bigl\{ \bigl\{ D_1(1), D_1(0) \bigr\} \text{ is equivalent to either } \bigl\{ D_1(z_{1}), D_1(1-z_{1}) \bigr\}
        \\&
        \text{ or } \bigl\{ D_1(z_{*,1}), D_t(1-z_{*,1}) \bigr\}  \Bigr\}
        \\
        \\
        &= \int_{\ell_0} \int_{\bigl(m^{(1)}_1, \ldots, m^{(J)}_1\bigr)} \int_{d_2}\int_{\bigl(m^{(1)}_2, \ldots, m^{(J)}_2\bigr)} \\&
        \mathbb{E}\Bigl[Y\Bigl(\textbf{z}_{2}, \textbf{d}_{2}(z_{2}),\mathbfcal{G}^{\textbf{z}_{*,2},\textbf{d}_{2}(z_{*,2})}_{m_{2}}\Bigr)\big| D_{1}(z_{1}),   D_{2}(z_{2}) = d_{2},  \mathbfcal{G}^{\textbf{z}_{*,2},\textbf{d}_{2}(z_{*,2})}_{m^{(1)}} =  \textbf{m}_2^{(1)}, \ldots, \mathbfcal{G}^{\textbf{z}_{*,2},\textbf{d}_{2}(z_{*,2})}_{m^{(J)}} =  \textbf{m}_2^{(J)}, L_0 = \ell_0 \Bigr]  \times 
        \\& 
       f\Bigl(\mathcal{G}^{\textbf{z}_{*,2},\textbf{d}_{2}(z_{*,2})}_{m^{(1)}} =  m_2^{(1)}, \ldots, \mathcal{G}^{\textbf{z}_{*,2},\textbf{d}_{2}(z_{*,2})}_{m^{(J)}} =  m_2^{(J)} \big|D_{1}(z_{*,1}),   D_{2}(z_{*,2}) = d_{2},  
        \mathcal{G}^{z_{*,1},d_{1}(z_{*,1})}_{m^{(1)}} =  m_{1}^{(1)}, \ldots, \mathcal{G}^{z_{*,1},d_{1}(z_{*,1})}_{m^{(J)}} =  m_{1}^{(J)},  
        \\& 
        L_0 = \ell_0 \Bigr)  \times
        f\Bigl(D_2(z_2) = d_{2}\big| D_{1}(z_{1}),    \mathbfcal{G}^{z_{*,1},d_{1}(z_{*,1})}_{m^{(1)}} =  m_{1}^{(1)}, \ldots, \mathbfcal{G}^{z_{*,1},d_{1}(z_{*,1})}_{m^{(J)}} =  m_{1}^{(J)},  L_0 = \ell_0 \Bigr) \times 
        \\& 
        f\Bigl(\mathcal{G}^{z_{*,1},\textbf{d}_{1}(z_{*,1})}_{m^{(1)}} =  m_1^{(1)}, \ldots, \mathcal{G}^{z_{*,1},\textbf{d}_{1}(z_{*,1})}_{m^{(J)}} =  m_1^{(J)} \big|D_{1}(z_{*,1}),   L_0 = \ell_0 \Bigr) \times  
        f_{L_0}(\ell_0) 
        \\&
        d\bigl(m^{(1)}_2, \ldots, m^{(J)}_2\bigr) d(d_2) d\bigl(m^{(1)}_1, \ldots, m^{(J)}_1\bigr)  d\ell_0  \quad \quad \quad 
        \Bigl\{  \text{ by \textbf{Assumption~\ref{assump::3} in the main text} }   \Bigr\}
        \\
        \\
        & =  \int_{\ell_0} \int_{\bigl(m^{(1)}_1, \ldots, m^{(J)}_1\bigr)} \int_{d_2}\int_{\bigl(m^{(1)}_2, \ldots, m^{(J)}_2\bigr)} \\& 
        \mathbb{E}\Bigl[Y\Bigl(\textbf{z}_{2}, \textbf{d}_{2}(z_{2}),\textbf{m}^{(1)}_{2}, \ldots, \textbf{m}^{(J)}_{2}\Bigr)\big| D_{1}(z_{1}),  D_{2}(z_{2}) = d_{2},   \mathbfcal{G}^{\textbf{z}_{*,2},\textbf{d}_{2}(z_{*,2})}_{m^{(1)}} =  \textbf{m}_2^{(1)}, \ldots, \mathbfcal{G}^{\textbf{z}_{*,2},\textbf{d}_{2}(z_{*,2})}_{m^{(J)}} =  \textbf{m}_2^{(J)}, L_0 = \ell_0 \Bigr]  \times 
        \\& 
       f\Bigl(M_{2}^{(1)}\bigl(\textbf{z}_{*,2}, \textbf{d}_2(z_{*,2})\bigr) = m_{2}^{(1)}, \ldots, M_{2}^{(J)}\bigl(\textbf{z}_{*,2}, \textbf{d}_2(z_{*,2})\bigr) = m_{2}^{(J)} \big|  D_{1}(z_{*,1}),  D_{2}(z_{*,2}) = d_{2}, \\&  
        M_{1}^{(1)}\bigl(z_{*,1}, d_{1}(z_{*,1})\bigr) = m_{1}^{(1)}, \ldots,  M_{1}^{(J)}\bigl(z_{*,1}, d_{1}(z_{*,1})\bigr) = m_{1}^{(J)},  
        L_0 = \ell_0 \Bigr)  \times
        \\& 
        f\Bigl(D_2(z_2) = d_{2} \big|  D_{1}(z_{1}),  \mathbfcal{G}^{z_{*,1},d_{1}(z_{*,1})}_{m^{(1)}} =  m_{1}^{(1)}, \ldots, \mathbfcal{G}^{z_{*,1},d_{1}(z_{*,1})}_{m^{(J)}} =  m_{1}^{(J)},  L_0 = \ell_0 \Bigr) \times 
        \\& 
        f\Bigl(M_{1}^{(1)}\bigl(z_{*,1}, d_{1}(z_{*,1})\bigr) = m_{1}^{(1)}, \ldots, M_{1}^{(J)}\bigl(z_{*,1}, d_{1}(z_{*,1})\bigr) = m_{1}^{(J)} \big| D_{1}(z_{*,1}),     L_0 = \ell_0 \Bigr) \times 
        f_{L_0}(\ell_0)  
        \\&
        d\bigl(m^{(1)}_2, \ldots, m^{(J)}_2\bigr) d(d_2) d\bigl(m^{(1)}_1, \ldots, m^{(J)}_1\bigr)  d\ell_0    
        \quad \quad \quad \Big\{\text{ by the definition of  } \mathbfcal{G}^{\textbf{z}_{*,t},\textbf{d}_{t}(\textbf{z}_{*,t})}_{m}  \Big\}
        \\
        \\
        & =  \int_{\ell_0} \int_{\bigl(m^{(1)}_1, \ldots, m^{(J)}_1\bigr)} \int_{d_2}\int_{\bigl(m^{(1)}_2, \ldots, m^{(J)}_2\bigr)} 
        \mathbb{E}\Bigl[Y\Bigl(\textbf{z}_{2}, \textbf{d}_{2}(z_{2}),\textbf{m}^{(1)}_{2}, \ldots, \textbf{m}^{(J)}_{2}\Bigr)\big| D_{1}(z_{1}),  D_{2}(z_{2}) = d_{2}, L_0 = \ell_0 \Bigr]  \times 
        \\& 
       f\Bigl(M_{2}^{(1)}\bigl(\textbf{z}_{*,2}, \textbf{d}_2(z_{*,2})\bigr) = m_{2}^{(1)}, \ldots, M_{2}^{(J)}\bigl(\textbf{z}_{*,2}, \textbf{d}_2(z_{*,2})\bigr) = m_{2}^{(J)} \big|  D_{1}(z_{*,1}),  D_{2}(z_{*,2}) = d_{2}, \\&  
        M_{1}^{(1)}\bigl(z_{*,1}, d_{1}(z_{*,1})\bigr) = m_{1}^{(1)}, \ldots,  M_{1}^{(J)}\bigl(z_{*,1}, d_{1}(z_{*,1})\bigr) = m_{1}^{(J)},  
        L_0 = \ell_0 \Bigr)  \times
        \\& 
        f\Bigl(D_2(z_2) = d_{2} \big|  D_{1}(z_{1}),   L_0 = \ell_0 \Bigr) \times 
        \\& 
        f\Bigl(M_{1}^{(1)}\bigl(z_{*,1}, d_{1}(z_{*,1})\bigr) = m_{1}^{(1)}, \ldots, M_{1}^{(J)}\bigl(z_{*,1}, d_{1}(z_{*,1})\bigr) = m_{1}^{(J)} \big| D_{1}(z_{*,1}),     L_0 = \ell_0 \Bigr) \times 
        f_{L_0}(\ell_0) 
        \\&
        d\bigl(m^{(1)}_2, \ldots, m^{(J)}_2\bigr) d(d_2) d\bigl(m^{(1)}_1, \ldots, m^{(J)}_1\bigr)  d\ell_0
        \\&
        \Bigl\{ \text{conditional on the past covariates },  \textbf{m}_{t}\text{ is a random draw }  \text{from } \mathcal{G}^{\textbf{z}_{*,t},\textbf{d}_{t}(z_{*,t})}_{m}
        \text{(randomized intervention) and does not affect }  \\& \text{potential treatment receipt } D_{t}(z_{t}) \text{ and potential outcome }Y\Bigl(\textbf{z}_{T}, \textbf{d}_{T}(z_{T}),\textbf{m}^{(1)}_{T}\bigl(\textbf{z}_{*,T},\textbf{d}_{T}(z_{*,T})\bigr), \ldots, \textbf{m}^{(J)}_{T}\bigl(\textbf{z}_{*,T},\textbf{d}_{T}(z_{*,T})\bigr)\Bigr)\Bigr\}
        \\
        \\
        & =  \int_{\ell_0} \int_{\bigl(m^{(1)}_1, \ldots, m^{(J)}_1\bigr)} \int_{d_2}\int_{\bigl(m^{(1)}_2, \ldots, m^{(J)}_2\bigr)} 
        \mathbb{E}\Bigl[Y\Bigl(\textbf{z}_{2}, \textbf{d}_{2}(z_{2}),\textbf{m}^{(1)}_{2}, \ldots, \textbf{m}^{(J)}_{2}\Bigr)\big| Z_1 = z_{1}, D_{1}(z_{1}),  D_{2}(z_{2}) = d_{2}, L_0 = \ell_0 \Bigr]  \times 
        \\& 
       f\Bigl(M_{2}^{(1)}\bigl(\textbf{z}_{*,2}, \textbf{d}_2(z_{*,2})\bigr) = m_{2}^{(1)}, \ldots, M_{2}^{(J)}\bigl(\textbf{z}_{*,2}, \textbf{d}_2(z_{*,2})\bigr) = m_{2}^{(J)} \big|  Z_1 = z_{*,1}, D_{1}(z_{*,1}),  D_{2}(z_{*,2}) = d_{2}, \\&  
        M_{1}^{(1)}\bigl(z_{*,1}, d_{1}(z_{*,1})\bigr) = m_{1}^{(1)}, \ldots,  M_{1}^{(J)}\bigl(z_{*,1}, d_{1}(z_{*,1})\bigr) = m_{1}^{(J)},  
        L_0 = \ell_0 \Bigr)  \times
        \\& 
        f\Bigl(D_2(z_2) = d_{2} \big| Z_1 = z_{1}, D_{1}(z_{1}),   L_0 = \ell_0 \Bigr) \times 
        \\& 
        f\Bigl(M_{1}^{(1)}\bigl(z_{*,1}, d_{1}(z_{*,1})\bigr) = m_{1}^{(1)}, \ldots, M_{1}^{(J)}\bigl(z_{*,1}, d_{1}(z_{*,1})\bigr) = m_{1}^{(J)} \big| Z_1 = z_{*,1},D_{1}(z_{*,1}),     L_0 = \ell_0 \Bigr) \times 
        f_{L_0}(\ell_0)   
        \\&
        d\bigl(m^{(1)}_2, \ldots, m^{(J)}_2\bigr) d(d_2) d\bigl(m^{(1)}_1, \ldots, m^{(J)}_1\bigr)  d\ell_0 \quad \quad \quad
        \Bigl\{ \text{by \textbf{Assumption 2}}\Bigr\}
        \\
        \\
        & =  \int_{\ell_0} \int_{\bigl(m^{(1)}_1, \ldots, m^{(J)}_1\bigr)} \int_{d_2}\int_{\bigl(m^{(1)}_2, \ldots, m^{(J)}_2\bigr)} \\& 
        \mathbb{E}\Bigl[Y\Bigl(\textbf{z}_{2}, \textbf{d}_{2}(z_{2}),\textbf{m}^{(1)}_{2}, \ldots, \textbf{m}^{(J)}_{2}\Bigr)\big|Z_1 = z_{1}, D_{1}= d_{1},  D_{2}(z_{2}) = d_{2},   L_0 = \ell_0 \Bigr]  \times 
        \\& 
       f\Bigl(M_{2}^{(1)}\bigl(\textbf{z}_{*,2}, \textbf{d}_2(z_{*,2})\bigr) = m_{2}^{(1)}, \ldots, M_{2}^{(J)}\bigl(\textbf{z}_{*,2}, \textbf{d}_2(z_{*,2})\bigr) = m_{2}^{(J)} \big| Z_1 = z_{*,1}, D_{1} = d_{1},  D_{2}(z_{*,2}) = d_{2}, \\&  
        M_{1}^{(1)} = m_{1}^{(1)}, \ldots,  M_{1}^{(J)} = m_{1}^{(J)},  
        L_0 = \ell_0 \Bigr)  \times
        \\& 
        f\Bigl(D_2(z_2) = d_{2} \big|Z_1 = z_{1},  D_{1}= d_{1},  L_0 = \ell_0 \Bigr) \times  
        f\Bigl(M_{1}^{(1)} = m_{1}^{(1)}, \ldots, M_{1}^{(J)} = m_{1}^{(J)} \big|Z_1 = z_{*,1}, D_{1}= d_{1},     L_0 = \ell_0 \Bigr) \times 
        f_{L_0}(\ell_0)  
        \\&
        d\bigl(m^{(1)}_2, \ldots, m^{(J)}_2\bigr) d(d_2) d\bigl(m^{(1)}_1, \ldots, m^{(J)}_1\bigr)  d\ell_0  
        \quad \quad
        \Bigl\{  \text{ by \textbf{Consistency (i),(ii)} }   \Bigr\}
        \\
        \\
        & =  \int_{\ell_0} \int_{\bigl(m^{(1)}_1, \ldots, m^{(J)}_1\bigr)} \int_{d_2}\int_{\bigl(m^{(1)}_2, \ldots, m^{(J)}_2\bigr)} \\& 
        \mathbb{E}\Bigl[Y\Bigl(\textbf{z}_{2}, \textbf{d}_{2}(z_{2}),\textbf{m}^{(1)}_{2}, \ldots, \textbf{m}^{(J)}_{2}\Bigr)\big|Z_1 = z_{1}, D_{1}= d_{1},  D_{2}(z_{2}) = d_{2}, M_{1}^{(1)} = m_{1}^{(1)}, \ldots, M_{1}^{(J)} = m_{1}^{(J)},  L_0 = \ell_0 \Bigr]  \times 
        \\& 
       f\Bigl(M_{2}^{(1)}\bigl(\textbf{z}_{*,2}, \textbf{d}_2(z_{*,2})\bigr) = m_{2}^{(1)}, \ldots, M_{2}^{(J)}\bigl(\textbf{z}_{*,2}, \textbf{d}_2(z_{*,2})\bigr) = m_{2}^{(J)} \big| Z_1 = z_{*,1}, D_{1} = d_{1},  D_{2}(z_{*,2}) = d_{2}, \\&  
        M_{1}^{(1)} = m_{1}^{(1)}, \ldots,  M_{1}^{(J)} = m_{1}^{(J)},  
        L_0 = \ell_0 \Bigr)  \times
        f\Bigl(D_2(z_2) = d_{2} \big|Z_1 = z_{1},  D_{1}= d_{1},  M_{1}^{(1)} = m_{1}^{(1)}, \ldots, M_{1}^{(J)} = m_{1}^{(J)},L_0 = \ell_0 \Bigr) \times  
        \\& 
        f\Bigl(M_{1}^{(1)} = m_{1}^{(1)}, \ldots, M_{1}^{(J)} = m_{1}^{(J)} \big|Z_1 = z_{*,1}, D_{1}= d_{1},    L_0 = \ell_0 \Bigr) \times 
        f_{L_0}(\ell_0)  
        \\&
        d\bigl(m^{(1)}_2, \ldots, m^{(J)}_2\bigr) d(d_2) d\bigl(m^{(1)}_1, \ldots, m^{(J)}_1\bigr)  d\ell_0  
        \quad \quad
        \Bigl\{ \text{by \textbf{Assumption 4}  }\Bigr\}
        \\
        \\
        & =  \int_{\ell_0} \int_{\bigl(m^{(1)}_1, \ldots, m^{(J)}_1\bigr)} \int_{d_2}\int_{\bigl(m^{(1)}_2, \ldots, m^{(J)}_2\bigr)} \\& 
        \mathbb{E}\Bigl[Y\Bigl(\textbf{z}_{2}, \textbf{d}_{2}(z_{2}),\textbf{m}^{(1)}_{2}, \ldots, \textbf{m}^{(J)}_{2}\Bigr)\big|\textbf{Z}_2 = \textbf{z}_{2}, D_{1}= d_{1},  D_{2}(z_{2}) = d_{2}, M_{1}^{(1)} = m_{1}^{(1)}, \ldots, M_{1}^{(J)} = m_{1}^{(J)},  L_0 = \ell_0 \Bigr]  \times 
        \\& 
       f\Bigl(M_{2}^{(1)}\bigl(\textbf{z}_{*,2}, \textbf{d}_2(z_{*,2})\bigr) = m_{2}^{(1)}, \ldots, M_{2}^{(J)}\bigl(\textbf{z}_{*,2}, \textbf{d}_2(z_{*,2})\bigr) = m_{2}^{(J)} \big| \textbf{Z}_2 = \textbf{z}_{*,2}, D_{1} = d_{1},  D_{2}(z_{*,2}) = d_{2}, \\&  
        M_{1}^{(1)} = m_{1}^{(1)}, \ldots,  M_{1}^{(J)} = m_{1}^{(J)},  
        L_0 = \ell_0 \Bigr)  \times
        f\Bigl(D_2(z_2) = d_{2} \big|\textbf{Z}_2 = \textbf{z}_{2},  D_{1}= d_{1},  M_{1}^{(1)} = m_{1}^{(1)}, \ldots, M_{1}^{(J)} = m_{1}^{(J)},L_0 = \ell_0 \Bigr) \times  
        \\& 
        f\Bigl(M_{1}^{(1)} = m_{1}^{(1)}, \ldots, M_{1}^{(J)} = m_{1}^{(J)} \big|Z_1 = z_{*,1}, D_{1}= d_{1},    L_0 = \ell_0 \Bigr) \times 
        f_{L_0}(\ell_0)  
        \\&
        d\bigl(m^{(1)}_2, \ldots, m^{(J)}_2\bigr) d(d_2) d\bigl(m^{(1)}_1, \ldots, m^{(J)}_1\bigr)  d\ell_0  
        \quad \quad
        \Bigl\{ \text{by \textbf{Assumption 2}  }\Bigr\}
        \\
        \\
        & =  \int_{\ell_0} \int_{\bigl(m^{(1)}_1, \ldots, m^{(J)}_1\bigr)} \int_{d_2}\int_{\bigl(m^{(1)}_2, \ldots, m^{(J)}_2\bigr)} \\& 
        \mathbb{E}\Bigl[Y\Bigl(\textbf{z}_{2}, \textbf{d}_{2}(z_{2}),\textbf{m}^{(1)}_{2}, \ldots, \textbf{m}^{(J)}_{2}\Bigr)\big|\textbf{Z}_2 = \textbf{z}_{2}, \textbf{D}_{2} = \textbf{d}_{2},  M_{1}^{(1)} = m_{1}^{(1)}, \ldots, M_{1}^{(J)} = m_{1}^{(J)},  L_0 = \ell_0 \Bigr]  \times 
        \\& 
       f\Bigl(M_{2}^{(1)} = m_{2}^{(1)}, \ldots, M_{2}^{(J)}= m_{2}^{(J)} \big| \textbf{Z}_2 = \textbf{z}_{*,2}, \textbf{D}_{2} = \textbf{d}_{2},    
        M_{1}^{(1)} = m_{1}^{(1)}, \ldots,  M_{1}^{(J)} = m_{1}^{(J)},  
        L_0 = \ell_0 \Bigr)  \times 
        \\&
        f\Bigl(D_2 = d_{2} \big|\textbf{Z}_2 = \textbf{z}_{2},  D_{1}= d_{1},  M_{1}^{(1)} = m_{1}^{(1)}, \ldots, M_{1}^{(J)} = m_{1}^{(J)},L_0 = \ell_0 \Bigr) \times  
        \\& 
        f\Bigl(M_{1}^{(1)} = m_{1}^{(1)}, \ldots, M_{1}^{(J)} = m_{1}^{(J)} \big|Z_1 = z_{*,1}, D_{1}= d_{1},    L_0 = \ell_0 \Bigr) \times 
        f_{L_0}(\ell_0)  
        \\&
        d\bigl(m^{(1)}_2, \ldots, m^{(J)}_2\bigr) d(d_2) d\bigl(m^{(1)}_1, \ldots, m^{(J)}_1\bigr)  d\ell_0  
        \quad \quad
        \Bigl\{ \text{by \textbf{Consistency (i),(ii)}  } \Bigr\}
        \\
        \\
        & =  \int_{\ell_0} \int_{\bigl(m^{(1)}_1, \ldots, m^{(J)}_1\bigr)} \int_{d_2}\int_{\bigl(m^{(1)}_2, \ldots, m^{(J)}_2\bigr)} \\& 
        \mathbb{E}\Bigl[Y\Bigl(\textbf{z}_{2}, \textbf{d}_{2}(z_{2}),\textbf{m}^{(1)}_{2}, \ldots, \textbf{m}^{(J)}_{2}\Bigr)\big|\textbf{Z}_2 = \textbf{z}_{2}, \textbf{D}_{2} = \textbf{d}_{2},  \textbf{M}_{2}^{(1)} = \textbf{m}_{2}^{(1)}, \ldots, \textbf{M}_{2}^{(J)} = \textbf{m}_{2}^{(J)},  L_0 = \ell_0 \Bigr]  \times 
        \\& 
       f\Bigl(M_{2}^{(1)} = m_{2}^{(1)}, \ldots, M_{2}^{(J)}= m_{2}^{(J)} \big| \textbf{Z}_2 = \textbf{z}_{*,2}, \textbf{D}_{2} = \textbf{d}_{2},    
        M_{1}^{(1)} = m_{1}^{(1)}, \ldots,  M_{1}^{(J)} = m_{1}^{(J)},  
        L_0 = \ell_0 \Bigr)  \times 
        \\&
        f\Bigl(D_2 = d_{2} \big|\textbf{Z}_2 = \textbf{z}_{2},  D_{1}= d_{1},  M_{1}^{(1)} = m_{1}^{(1)}, \ldots, M_{1}^{(J)} = m_{1}^{(J)},L_0 = \ell_0 \Bigr) \times  
        \\& 
        f\Bigl(M_{1}^{(1)} = m_{1}^{(1)}, \ldots, M_{1}^{(J)} = m_{1}^{(J)} \big|Z_1 = z_{*,1}, D_{1}= d_{1},    L_0 = \ell_0 \Bigr) \times 
        f_{L_0}(\ell_0)  
        \\&
        d\bigl(m^{(1)}_2, \ldots, m^{(J)}_2\bigr) d(d_2) d\bigl(m^{(1)}_1, \ldots, m^{(J)}_1\bigr)  d\ell_0  
        \quad \quad
        \Bigl\{ \text{by \textbf{Assumption 4 b)}  } \Bigr\}
        \\
        \\
        & =  \int_{\ell_0} \int_{\bigl(m^{(1)}_1, \ldots, m^{(J)}_1\bigr)} \int_{d_2}\int_{\bigl(m^{(1)}_2, \ldots, m^{(J)}_2\bigr)} \\& 
        \mathbb{E}\Bigl[Y\big|\textbf{Z}_2 = \textbf{z}_{2}, \textbf{D}_{2} = \textbf{d}_{2},  \textbf{M}_{2}^{(1)} = \textbf{m}_{2}^{(1)}, \ldots, \textbf{M}_{2}^{(J)} = \textbf{m}_{2}^{(J)},  L_0 = \ell_0 \Bigr]  \times 
        \\& 
       f\Bigl(M_{2}^{(1)} = m_{2}^{(1)}, \ldots, M_{2}^{(J)}= m_{2}^{(J)} \big| \textbf{Z}_2 = \textbf{z}_{*,2}, \textbf{D}_{2} = \textbf{d}_{2},    
        M_{1}^{(1)} = m_{1}^{(1)}, \ldots,  M_{1}^{(J)} = m_{1}^{(J)},  
        L_0 = \ell_0 \Bigr)  \times 
        \\&
        f\Bigl(D_2 = d_{2} \big|\textbf{Z}_2 = \textbf{z}_{2},  D_{1}= d_{1},  M_{1}^{(1)} = m_{1}^{(1)}, \ldots, M_{1}^{(J)} = m_{1}^{(J)},L_0 = \ell_0 \Bigr) \times  
        \\& 
        f\Bigl(M_{1}^{(1)} = m_{1}^{(1)}, \ldots, M_{1}^{(J)} = m_{1}^{(J)} \big|Z_1 = z_{*,1}, D_{1}= d_{1},    L_0 = \ell_0 \Bigr) \times 
        f_{L_0}(\ell_0)  
        \\&
        d\bigl(m^{(1)}_2, \ldots, m^{(J)}_2\bigr) d(d_2) d\bigl(m^{(1)}_1, \ldots, m^{(J)}_1\bigr)  d\ell_0  
        \quad \quad
        \Bigl\{ \text{by \textbf{Consistency iii)}  } \Bigr\}
\end{align*}}

For $T= 2$, we have shown that:
{\scriptsize
\thinmuskip=0mu\relax\medmuskip=0mu\relax\thickmuskip=0mu\relax
\begin{align*}
    \theta(z,z_{*}) 
        &= \int_{\ell_0} \int_{\bigl(m^{(1)}_1, \ldots, m^{(J)}_1\bigr)} \int_{d_2}\int_{\bigl(m^{(1)}_2, \ldots, m^{(J)}_2\bigr)} \\& 
        \mathbb{E}\Bigl[Y\big|\textbf{Z}_2 = \textbf{z}_{2}, \textbf{D}_{2} = \textbf{d}_{2},  \textbf{M}_{2}^{(1)} = \textbf{m}_{2}^{(1)}, \ldots, \textbf{M}_{2}^{(J)} = \textbf{m}_{2}^{(J)},  L_0 = \ell_0 \Bigr]  \times 
        \\& 
       f\Bigl(M_{2}^{(1)} = m_{2}^{(1)}, \ldots, M_{2}^{(J)}= m_{2}^{(J)} \big| \textbf{Z}_2 = \textbf{z}_{*,2}, \textbf{D}_{2} = \textbf{d}_{2},    
        M_{1}^{(1)} = m_{1}^{(1)}, \ldots,  M_{1}^{(J)} = m_{1}^{(J)},  
        L_0 = \ell_0 \Bigr)  \times 
        \\&
        f\Bigl(D_2 = d_{2} \big|\textbf{Z}_2 = \textbf{z}_{2},  D_{1}= d_{1},  M_{1}^{(1)} = m_{1}^{(1)}, \ldots, M_{1}^{(J)} = m_{1}^{(J)},L_0 = \ell_0 \Bigr) \times  
        \\& 
        f\Bigl(M_{1}^{(1)} = m_{1}^{(1)}, \ldots, M_{1}^{(J)} = m_{1}^{(J)} \big|Z_1 = z_{*,1}, D_{1}= d_{1},    L_0 = \ell_0 \Bigr) \times 
        f_{L_0}(\ell_0)  
        \\&
        d\bigl(m^{(1)}_2, \ldots, m^{(J)}_2\bigr) d(d_2) d\bigl(m^{(1)}_1, \ldots, m^{(J)}_1\bigr)  d\ell_0 .
\end{align*}}
Using the total law of probability again with random intercepts, we get:
{\scriptsize
\thinmuskip=0mu\relax\medmuskip=0mu\relax\thickmuskip=0mu\relax
\begin{align*}
    \theta(z,z_{*}) 
        &= \int_{\textbf{b}} \int_{\ell_0} \int_{\bigl(m^{(1)}_1, \ldots, m^{(J)}_1\bigr)} \int_{d_2}\int_{\bigl(m^{(1)}_2, \ldots, m^{(J)}_2\bigr)}   \\& 
        \mathbb{E}\Bigl[Y\big|\textbf{Z}_2 = \textbf{z}_{2}, \textbf{D}_{2} = \textbf{d}_{2},  \textbf{M}_{2}^{(1)} = \textbf{m}_{2}^{(1)}, \ldots, \textbf{M}_{2}^{(J)} = \textbf{m}_{2}^{(J)},  L_0 = \ell_0 \Bigr]  \times 
        \\& 
       f\Bigl(M_{2}^{(1)} = m_{2}^{(1)}, \ldots, M_{2}^{(J)}= m_{2}^{(J)} \big| \textbf{Z}_2 = \textbf{z}_{*,2}, \textbf{D}_{2} = \textbf{d}_{2},    
        M_{1}^{(1)} = m_{1}^{(1)}, \ldots,  M_{1}^{(J)} = m_{1}^{(J)},  
        L_0 = \ell_0, b^{M^{(1)}}, \ldots, b^{M^{(J)}} \Bigr)  \times 
        \\&
        f\Bigl(D_2 = d_{2} \big|\textbf{Z}_2 = \textbf{z}_{2},  D_{1}= d_{1},  M_{1}^{(1)} = m_{1}^{(1)}, \ldots, M_{1}^{(J)} = m_{1}^{(J)},L_0 = \ell_0, b^{D} \Bigr) \times  
        \\& 
        f\Bigl(M_{1}^{(1)} = m_{1}^{(1)}, \ldots, M_{1}^{(J)} = m_{1}^{(J)} \big|Z_1 = z_{*,1}, D_{1}= d_{1},    L_0 = \ell_0, b^{M^{(1)}}, \ldots, b^{M^{(J)}} \Bigr) \times 
        f_{L_0}(\ell_0)  \times f(\textbf{b})
        \\&
         d\bigl(m^{(1)}_2, \ldots, m^{(J)}_2\bigr) d(d_2) d\bigl(m^{(1)}_1, \ldots, m^{(J)}_1\bigr)  d\ell_0 d\textbf{b}.
\end{align*}}
Hence, the result follows.

\end{proof}

\section{Sensitivity analysis for the violation of ignorability of cross-world treatment receipt (Assumption 3.3)}\label{sec:S3}

In this section, we define $\textbf{m}_{T} = \big\{m^{(1)}_{t}, \ldots, m^{(J)}_{t} \big\}_{t = 1}^{T}$ for notational convenience. Recall that Assumption 3 states the following:
    \begin{enumerate}\label{assump::3sen}
        \item $D_{s}(z_{s}) \indep  D_1(1-z_{1}) \big|  D_1(z_{1}),  \textbf{D}_{s-1}(z_{s-1}), L_0$,
        \item $\mathcal{G}^{z_{*,w},d_{w}(z_{*,w})}_{m_{w}}   \indep   D_1(1-z_{*,1})  \big|  D_t(z_{*,1}),  \textbf{D}_{2:w}(z_{*,2:w}), \mathbfcal{G}^{z_{*,w-1},d_{w-1}(z_{*,w-1})}_{m_{w-1}},  L_0$, and 
        \item $Y(\textbf{z}_{T},\textbf{d}_{T}, \textbf{m}_{T})  \indep  D_1(1-z_{1})  \big|  D_1(z_{1}),\textbf{D}_{2:T}(z_{2:T}),\mathbfcal{G}^{\textbf{z}_{*,T},\textbf{d}_{T}(z_{*,T})}_{m_{T}}, L_0.$
    \end{enumerate}
where $z_{1},z_{*,1} \in \{0,1\}$  and for all $w \geq 1, s \geq 2$.
In this section, we evaluate the robustness of the estimated principal causal effects to violations of Assumption 3.3, assuming all other identification assumptions hold. We note that our proposed sensitivity analysis framework (adapted from \cite{daniels2023bayesian}) is not a comprehensive technique for assessing the full extent of the impact of violating Assumption 3 on the estimated causal effects. Given the nature of the nonparametric estimator of $\theta(z, z_{*})$ in (\ref{Eq::non-paramEstimatorWithRE}) and the dimensionality of the longitudinal variables in our setting, it is challenging to develop a complete sensitivity analysis framework---this is beyond the scope of this paper. Instead, our proposed framework is intended to provide a simple assessment of the potential impact of violating the cross-world ignorability assumption on our causal effect estimates.

\subsection{Sensitivity parameters}
Suppose that all identification assumptions hold except for Assumption 3.3. Our sensitivity analysis is based on the following parameters measuring departure from Assumption 3.3: 
\begin{center}
\resizebox{\textwidth}{!}{$\displaystyle
\begin{aligned}
    \rho_Y\big(\mathbb{1}_{T},\mathbb{1}_{T} \big)
    &= E\Big[Y\Bigl(\mathbb{1}_{T}, \textbf{d}_{T},\textbf{m}_{T}\Bigr) \big|D_1(1) = d_{1}, D_1(0) = d_{*,1}, \textbf{D}_{2:T}(\mathbb{1}_{2:T}) = \textbf{d}_{2:T}, \mathbfcal{G}^{\mathbb{1}_{T},\textbf{d}_{T}(\mathbb{1}_{T})}_{m_{T}} = \textbf{m}_T, \textbf{Z}_{T} =\mathbb{1}_{T},  L_0 = \ell_0 \Big]
    \\& -
    E\Big[Y\Bigl(\mathbb{1}_{T}, \textbf{d}_{T},\textbf{m}_{T}\Bigr) \big|  \textbf{D}_T(\mathbb{1}_{T}) = \textbf{d}_T, \mathbfcal{G}^{\mathbb{1}_{T},\textbf{d}_{T}(\mathbb{1}_{T})}_{m_{T}} = \textbf{m}_T, \textbf{Z}_{T} =\mathbb{1}_{T},  L_0 = \ell_0 \Big],
    \\
    \rho_Y\big(\mathbb{0}_{T},\mathbb{0}_{T} \big)
    &= E\Big[Y\Bigl(\mathbb{0}_{T}, \textbf{d}_{T},\textbf{m}_{T}\Bigr) \big| D_1(1) = d_{*,1}, D_1(0) = d_{1}, \textbf{D}_{2:T}(\mathbb{0}_{2:T}) = \textbf{d}_{2:T}, \mathbfcal{G}^{\mathbb{0}_{T},\textbf{d}_{T}(\mathbb{0}_{T})}_{m_{T}} = \textbf{m}_T, \textbf{Z}_{T} =\mathbb{0}_{T},  L_0 = \ell_0 \Big]
    \\& -
    E\Big[Y\Bigl(\mathbb{0}_{T}, \textbf{d}_{T},\textbf{m}_{T}\Bigr) \big|  \textbf{D}_T(\mathbb{0}_{T}) = \textbf{d}_T, \mathbfcal{G}^{\mathbb{0}_{T},\textbf{d}_{T}(\mathbb{0}_{T})}_{m_{T}} = \textbf{m}_T, \textbf{Z}_{T} =\mathbb{0}_{T},  L_0 = \ell_0 \Big]
    \\
    \rho_Y\big(\mathbb{1}_{T},\mathbb{0}_{T} \big)
    &= E\Big[Y\Bigl(\mathbb{1}_{T}, \textbf{d}_{T},\textbf{m}_{T}\Bigr) \big| D_1(1) = d_{1}, D_1(0) = d_{*,1}, \textbf{D}_{2:T}(\mathbb{1}_{2:T}) = \textbf{d}_{2:T}, \mathbfcal{G}^{\mathbb{0}_{T},\textbf{d}_{T}(\mathbb{0}_{T})}_{m_{T}} = \textbf{m}_T, \textbf{Z}_{T} =\mathbb{1}_{T},  L_0 = \ell_0 \Big]
    \\& -
    E\Big[Y\Bigl(\mathbb{1}_{T}, \textbf{d}_{T},\textbf{m}_{T}\Bigr) \big|  \textbf{D}_T(\mathbb{1}_{T}) = \textbf{d}_T, \mathbfcal{G}^{\mathbb{0}_{T},\textbf{d}_{T}(\mathbb{0}_{T})}_{m_{T}} = \textbf{m}_T, \textbf{Z}_{T} =\mathbb{1}_{T},  L_0 = \ell_0 \Big]
\end{aligned}
$}
\end{center}
Using the ignorability assumptions, Assumption 2 and Assumption 4, the sensitivity parameters $\rho_Y\big(\mathbb{1}_{T},\mathbb{1}_{T} \big)$, $\rho_Y\big(\mathbb{0}_{T},\mathbb{0}_{T} \big)$ and $\rho_Y\big(\mathbb{1}_{T},\mathbb{0}_{T} \big)$ can be simplified to:

\begin{equation}\label{Eq::SAreducedRho11}
\begin{aligned}
            \rho_Y\big(\mathbb{1}_{T},\mathbb{1}_{T} \big) 
            &= E\Big[Y\Bigl(\mathbb{1}_{T}, \textbf{d}_{T},\textbf{m}_{T}\Bigr) \big|D_1(1) = d_{1}, D_1(0) = d_{*,1}, \textbf{D}_{2:T}(\mathbb{1}_{2:T}) = \textbf{d}_{2:T}, L_0 = \ell_0 \Big]
            \\& - 
            E\Big[Y\Bigl(\mathbb{1}_{T}, \textbf{d}_{T},\textbf{m}_{T}\Bigr) \big|  \textbf{D}_T(\mathbb{1}_{T}) = \textbf{d}_T, L_0 = \ell_0 \Big]
\end{aligned}  
\end{equation}
\begin{equation}\label{Eq::SAreducedRho00}
\begin{aligned}
            \rho_Y\big(\mathbb{0}_{T},\mathbb{0}_{T} \big) 
            &=  E\Big[Y\Bigl(\mathbb{0}_{T}, \textbf{d}_{T},\textbf{m}_{T}\Bigr) \big| D_1(1) = d_{*,1}, D_1(0) = d_{1}, \textbf{D}_{2:T}(\mathbb{0}_{2:T}) = \textbf{d}_{2:T},  L_0 = \ell_0 \Big]
            \\& - 
            E\Big[Y\Bigl(\mathbb{0}_{T}, \textbf{d}_{T},\textbf{m}_{T}\Bigr) \big|  \textbf{D}_T(\mathbb{0}_{T}) = \textbf{d}_T,  L_0 = \ell_0 \Big]
\end{aligned}  
\end{equation}
\begin{equation}\label{Eq::SAreducedRho10}
\begin{aligned}
            \rho_Y\big(\mathbb{1}_{T},\mathbb{0}_{T} \big) 
            &= E\Big[Y\Bigl(\mathbb{1}_{T}, \textbf{d}_{T},\textbf{m}_{T}\Bigr) \big| D_1(1) = d_{1}, D_1(0) = d_{*,1}, \textbf{D}_{2:T}(\mathbb{1}_{2:T}) = \textbf{d}_{2:T},   L_0 = \ell_0 \Big]
            \\& - 
            E\Big[Y\Bigl(\mathbb{1}_{T}, \textbf{d}_{T},\textbf{m}_{T}\Bigr) \big|  \textbf{D}_T(\mathbb{1}_{T}) = \textbf{d}_T, L_0 = \ell_0 \Big].
\end{aligned}
\end{equation}
Note that $\rho_Y(\mathbb{1}_{T},\mathbb{1}_{T})$ and $\rho_Y(\mathbb{1}_{T},\mathbb{0}_{T})$ reduce to the same expression after invoking Assumptions~2 and~4, as shown in \eqref{Eq::SAreducedRho11} and \eqref{Eq::SAreducedRho10}.

\subsection{Calibrating sensitivity parameters}
To calibrate $\rho_Y\big(\mathbb{1}_{T},\mathbb{1}_{T} \big)$ \big(as well as $\rho_Y\big(\mathbb{1}_{T},\mathbb{0}_{T} \big)$\big) using the observed data, we begin by computing the total amount of variability in the outcomes $Y$ under $\textbf{Z}_T = \mathbb{1}_{T}$ that is explained by $\textbf{d}_T$ and $\ell_0$. Specifically, we compute the coefficient of determination from the regression of $Y$ on $\textbf{d}_T$ and $\ell_0$ for the treated subjects, denoted by $R_1$. We let $k_1$ represent the percentage of total variance not explained by $\textbf{d}_T$ and $\ell_0$. Using $R_1$ and $k_1$, we assume a bound on $\rho_Y\big(\mathbb{1}_{T},\mathbb{1}_{T} \big)$ \big(as well as on $\rho_Y\big(\mathbb{1}_{T},\mathbb{0}_{T} \big)$\big) of the form: $$\big|\rho_Y\big(\mathbb{1}_{T},\mathbb{1}_{T} \big)\big| \leq \sqrt{\text{Var}(Y\big|\textbf{Z}_T = \mathbb{1}_{T})\times (1-R_1) \times k_1}.$$
Intuitively, this implies that any unmeasured confounding that violates Assumption 3.3 is expected to have an effect bounded by the variance of $Y\big|\textbf{Z}_T = \mathbb{1}_{T}$ scaled by $100 \times k_1\%$ of the variability not explained by the measured confounders. Similarly, to calibrate $\rho_Y\big(\mathbb{0}_{T},\mathbb{0}_{T} \big)$ using observed data, we assume the bound:
$$\big|\rho_Y\big(\mathbb{0}_{T},\mathbb{0}_{T} \big)\big| \leq \sqrt{\text{Var}(Y\big|\textbf{Z}_T = \mathbb{0}_{T})\times (1-R_1) \times k_0}.$$
Thus, we define two sensitivity parameters $\big(k_0, k_1\big)$ bounded in the region of the unit square.

\subsection{Estimation for the violation of Assumption 3.3. assuming all other assumptions hold}
To conduct posterior inference for the causal effect estimates, we begin by taking posterior samples from the usual difference in means to compute the effect estimates, assuming the ignorability condition in Assumption~\ref{assump::3}.3 in the main text holds. Next, we generate random draws, one for each MCMC iteration, from the distribution of the sensitivity parameters $\big(k_0, k_1\big)$ and use these to compute the bounds on $\rho_Y\big(\mathbb{0}_{T},\mathbb{0}_{T} \big)$, $\rho_Y\big(\mathbb{1}_{T},\mathbb{1}_{T} \big)$, and $\rho_Y\big(\mathbb{1}_{T},\mathbb{0}_{T} \big)$. Finally, we adjust the causal effect estimates obtained under the ignorability assumption by subtracting the corresponding values of $\rho_Y(\mathbb{0}_{T},\mathbb{0}_{T})$, $\rho_Y(\mathbb{1}_{T},\mathbb{1}_{T})$, or $\rho_Y(\mathbb{1}_{T},\mathbb{0}_{T})$ from the unadjusted estimates. Formally, the causal parameters $\theta(\mathbb{1}_{T},\mathbb{1}_{T})$, $\theta(\mathbb{0}_{T},\mathbb{0}_{T})$, and $\theta(\mathbb{1}_{T},\mathbb{0}_{T})$ can be expressed in terms of the corresponding $\rho_Y$ parameters described above as:

\[\resizebox{\textwidth}{!}{$\displaystyle
\begin{aligned}
    \theta(\mathbb{1}_{T},\mathbb{1}_{T})
         &= \mathbb{E}\Bigl[Y\Bigl(\mathbb{1}_{T}, \textbf{d}_{T}(\mathbb{1}_{T}),\mathbfcal{G}^{\mathbb{1}_{T},\textbf{d}_{T}(\mathbb{1}_{T})}_{m_{T}}\Bigr)\big| U_1 = \bigl(d_{1}(1), d_{1}(0)\bigr) \Bigr]
        \\
        &= \int_{\textbf{b}} \int_{\ell_0} \int_{\textbf{d}_{2:T}}\int_{\textbf{m}_T}  
        \\&
        \mathbb{E}\Big[Y\Bigl(\mathbb{1}_{T}, \textbf{d}_{T},\textbf{m}_{T}\Bigr) \big| U_1 = \bigl(d_{1}(1), d_{1}(0)\bigr), \textbf{D}_{2:T}(\mathbb{1}_{2:T}) = \textbf{d}_{2:T}, \mathbfcal{G}^{\mathbb{1}_{T},\textbf{d}_{T}(\mathbb{1}_{T})}_{\textbf{m}_{T}} = \textbf{m}_T, \textbf{Z}_{T} =\mathbb{1}_{T},  L_0 = \ell_0 \Big]  \times 
        \\& 
        \Bigl\{ \prod_{t=1}^{T}   f\Bigl(\mathcal{G}^{\mathbb{1}_{t},\textbf{d}_{t}(1)}_{\textbf{m}_{t}} = \textbf{m}_t \big| U_1 = \bigl(d_{1}(1), d_{1}(0)\bigr),  \textbf{D}_{2:t}(\mathbb{1}_{2:t}) = \textbf{d}_{2:t}, \mathbfcal{G}^{\mathbb{1}_{t-1},\textbf{d}_{t-1}(\mathbb{1}_{t-1})}_{\textbf{m}_{t-1}} = \textbf{m}_{t-1},\textbf{Z}_{t} =\mathbb{1}_{t},    
         L_0 = \ell_0, b^{M} \Bigr)  \Bigr\}
        \times 
        \\&
        \Bigl\{ \prod_{s=2}^{T} f\Bigl(D_s(1) = d_{s} \big| U_1 = \bigl(d_{1}(1), d_{1}(0)\bigr),  \textbf{D}_{s-1}(\mathbb{1}_{s-1}) = \textbf{d}_{s-1}, \mathbfcal{G}^{\mathbb{1}_{s-1},\textbf{d}_{s-1}(\mathbb{1}_{s-1})}_{\textbf{m}_{s-1}} = \textbf{m}_{s-1},\textbf{Z}_{s} =\mathbb{1}_{s},
        L_0 = \ell_0, b^{D} \Bigr)\Bigr\} \\&
        \times f_{L_0}(\ell_0) d\textbf{m}_T d\textbf{d}_{2:T} d\ell_0 d\textbf{b}
        \\
        &= \rho_Y\big(\mathbb{1}_{T},\mathbb{1}_{T} \big) 
            + \int_{\textbf{b}}\int_{\ell_0} \int_{\textbf{d}_{2:T}}\int_{\textbf{m}_T}  
            E\Big[Y\Bigl(\mathbb{1}_{T}, \textbf{d}_{T},\textbf{m}_{T}\Bigr) \big|  \textbf{D}_T(\mathbb{1}_{T}) = \textbf{d}_T, \mathbfcal{G}^{\mathbb{1}_{T},\textbf{d}_{T}(\mathbb{1}_{T})}_{m_{T}} = \textbf{m}_T, \textbf{Z}_{T} =\mathbb{1}_{T},  L_0 = \ell_0 \Big]
        \times 
        \\& 
        \Bigl\{ \prod_{t=1}^{T}   f\Bigl(\mathcal{G}^{\mathbb{1}_{t},\textbf{d}_{t}(1)}_{\textbf{m}_{t}} = \textbf{m}_t \big| U_1 = \bigl(d_{1}(1), d_{1}(0)\bigr),  \textbf{D}_{2:t}(\mathbb{1}_{2:t}) = \textbf{d}_{2:t}, \mathbfcal{G}^{\mathbb{1}_{t-1},\textbf{d}_{t-1}(\mathbb{1}_{t-1})}_{\textbf{m}_{t-1}} = \textbf{m}_{t-1},\textbf{Z}_{t} =\mathbb{1}_{t},    
         L_0 = \ell_0, b^{M} \Bigr)  \Bigr\}
        \times 
        \\&
        \Bigl\{ \prod_{s=2}^{T} f\Bigl(D_s(1) = d_{s} \big| U_1 = \bigl(d_{1}(1), d_{1}(0)\bigr),  \textbf{D}_{s-1}(\mathbb{1}_{s-1}) = \textbf{d}_{s-1}, \mathbfcal{G}^{\mathbb{1}_{s-1},\textbf{d}_{s-1}(\mathbb{1}_{s-1})}_{\textbf{m}_{s-1}} = \textbf{m}_{s-1},\textbf{Z}_{s} =\mathbb{1}_{s},
        L_0 = \ell_0, b^{D} \Bigr)\Bigr\} \\&
        \times f_{L_0}(\ell_0) d\textbf{m}_T d\textbf{d}_{2:T} d\ell_0 d\textbf{b},
\end{aligned}
$}\]
\[\resizebox{\textwidth}{!}{$\displaystyle
\begin{aligned}
    \theta(\mathbb{0}_{T},\mathbb{0}_{T})
         &= \mathbb{E}\Bigl[Y\Bigl(\mathbb{0}_{T}, \textbf{d}_{T}(\mathbb{0}_{T}),\mathbfcal{G}^{\mathbb{0}_{T},\textbf{d}_{T}(\mathbb{0}_{T})}_{m_{T}}\Bigr)\big| U_1 = \bigl(d_{1}(1), d_{1}(0)\bigr) \Bigr]
        \\
        &= \int_{\textbf{b}} \int_{\ell_0} \int_{\textbf{d}_{2:T}}\int_{\textbf{m}_T}  
        \\&
        \mathbb{E}\Big[Y\Bigl(\mathbb{0}_{T}, \textbf{d}_{T},\textbf{m}_{T}\Bigr) \big| U_1 = \bigl(d_{1}(1), d_{1}(0)\bigr), \textbf{D}_{2:T}(\mathbb{0}_{2:T}) = \textbf{d}_{2:T}, \mathbfcal{G}^{\mathbb{0}_{T},\textbf{d}_{T}(\mathbb{0}_{T})}_{\textbf{m}_{T}} = \textbf{m}_T, \textbf{Z}_{T} =\mathbb{0}_{T},  L_0 = \ell_0 \Big]  \times 
        \\& 
        \Bigl\{ \prod_{t=1}^{T}   f\Bigl(\mathcal{G}^{\mathbb{0}_{t},\textbf{d}_{t}(0)}_{\textbf{m}_{t}} = \textbf{m}_t \big| U_1 = \bigl(d_{1}(1), d_{1}(0)\bigr),  \textbf{D}_{2:t}(\mathbb{0}_{2:t}) = \textbf{d}_{2:t}, \mathbfcal{G}^{\mathbb{0}_{t-1},\textbf{d}_{t-1}(\mathbb{0}_{t-1})}_{\textbf{m}_{t-1}} = \textbf{m}_{t-1},\textbf{Z}_{t} =\mathbb{0}_{t},    
         L_0 = \ell_0, b^{M} \Bigr)  \Bigr\}
        \times 
        \\&
        \Bigl\{ \prod_{s=2}^{T} f\Bigl(D_s(0) = d_{s} \big| U_1 = \bigl(d_{1}(1), d_{1}(0)\bigr),  \textbf{D}_{s-1}(\mathbb{0}_{s-1}) = \textbf{d}_{s-1}, \mathbfcal{G}^{\mathbb{0}_{s-1},\textbf{d}_{s-1}(\mathbb{0}_{s-1})}_{\textbf{m}_{s-1}} = \textbf{m}_{s-1},\textbf{Z}_{s} =\mathbb{0}_{s},
        L_0 = \ell_0, b^{D} \Bigr)\Bigr\} \\&
        \times f_{L_0}(\ell_0) d\textbf{m}_T d\textbf{d}_{2:T} d\ell_0 d\textbf{b}
        \\
        &= \rho_Y\big(\mathbb{0}_{T},\mathbb{0}_{T} \big) 
            + \int_{\textbf{b}} \int_{\ell_0} \int_{\textbf{d}_{2:T}}\int_{\textbf{m}_T}   
            E\Big[Y\Bigl(\mathbb{0}_{T}, \textbf{d}_{T},\textbf{m}_{T}\Bigr) \big|  \textbf{D}_T(\mathbb{0}_{T}) = \textbf{d}_T, \mathbfcal{G}^{\mathbb{0}_{T},\textbf{d}_{T}(\mathbb{0}_{T})}_{m_{T}} = \textbf{m}_T, \textbf{Z}_{T} =\mathbb{0}_{T},  L_0 = \ell_0 \Big]
        \times 
        \\& 
        \Bigl\{ \prod_{t=1}^{T}   f\Bigl(\mathcal{G}^{\mathbb{0}_{t},\textbf{d}_{t}(0)}_{\textbf{m}_{t}} = \textbf{m}_t \big| U_1 = \bigl(d_{1}(1), d_{1}(0)\bigr),  \textbf{D}_{2:t}(\mathbb{0}_{2:t}) = \textbf{d}_{2:t}, \mathbfcal{G}^{\mathbb{0}_{t-1},\textbf{d}_{t-1}(\mathbb{0}_{t-1})}_{\textbf{m}_{t-1}} = \textbf{m}_{t-1},\textbf{Z}_{t} =\mathbb{0}_{t},    
         L_0 = \ell_0, b^{M} \Bigr)  \Bigr\}
        \times 
        \\&
        \Bigl\{ \prod_{s=2}^{T} f\Bigl(D_s(0) = d_{s} \big| U_1 = \bigl(d_{1}(1), d_{1}(0)\bigr),  \textbf{D}_{s-1}(\mathbb{0}_{s-1}) = \textbf{d}_{s-1}, \mathbfcal{G}^{\mathbb{0}_{s-1},\textbf{d}_{s-1}(\mathbb{0}_{s-1})}_{\textbf{m}_{s-1}} = \textbf{m}_{s-1},\textbf{Z}_{s} =\mathbb{0}_{s},
        L_0 = \ell_0, b^{D} \Bigr)\Bigr\} \\&
        \times f_{L_0}(\ell_0) d\textbf{m}_T d\textbf{d}_{2:T} d\ell_0 d\textbf{b},
\end{aligned}
$}\]
and
\[\resizebox{\textwidth}{!}{$\displaystyle
\begin{aligned}
    \theta(\mathbb{1}_{T},\mathbb{0}_{T})
         &= \mathbb{E}\Bigl[Y\Bigl(\mathbb{1}_{T}, \textbf{d}_{T}(\mathbb{1}_{T}),\mathbfcal{G}^{\mathbb{0}_{T},\textbf{d}_{T}(\mathbb{0}_{T})}_{m_{T}}\Bigr)\big| U_1 = \bigl(d_{1}(1), d_{1}(0)\bigr) \Bigr]
        \\
        &=\int_{\textbf{b}} \int_{\ell_0} \int_{\textbf{d}_{2:T}}\int_{\textbf{m}_T}  
        \\&
        \mathbb{E}\Big[Y\Bigl(\mathbb{1}_{T}, \textbf{d}_{T},\textbf{m}_{T}\Bigr) \big| U_1 = \bigl(d_{1}(1), d_{1}(0)\bigr), \textbf{D}_{2:T}(\mathbb{1}_{2:T}) = \textbf{d}_{2:T}, \mathbfcal{G}^{\mathbb{0}_{T},\textbf{d}_{T}(\mathbb{0}_{T})}_{\textbf{m}_{T}} = \textbf{m}_T, \textbf{Z}_{T} =\mathbb{1}_{T},  L_0 = \ell_0 \Big]  \times 
        \\& 
        \Bigl\{ \prod_{t=1}^{T}   f\Bigl(\mathcal{G}^{\mathbb{0}_{t},\textbf{d}_{t}(0)}_{\textbf{m}_{t}} = \textbf{m}_t \big| U_1 = \bigl(d_{1}(1), d_{1}(0)\bigr),  \textbf{D}_{2:t}(\mathbb{0}_{2:t}) = \textbf{d}_{2:t}, \mathbfcal{G}^{\mathbb{0}_{t-1},\textbf{d}_{t-1}(\mathbb{0}_{t-1})}_{\textbf{m}_{t-1}} = \textbf{m}_{t-1},\textbf{Z}_{t} =\mathbb{0}_{t},    
         L_0 = \ell_0 , b^{M}\Bigr)  \Bigr\}
        \times 
        \\&
        \Bigl\{ \prod_{s=2}^{T} f\Bigl(D_s(1) = d_{s} \big| U_1 = \bigl(d_{1}(1), d_{1}(0)\bigr),  \textbf{D}_{s-1}(\mathbb{1}_{s-1}) = \textbf{d}_{s-1}, \mathbfcal{G}^{\mathbb{0}_{s-1},\textbf{d}_{s-1}(\mathbb{0}_{s-1})}_{\textbf{m}_{s-1}} = \textbf{m}_{s-1},\textbf{Z}_{s} =\mathbb{1}_{s},
        L_0 = \ell_0, b^{D} \Bigr)\Bigr\} \\&
        \times f_{L_0}(\ell_0) d\textbf{m}_T d\textbf{d}_{2:T} d\ell_0 d\textbf{b}
        \\
        &= \rho_Y\big(\mathbb{1}_{T},\mathbb{0}_{T} \big) 
            + \int_{\textbf{b}} \int_{\ell_0} \int_{\textbf{d}_{2:T}}\int_{\textbf{m}_T}  
            \mathbb{E}\Big[Y\Bigl(\mathbb{1}_{T}, \textbf{d}_{T},\textbf{m}_{T}\Bigr) \big|  \textbf{D}_{T}(\mathbb{1}_{T}) = \textbf{d}_{T}, \mathbfcal{G}^{\mathbb{0}_{T},\textbf{d}_{T}(\mathbb{0}_{T})}_{\textbf{m}_{T}} = \textbf{m}_T, \textbf{Z}_{T} =\mathbb{1}_{T},  L_0 = \ell_0 \Big] 
         \times   
        \\& 
        \Bigl\{ \prod_{t=1}^{T}   f\Bigl(\mathcal{G}^{\mathbb{0}_{t},\textbf{d}_{t}(0)}_{\textbf{m}_{t}} = \textbf{m}_t \big| U_1 = \bigl(d_{1}(1), d_{1}(0)\bigr),  \textbf{D}_{2:t}(\mathbb{0}_{2:t}) = \textbf{d}_{2:t}, \mathbfcal{G}^{\mathbb{0}_{t-1},\textbf{d}_{t-1}(\mathbb{0}_{t-1})}_{\textbf{m}_{t-1}} = \textbf{m}_{t-1},\textbf{Z}_{t} =\mathbb{0}_{t},    
         L_0 = \ell_0, b^{M} \Bigr)  \Bigr\}
        \times 
        \\&
        \Bigl\{ \prod_{s=2}^{T} f\Bigl(D_s(1) = d_{s} \big| U_1 = \bigl(d_{1}(1), d_{1}(0)\bigr),  \textbf{D}_{s-1}(\mathbb{1}_{s-1}) = \textbf{d}_{s-1}, \mathbfcal{G}^{\mathbb{0}_{s-1},\textbf{d}_{s-1}(\mathbb{0}_{s-1})}_{\textbf{m}_{s-1}} = \textbf{m}_{s-1},\textbf{Z}_{s} =\mathbb{1}_{s},
        L_0 = \ell_0, b^{D} \Bigr)\Bigr\} \\&
        \times f_{L_0}(\ell_0) d\textbf{m}_T d\textbf{d}_{2:T} d\ell_0 d\textbf{b},
\end{aligned}
$}\]
respectively.

\subsection{Sensitivity analysis results}
In this section, we present results for the sensitivity analysis. Recall that Table~\ref{Table:EffectEstimatesvsBaselineRegime} in the main text compare three treatment assignment regimes---$\textbf{Z}_{T} =$ \{V,P,P\}, $\textbf{Z}_{T} =$ \{V,V,P\}, and $\textbf{Z}_{T} = $ \{V,V,V\}---against the baseline regime $\textbf{Z}_{*,T} =$ \{P,P,P\}.

  Table~\ref{TableSA:EffectEstimatesvsBaselineRegime} reports the sensitivity-adjusted effect estimates across these comparisons. Consistent with the results in the main text, the adjusted indirect effects remain negative, and statistical significance is observed for only \{V,V,V\} vs. \{P,P,P\} only. Likewise, the posterior means of the direct and total effects are positive and significant, consistent with the findings reported in the case study. Overall, the sign and statistical significance of all sensitivity-adjusted estimates align with those obtained under the unadjusted analysis, assuming Assumption 3.3 holds. Thus, under our proposed sensitivity analysis framework, the estimates are robust to possible violations of Assumption 3.3.

\begin{table}[h]
\caption{Causal effect estimates (posterior mean and $95\%$ Bayesian credible interval) across treatment regime contrasts adjusting for sensitivity parameters}
\label{TableSA:EffectEstimatesvsBaselineRegime}
\begin{tabular}{@{}lcrcrrr@{}}
\hline
Principal strata & Effects &  \{V,P,P\} vs. \{P,P,P\} &  \{V,V,P\} vs. \{P,P,P\}   &  \{V,V,V\} vs. \{P,P,P\} \\
 at $t=1$ &   &  &  & 
\\
\hline
{Active}  & {Direct} &  2.13(0.96, 3.37) & 3.49(2.19, 5.00) & 4.17(2.55, 5.72) \\
{customers}  & {Indirect}  &  $-$0.03($-$0.58, 0.80) & $-$0.55($-$1.25, 0.24) & $-$1.15($-$2.22, $-$0.33) \\
          & {Total}   &  2.09(0.78, 3.66) & 2.94(1.51, 4.56) & 3.02(1.29, 4.38) \\[6pt]
{Non-price}  & {Direct} &  2.10(0.92, 3.26) & 3.47(2.13, 4.78) & 4.25(2.61, 5.79) \\
 {value-attentive}  & {Indirect}  &  $-$0.27($-$0.85, 0.67) & $-$0.76($-$1.45, 0.20) & $-$1.36($-$2.45, $-$0.68) \\
  {customers}  & {Total}   &  1.84(0.54, 3.32) & 2.71(1.25, 4.23) & 2.82(1.18, 4.12) \\[6pt]
{Price-attentive} & {Direct} &  2.14(0.89, 3.67) & 3.50(2.18, 5.20) & 4.14(2.55, 5.78) \\
{customers}  & {Indirect}  &  $-$0.16($-$0.77, 0.50) & $-$0.67($-$1.40, 0.03) & $-$1.22($-$2.31, $-$0.45) \\
          & {Total}   &  1.98(0.68, 3.58) & 2.83(1.45, 4.56) & 2.92(1.32, 4.26) \\[6pt]
{Non-active}  & {Direct} &  2.12(0.95, 3.30) & 3.48(2.19, 4.91) & 4.20(2.68, 5.71) \\
{customers}   & {Indirect}  &  $-$0.39($-$0.93, 0.36) & $-$0.88($-$1.58, $-$0.09) & $-$1.48($-$2.53, $-$0.80) \\
          & {Total}   &  1.73(0.50, 3.16) & 2.60(1.25, 4.12) & 2.72(1.14, 3.93) \\[6pt]
\hline
\end{tabular}
\end{table}

\newpage

\section{Details on the EDPM model and g-computation}\label{sec:S4}
The notation $EDP(\alpha^{\beta}, \alpha^{\theta|\beta}, H_0 )$ means that $H_{\beta} \sim DP(\alpha^{\beta}, H_{0\beta})$ and $H_{\theta|\beta} \sim DP(\alpha^{\theta|\beta}, H_{0\theta|\beta})$, where $\alpha^{\beta}$ and $\alpha^{\theta|\beta}$ are positive valued parameters and $H_0 = H_{0\beta} \times H_{0\theta|\beta}$ is the base distribution. 

\subsection{Priors for the $\beta$-level parameters}
For the coefficient of the $p$th covariate, $\beta_{i,p}$, in the local outcome regression model, we assume the following prior:
\begin{align*}
    \beta_{i,p}\stackrel{ind}{\sim}\text{N}\big(\beta_{0,p}, c\,\sigma^{2,Y}_{0,p} \big).
\end{align*}
We set $\beta_{0,p}$ and $\sigma^{2,Y}_{0,p}$ to the maximum likelihood estimates of the $p$th coefficient and its variance, respectively, obtained from fitting a linear regression of the outcome on the covariates using the full dataset. In other words, our prior belief for the cluster-specific coefficients in the local outcome regression model corresponds to the coefficients from a linear model fitted to the entire dataset, with the associated uncertainty represented by a constant $c > 1$ multiplied by the variance estimates. Following the arguments of \cite{roy2018bayesian}, we set $c = n/5$. Similarly, for the variance parameter $\sigma_{i}^{2,Y}$, we specify the prior:
\begin{align*}
    \sigma_{i}^{2,Y} \sim \text{Inverse-Gamma}(a^{Y}, b^{Y}),
\end{align*}
where we set the shape parameter \( a^{Y} = 3 \) and define the scale parameter \( b^{Y} \) as $2 \times$ variance from the linear regression of the outcome on the covariates using the full dataset. Since the mean of an Inverse-Gamma$(a,b)$ distribution is $b/(a-1)$ for $a>1$, the choice $a^{Y}=3$ yields a prior mean centering the prior at the full-data variance estimate. Finally, for the hurdle model probability parameter \( \pi_{i}^{Y} \), we assign an uninformative \(\text{Beta}(1,1)\) prior.

 In summary, we assume:
\begin{align*}
            H_{0\beta} &\sim \underbrace{\text{Beta}(1,1 )}_{\pi_{i}^{Y}} \times  \underbrace{\text{Inverse-Gamma}(a^{Y},b^{Y})}_{\sigma_{i}^{2,Y}} \times \prod_{p=1}^{P}
          \underbrace{\text{N}\big(\beta^Y_{0,p}, c\,\sigma_{0,p}^{2,Y}\big)}_{\beta_{i,p}^Y, p = 1, \ldots, P},
        \end{align*}
where $P$ denotes the number of covariates, including both baseline and time-varying variables, in the local outcome regression model.

\subsection{Priors for the $\theta$-level parameters}
At the $\theta-$level, we have regression parameters $ \big(\boldsymbol{\theta}^{M^{(1)}}_{i,t}, \boldsymbol{\theta}^{M^{(2)}}_{i,t}, \boldsymbol{\theta}^D_{i,t}, \boldsymbol{\theta}^Z_{i,t}\big)$, $t = 1, \ldots, T$, for modeling the time-varying variables and regression parameters $\big(\boldsymbol{\theta}_{i}^{L_0}\big)$ for modeling the baseline confounders. Without loss of generality, let us assume that we have continuous time-varying mediators. We assume the following base measures for these parameters:
\begin{equation}\label{Eq:thetaLevelPriors}
        \begin{aligned}
             H_{0\theta|\beta} &\sim \prod_{t=1}^{T}\bigg\{ \prod_{q=1}^{Q}\Big\{ \underbrace{ \text{N} \big(\theta^{Z_t}_{0,q}, c\sigma^{2,\theta^{Z_t}}_{0,q}\big)}_{\theta_{i,q}^{Z_t}}  \times \underbrace{\text{N} \big(\theta^{D_t}_{0,q}, c\sigma^{2,\theta^{D_t}}_{0,q}\big)}_{\theta_{i,q}^{D_t}}
             \times \underbrace{\text{N} \big(\theta^{M^{(1)}_t}_{0,q}, c\sigma^{2,\theta^{M^{(1)}_t}}_{0,q}\big)}_{\theta_{i,q}^{M^{(1)}_t}} \times \underbrace{\text{N} \big(\theta^{M^{(2)}_t}_{0,q}, c\sigma^{2,\theta^{M^{(2)}_t}}_{0,q}\big)}_{\theta_{i,q}^{M^{(2)}_t}} \Big\} \times 
             \\& \underbrace{\text{Inverse-Gamma}\big(a^{M^{(1)}_t},b^{M^{(1)}_t} \big)}_{\sigma_{i}^{2,M^{(1)}_t}} \times \underbrace{\text{Inverse-Gamma}\big(a^{M^{(2)}_t},b^{M^{(2)}_t} \big)}_{\sigma_{i}^{2,M^{(2)}_t}} \times \underbrace{\text{Beta}(1,1 )}_{\pi_{i}^{M^{(1)}_t}} \times \underbrace{\text{Beta}(1,1 )}_{\pi_{i}^{M^{(2)}_t}} \bigg\}  \times 
                f_0\big(\theta_i^{\textbf{L}_0}\big),
        \end{aligned}
\end{equation}
where $Q$ denotes the number of regression coefficients in each local longitudinal regression (the covariates plus an intercept). We center and scale the base measures using maximum likelihood estimates from ordinary linear or logistic regressions applied to all of the data, again setting $c = n/5$. In the specification above, $f_0(\theta_i^{\textbf{L}_0}) = \prod_{s = 1}^{K}f_{0,s}(\theta_{i,s}^{\textbf{L}_0})$ for $K$ baseline covariates with 
            \begin{align*}
                f_{0,s}\big( \theta_{i,s}^{\textbf{L}_0}\big) = \begin{cases}
                    \underbrace{\text{Inverse-Gamma}\big(a_{\boldsymbol{\ell}_{0}},b_{\boldsymbol{\ell}_{0}} \big)}_{\sigma_{i,s}^{2,\textbf{L}_{0}}} \times \underbrace{\text{N}\big(\mu_{\boldsymbol{\ell}_{0}},\sigma_{\boldsymbol{\ell}_{0}}^2\big)}_{\theta_{i,s}^{\textbf{L}_0}}  \quad\quad\quad \text{for continuous baseline covariates}\\
                 \underbrace{\text{Beta} \big(a_{\boldsymbol{\ell}_{0}},b_{\boldsymbol{\ell}_{0}}\big)}_{p_{i,s}^{\textbf{L}_0}} \quad\quad\quad \quad\quad\quad \quad\quad\quad \quad\quad\quad \text{     for binary baseline covariates}.
                \end{cases}
            \end{align*}
 We assume uninformative conjugate priors for the baseline confounder parameters.

 \subsection{Priors for the EDP mass parameters}
The number of clusters in the EDP model is influenced by the concentration parameters $\alpha^{\beta}$ and $\alpha^{\theta|\beta}$, with smaller values of these parameters corresponding to fewer clusters. Thus, careful selection of these concentration parameters is essential for determining appropriate truncation levels $N$ and $M$ in EDP mixtures. Since the square-breaking weights decay exponentially, $N$ and $M$ can typically be chosen to be relatively small. Chapter 6 of \cite{daniels2023bayesian} provides a simple calculation showing that, on average, when the concentration parameter $\alpha = 1$, the first 20 stick-breaking weights in the case of Dirichlet process priors sum to approximately 1. Among these first 20 weights, the latter 10 are, in expectation, approximately 0 when $\alpha = 1$. Utilizing this observation, we set $\alpha^{\beta} = 1$ in our work and choose $N = 10$.

 The choice of the inner-level truncation value $M$ in our work is data-driven in the sense that we evaluate several candidate values of $M$ ranging from $1$ to $10$ when fitting the observed data models. Empirically, we find that only $2$--$3$ clusters typically have non-negligible membership probabilities across the four MCMC chains used in the analysis. Accordingly, we fix the inner-level EDP concentration parameter at $\alpha^{\theta|\beta} = 0.5$ and set the number of inner clusters to $M = 4$ as a conservative choice. Consequently, the mass parameters $\alpha^{\beta}$ and $\alpha^{\theta|\beta}$ are treated as fixed in our analyses--the corresponding Gibbs updates in Appendix~E (Steps 6 and 7) are presented for completeness and are omitted when these parameters are held fixed.

\subsection{Random effects}
The random effects $b_i^{M^{(1)}}, b_i^{M^{(2)}}, b_i^{D},$ and $b_i^{Z}$, which are not included in the EDP prior, are assumed to follow mean-zero normal distributions:
$$b_i^{M^{(1)}} \sim N\big(0, \sigma_{b}^{2,M^{(1)}}\big),
b_i^{M^{(2)}} \sim N\big(0, \sigma_{b}^{2,M^{(2)}}\big), b_i^{D} \sim N\big(0, \sigma_{b}^{2,D}\big), \text{ and } b_i^{Z} \sim N\big(0, \sigma_{b}^{2,Z}\big).$$ 
The corresponding variances of these random intercepts, $ \sigma_{b}^{2,M^{(1)}}, \sigma_{b}^{2,M^{(2)}}, \sigma_{b}^{2,D}$, and  $\sigma_{b}^{2,Z}$, are assigned conjugate inverse-gamma priors:
\begin{align*}
    \sigma_{b}^{2,M^{(1)}} &\sim IG\big(\alpha_{\sigma_{b}^{2,M^{(1)}}}, \beta_{\sigma_{b}^{2,M^{(1)}}} \big),
    \\
    \sigma_{b}^{2,M^{(2)}} &\sim IG\big(\alpha_{\sigma_{b}^{2,M^{(2)}}}, \beta_{\sigma_{b}^{2,M^{(2)}}} \big),
    \\
    \sigma_{b}^{2,D} &\sim IG\big(\alpha_{\sigma_{b}^{2,D}}, \beta_{\sigma_{b}^{2,D}} \big), \text{ and }
    \\
    \sigma_{b}^{2,Z} &\sim IG\big(\alpha_{\sigma_{b}^{2,Z}}, \beta_{\sigma_{b}^{2,Z}} \big).
\end{align*}

\subsection{Details on the EDP mixture models}

 For conciseness, let $\textbf{b}_i = \big\{ b_i^{M^{(2)}}, b_i^{M^{(1)}}, b_i^{D}, b_i^{Z}\big\}$ denote a vector of random effects for subject $i$. We can express the joint density of all random variables under the EDPM specification in Section~\ref{sec:4.1} of the main text as an infinite mixture:
    \begin{equation}\label{Eq:EDPJointMixture}
\resizebox{\textwidth}{!}{$\displaystyle
        \begin{aligned}
            & f_{H}\bigl(y_i, \textbf{m}^{(1)}_{i,T},\textbf{m}^{(2)}_{i,T},\textbf{d}_{i,T},\textbf{z}_{i,T},\boldsymbol{\ell}_{i,0},\textbf{b}_i\bigr)
            \\
            & =f_{H}\bigl(y_i, \textbf{m}^{(1)}_{i,T},\textbf{m}^{(2)}_{i,T},\textbf{d}_{i,T},\textbf{z}_{i,T},\boldsymbol{\ell}_{i,0}\big|\textbf{b}_i\bigr) \times f_{H} \big(\textbf{b}_i\big)
            \\
            & = \sum_{r=1}^{\infty}\Bigl\{\gamma_r p\bigl(y_i\big|\textbf{m}^{(1)}_{i,T},\textbf{m}^{(2)}_{i,T}, \textbf{d}_{i,T},\textbf{z}_{i,T}, \boldsymbol{\ell}_{i,0};\beta_r\bigr) \times \sum_{s=1}^{\infty} \gamma_{s|r} p\bigl(\textbf{m}^{(1)}_{i,T},\textbf{m}^{(2)}_{i,T}, \textbf{d}_{i,T},\textbf{z}_{i,T},\boldsymbol{\ell}_{i,0} \big| \textbf{b}_i;\theta_{s|r}\bigr)\Bigr\} \times f_{H}\big(\textbf{b}_i\big)
            \\&
            = \sum_{r=1}^{\infty}\Bigl\{\gamma_r p\bigl(y_i\big| \textbf{m}^{(1)}_{i,T},\textbf{m}^{(2)}_{i,T}, \textbf{d}_{i,T},\textbf{z}_{i,T}, \boldsymbol{\ell}_{i,0};\beta_r\bigr) \times  \sum_{s=1}^{\infty} \gamma_{s|r}
            \prod_{t = 1}^{T} \big\{
            p\bigl(m^{(1)}_{i,t},m^{(2)}_{i,t}\big| \textbf{d}_{i,t},\textbf{z}_{i,t}, \boldsymbol{\ell}_{i,0}, b_i^{M^{(2)}}, b_i^{M^{(1)}};\theta_{s|r}\bigr)   \times  \\&
             p\bigl(d_{i,t}\big| \textbf{z}_{i,t}, \boldsymbol{\ell}_{i,0}, b_i^{D};\theta_{s|r}\bigr)
             \times   p\bigl(z_{i,t}\big| \boldsymbol{\ell}_{i,0}, b_i^{Z};\theta_{s|r}\bigr) \big\} \times p\bigl( \boldsymbol{\ell}_{i,0};\theta_{s|r}\bigr)\Bigr\} \times f_{H}\big(\textbf{b}_i\big)
        \end{aligned}
$}
    \end{equation}
where $p(.)$ denotes the corresponding local density associated with distributions in the EDPM model specified in the main text. The third equality in Equation (\ref{Eq:EDPJointMixture}) follows from assuming local independence within clusters.

  Using (\ref{Eq:EDPJointMixture}), we can write:
\begin{equation}
\resizebox{\textwidth}{!}{$\displaystyle
        \begin{aligned}
            f_{H}\bigl( \textbf{m}^{(1)}_{i,T},\textbf{m}^{(2)}_{i,T},\textbf{d}_{i,T},\textbf{z}_{i,T},\boldsymbol{\ell}_{i,0}, \textbf{b}_i\bigr)
            & = \sum_{r=1}^{\infty}\Bigl\{\gamma_r  \sum_{s=1}^{\infty} \gamma_{s|r}
            \prod_{t = 1}^{T} \big\{
            p\bigl(m^{(1)}_{i,t},m^{(2)}_{i,t}\big| \textbf{d}_{i,t},\textbf{z}_{i,t}, \boldsymbol{\ell}_{i,0}, b_i^{M^{(2)}}, b_i^{M^{(1)}};\theta_{s|r}\bigr)   \times  \\&
             p\bigl(d_{i,t}\big|\textbf{z}_{i,t}, \boldsymbol{\ell}_{i,0}, b_i^{D};\theta_{s|r}\bigr)
             \times  p\bigl(z_{i,t}\big| \boldsymbol{\ell}_{i,0}, b_i^{Z};\theta_{s|r}\bigr) \big\} \times p\bigl( \boldsymbol{\ell}_{i,0};\theta_{s|r}\bigr)\Bigr\}\times f_{H}\big(\textbf{b}_i\big),
        \end{aligned}
$}
    \end{equation}
\begin{equation}
\resizebox{\textwidth}{!}{$\displaystyle
        \begin{aligned}
             f_{H}\bigl(\textbf{d}_{i,T}, \textbf{z}_{i,T}, \textbf{m}^{(1)}_{i,T-1},\textbf{m}^{(2)}_{i,T-1},\boldsymbol{\ell}_{i,0}, \textbf{b}_i\bigr)
             & = \sum_{r=1}^{\infty}\Bigl\{\gamma_r  \sum_{s=1}^{\infty} \gamma_{s|r}
             \prod_{t = 1}^{T} \big\{ p\bigl(d_{i,t}\big| \textbf{z}_{i,t},\boldsymbol{\ell}_{i,0}, b_i^{D};\theta_{s|r}\bigr)
              \times p\bigl(z_{i,t}\big| \boldsymbol{\ell}_{i,0}, b_i^{Z};\theta_{s|r}\bigr) \\& \times p\bigl(m^{(1)}_{i,t-1},m^{(2)}_{i,t-1}\big| \textbf{d}_{i,t-1},\textbf{z}_{i,t-1}, \boldsymbol{\ell}_{i,0}, b_i^{M^{(2)}}, b_i^{M^{(1)}};\theta_{s|r}\bigr) \big\} \times p\bigl( \boldsymbol{\ell}_{i,0};\theta_{s|r}\bigr)\Bigr\}\times f_{H}\big(\textbf{b}_i\big).
        \end{aligned}
$}
    \end{equation}
The EDPM induces the following conditional density for $f_{H}\bigl(y_i\big| \textbf{m}^{(1)}_{i,T},\textbf{m}^{(2)}_{i,T},\textbf{d}_{i,T},\textbf{z}_{i,T},\boldsymbol{\ell}_{i,0}, \textbf{b}_i\bigr)$:
    \begin{equation}
\resizebox{\textwidth}{!}{$\displaystyle
        \begin{aligned}
            & f_{H}\bigl(y_i\big| \textbf{m}^{(1)}_{i,T},\textbf{m}^{(2)}_{i,T},\textbf{d}_{i,T},\textbf{z}_{i,T},\boldsymbol{\ell}_{i,0}, \textbf{b}_i\bigr)
            \\
            &= \frac{f_{H}\bigl(y_i, \textbf{m}^{(1)}_{i,T},\textbf{m}^{(2)}_{i,T},\textbf{d}_{i,T},\textbf{z}_{i,T},\boldsymbol{\ell}_{i,0}, \textbf{b}_i\bigr)}{f_{H}\bigl( \textbf{m}^{(1)}_{i,T},\textbf{m}^{(2)}_{i,T},\textbf{d}_{i,T},\textbf{z}_{i,T},\boldsymbol{\ell}_{i,0}, \textbf{b}_i\bigr)}
            \\
            &= \frac{\splitfrac{\sum_{u=1}^{\infty}\Bigl\{\gamma_u p\bigl(y_i\big|\textbf{m}^{(1)}_{i,T},\textbf{m}^{(2)}_{i,T},\textbf{d}_{i,T},\textbf{z}_{i,T},   \boldsymbol{\ell}_{i,0};\beta_u\bigr)  \sum_{v=1}^{\infty} \gamma_{v|u} \prod_{t=1}^{T}\big\{ p\bigl(m^{(1)}_{i,t},m^{(2)}_{i,t}\big|\textbf{d}_{i,t},\textbf{z}_{i,t}, \boldsymbol{\ell}_{i,0}, b_i^{M^{(2)}}, b_i^{M^{(1)}};\theta_{v|u}\bigr)}{ p\bigl(d_{i,t}\big| \textbf{z}_{i,t}, \boldsymbol{\ell}_{i,0}, b_i^{D};\theta_{v|u}\bigr) p\bigl(z_{i,t}\big| \boldsymbol{\ell}_{i,0}, b_i^{Z};\theta_{v|u}\bigr) \bigr\} p\bigl(\boldsymbol{\ell}_{i,0};\theta_{v|u}\bigr)\Bigr\} \times f_{H}\big(\textbf{b}_i\big)}}{\splitfrac{\sum_{h=1}^{\infty}\Bigl\{\gamma_h  \sum_{v=1}^{\infty} \gamma_{v|h} \prod_{t=1}^{T} \big\{p\bigl(m^{(1)}_{i,t},m^{(2)}_{i,t}\big| \textbf{d}_{i,t},\textbf{z}_{i,t}, \boldsymbol{\ell}_{i,0}, b_i^{M^{(2)}}, b_i^{M^{(1)}};\theta_{v|h}\bigr) }{ p\bigl(d_{i,t}\big| \textbf{z}_{i,t}, \boldsymbol{\ell}_{i,0}, b_i^{D};\theta_{v|h}\bigr)  p\bigl(z_{i,t}\big| \boldsymbol{\ell}_{i,0}, b_i^{Z};\theta_{v|h}\bigr)\bigr\} p\bigl(\boldsymbol{\ell}_{i,0};\theta_{v|h}\bigr)\Bigr\}\times f_{H}\big(\textbf{b}_i\big)}}
            \\
            &= \sum_{u=1}^{\infty}\frac{\splitfrac{\Bigl\{\gamma_u   \sum_{v=1}^{\infty} \gamma_{v|u}
            \prod_{t=1}^{T} \big\{ p\bigl(m^{(1)}_{i,t},m^{(2)}_{i,t}\big| \textbf{d}_{i,t},\textbf{z}_{i,t}, \boldsymbol{\ell}_{i,0}, b_i^{M^{(2)}}, b_i^{M^{(1)}};\theta_{v|u}\bigr)}{ p\bigl(d_{i,t}\big| \textbf{z}_{i,t}, \boldsymbol{\ell}_{i,0}, b_i^{D};\theta_{v|u}\bigr)p\bigl(z_{i,t}\big| \boldsymbol{\ell}_{i,0}, b_i^{Z};\theta_{v|u}\bigr)\bigr\} p\bigl(\boldsymbol{\ell}_{i,0};\theta_{v|u}\bigr)\Bigr\}}}{\splitfrac{\sum_{h=1}^{\infty}\Bigl\{\gamma_h  \sum_{v=1}^{\infty} \gamma_{v|h}
             \prod_{t=1}^{T} \big\{ p\bigl(m^{(1)}_{i,t},m^{(2)}_{i,t}\big|\textbf{d}_{i,t},\textbf{z}_{i,t},  \boldsymbol{\ell}_{i,0}, b_i^{M^{(2)}}, b_i^{M^{(1)}};\theta_{v|h}\bigr)  } {p\bigl(d_{i,t}\big| \textbf{z}_{i,t},\boldsymbol{\ell}_{i,0}, b_i^{D};\theta_{v|h}\bigr) p\bigl(z_{i,t}\big| \boldsymbol{\ell}_{i,0}, b_i^{Z};\theta_{v|h}\bigr) \bigr\} p\bigl(\boldsymbol{\ell}_{i,0};\theta_{v|h}\bigr)\Bigr\}}}  \times   p\bigl(y_i\big|\textbf{m}^{(1)}_{i,T},\textbf{m}^{(2)}_{i,T}, \textbf{d}_{i,T},\textbf{z}_{i,T}, \boldsymbol{\ell}_{i,0};\beta_u\bigr)
            \\
            &= \sum_{u = 1}^{\infty} w_u\bigl( \textbf{m}^{(1)}_{i,T},\textbf{m}^{(2)}_{i,T},\textbf{d}_{i,T},\textbf{z}_{i,T},\boldsymbol{\ell}_{i,0}\bigr) \times p\bigl(y_i\big| \textbf{m}^{(1)}_{i,T},\textbf{m}^{(2)}_{i,T}, \textbf{d}_{i,T},\textbf{z}_{i,T}, \boldsymbol{\ell}_{i,0};\beta_u\bigr),
        \end{aligned}
$}
    \end{equation}
    where
        \[\resizebox{\textwidth}{!}{$\displaystyle
    \begin{aligned}
        & w_u\bigl( \textbf{m}^{(1)}_{i,T},\textbf{m}^{(2)}_{i,T},\textbf{d}_{i,T},\textbf{z}_{i,T},\boldsymbol{\ell}_{i,0}\bigr)
        \\&
        = \frac{\Bigl\{\gamma_u   \sum_{v=1}^{\infty} \gamma_{v|u}
            \prod_{t=1}^{T} \big\{ p\bigl(m^{(1)}_{i,t},m^{(2)}_{i,t}\big| \textbf{d}_{i,t},\textbf{z}_{i,t}, \boldsymbol{\ell}_{i,0}, b_i^{M^{(2)}}, b_i^{M^{(1)}};\theta_{v|u}\bigr) p\bigl(d_{i,t}\big| \textbf{z}_{i,t},\boldsymbol{\ell}_{i,0}, b_i^{D};\theta_{v|u}\bigr)p\bigl(z_{i,t}\big| \boldsymbol{\ell}_{i,0}, b_i^{Z};\theta_{v|u}\bigr)\bigr\} p\bigl(\boldsymbol{\ell}_{i,0};\theta_{v|u}\bigr)\Bigr\}}{\sum_{h=1}^{\infty}\Bigl\{\gamma_h  \sum_{v=1}^{\infty} \gamma_{v|h}
             \prod_{t=1}^{T} \big\{ p\bigl(m^{(1)}_{i,t},m^{(2)}_{i,t}\big| \textbf{d}_{i,t},\textbf{z}_{i,t},\boldsymbol{\ell}_{i,0}, b_i^{M^{(2)}}, b_i^{M^{(1)}};\theta_{v|h}\bigr) p\bigl(d_{i,t}\big| \textbf{z}_{i,t},\boldsymbol{\ell}_{i,0}, b_i^{D};\theta_{v|h}\bigr)  p\bigl(z_{i,t}\big| \boldsymbol{\ell}_{i,0}, b_i^{Z};\theta_{v|h}\bigr) \bigr\} p\bigl(\boldsymbol{\ell}_{i,0};\theta_{v|h}\bigr)\Bigr\}}.
    \end{aligned}
    $}\]
Next, we derive the conditional density $f\bigl( m^{(1)}_{i,T}, m^{(2)}_{i,T} \big|\textbf{m}^{(1)}_{i,T-1}, \textbf{m}^{(2)}_{i,T-1},\textbf{d}_{i,T},\textbf{z}_{i,T},\boldsymbol{\ell}_{i,0}, \textbf{b}_i\bigr)$: 
        \begin{equation}
    \resizebox{\textwidth}{!}{$\displaystyle
        \begin{aligned}
            & f_{H}\bigl( m^{(1)}_{i,T}, m^{(2)}_{i,T} \big|\textbf{m}^{(1)}_{i,T-1}, \textbf{m}^{(2)}_{i,T-1},\textbf{d}_{i,T},\textbf{z}_{i,T},\boldsymbol{\ell}_{i,0}, \textbf{b}_i\bigr) \\
            &= \frac{f_{H}\bigl(\textbf{m}^{(1)}_{i,T},\textbf{m}^{(2)}_{i,T},\textbf{d}_{i,T},\textbf{z}_{i,T},\boldsymbol{\ell}_{i,0}, \textbf{b}_i\bigr)}{f_{H}\bigl( \textbf{m}^{(1)}_{i,T-1}, \textbf{m}^{(2)}_{i,T-1},\textbf{d}_{i,T},\textbf{z}_{i,T},\boldsymbol{\ell}_{i,0}, \textbf{b}_i\bigr)}
            \\
            &= \frac{\splitfrac{\sum_{u=1}^{\infty}\Bigl\{\gamma_u   \sum_{v=1}^{\infty} \gamma_{v|u}
            \prod_{t=1}^{T} \big\{ p\bigl(m^{(1)}_{i,t},m^{(2)}_{i,t}\big| \textbf{d}_{i,t},\textbf{z}_{i,t}, \boldsymbol{\ell}_{i,0}, b_i^{M^{(2)}}, b_i^{M^{(1)}};\theta_{v|u}\bigr)p\bigl(d_{i,t}\big| \textbf{z}_{i,t}, \boldsymbol{\ell}_{i,0}, b_i^{D};\theta_{v|u}\bigr)}{   p\bigl(z_{i,t}\big| \boldsymbol{\ell}_{i,0}, b_i^{Z};\theta_{v|u}\bigr) \bigr\} p\bigl(\boldsymbol{\ell}_{i,0};\theta_{v|u}\bigr)\Bigr\} \times f_{H}\big(\textbf{b}_i\big)}}{\splitfrac{\sum_{h=1}^{\infty}\Bigl\{\gamma_h  \sum_{v=1}^{\infty} \gamma_{v|h} \prod_{s=1}^{T-1} \big\{ p\bigl(m^{(1)}_{i,s},m^{(2)}_{i,s}\big| \textbf{d}_{i,s},\textbf{z}_{i,s}, \boldsymbol{\ell}_{i,0}, b_i^{M^{(2)}}, b_i^{M^{(1)}};\theta_{v|h}\bigr) \big\} \prod_{t=1}^{T}\big\{ p\bigl(d_{i,t}\big| \textbf{z}_{i,t},\boldsymbol{\ell}_{i,0}, b_i^{D};\theta_{v|h}\bigr)  p\bigl(z_{i,t}\big| \boldsymbol{\ell}_{i,0}, b_i^{Z};\theta_{v|h}\bigr) \bigr\}}{ p\bigl(\boldsymbol{\ell}_{i,0};\theta_{v|h}\bigr)\Bigr\}\times f_{H}\big(\textbf{b}_i\big)}}
            \\
            &= \sum_{u=1}^{\infty}\frac{\splitfrac{\Bigl\{\gamma_u \sum_{v=1}^{\infty}\gamma_{v|u}\prod_{s=1}^{T-1} \big\{ p\bigl(m^{(1)}_{i,s},m^{(2)}_{i,s}\big| \textbf{d}_{i,s},\textbf{z}_{i,s}, \boldsymbol{\ell}_{i,0}, b_i^{M^{(2)}}, b_i^{M^{(1)}};\theta_{v|u}\bigr)\big\}}{\prod_{t=1}^{T} \big\{ p\bigl(d_{i,t}\big| \textbf{z}_{i,t}, \boldsymbol{\ell}_{i,0}, b_i^{D};\theta_{v|u}\bigr)p\bigl(z_{i,t}\big| \boldsymbol{\ell}_{i,0}, b_i^{Z};\theta_{v|u}\bigr) \bigr\} p\bigl(\boldsymbol{\ell}_{i,0};\theta_{v|u}\bigr)\Bigr\}}}{\splitfrac{\sum_{h=1}^{\infty}\Bigl\{\gamma_h \sum_{g=1}^{\infty} \gamma_{g|h}\prod_{s=1}^{T-1} \big\{ p\bigl(m^{(1)}_{i,s},m^{(2)}_{i,s}\big| \textbf{d}_{i,s},\textbf{z}_{i,s}, \boldsymbol{\ell}_{i,0}, b_i^{M^{(2)}}, b_i^{M^{(1)}};\theta_{g|h}\bigr) \big\}  }{\prod_{t=1}^{T} \big\{p\bigl(d_{i,t}\big| \textbf{z}_{i,t},\boldsymbol{\ell}_{i,0}, b_i^{D};\theta_{g|h}\bigr) p\bigl(z_{i,t}\big| \boldsymbol{\ell}_{i,0}, b_i^{Z};\theta_{g|h}\bigr) \bigr\} p\bigl(\boldsymbol{\ell}_{i,0};\theta_{g|h}\bigr)\Bigr\}}}
            \times p\bigl(m^{(1)}_{i,T},m^{(2)}_{i,T}\big|\textbf{d}_{i,T},\textbf{z}_{i,T}, \boldsymbol{\ell}_{i,0},b_i^{M^{(2)}}, b_i^{M^{(1)}};\theta_{v|u}\bigr)
            \\
            &= \sum_{u = 1}^{\infty} w_{u,v}\bigl( \textbf{m}^{(1)}_{i,T-1}, \textbf{m}^{(2)}_{i,T-1},\textbf{d}_{i,T},\textbf{z}_{i,T},\boldsymbol{\ell}_{i,0}\bigr) \times p\bigl(m^{(1)}_{i,T},m^{(2)}_{i,T}\big|\textbf{d}_{i,T},\textbf{z}_{i,T},\boldsymbol{\ell}_{i,0}, b_i^{M^{(2)}}, b_i^{M^{(1)}};\theta_{v|u}\bigr)
        \end{aligned}
    $}
    \end{equation}
    where
        \[\resizebox{\textwidth}{!}{$\displaystyle
    \begin{aligned}
        & w_{u,v}\bigl( \textbf{m}^{(1)}_{i,T-1}, \textbf{m}^{(2)}_{i,T-1},\textbf{d}_{i,T},\textbf{z}_{i,T},\boldsymbol{\ell}_{i,0}\bigr)
        \\&
        = \frac{\Bigl\{\gamma_u \sum_{v=1}^{\infty} \gamma_{v|u}\prod_{s=1}^{T-1} \big\{ p\bigl(m^{(1)}_{i,s},m^{(2)}_{i,s}\big| \textbf{d}_{i,s},\textbf{z}_{i,s}, \boldsymbol{\ell}_{i,0}, b_i^{M^{(2)}}, b_i^{M^{(1)}};\theta_{v|u}\bigr)\big\}\prod_{t=1}^{T} \big\{ p\bigl(d_{i,t}\big| \textbf{z}_{i,t}, \boldsymbol{\ell}_{i,0}, b_i^{D};\theta_{v|u}\bigr)p\bigl(z_{i,t}\big| \boldsymbol{\ell}_{i,0}, b_i^{Z};\theta_{v|u}\bigr) \bigr\} p\bigl(\boldsymbol{\ell}_{i,0};\theta_{v|u}\bigr)\Bigr\}}{\sum_{h=1}^{\infty}\Bigl\{\gamma_h \sum_{g=1}^{\infty} \gamma_{g|h}\prod_{s=1}^{T-1} \big\{ p\bigl(m^{(1)}_{i,s},m^{(2)}_{i,s}\big| \textbf{d}_{i,s},\textbf{z}_{i,s}, \boldsymbol{\ell}_{i,0}, b_i^{M^{(2)}}, b_i^{M^{(1)}};\theta_{g|h}\bigr) \big\}  \prod_{t=1}^{T} \big\{p\bigl(d_{i,t}\big| \textbf{z}_{i,t},\boldsymbol{\ell}_{i,0}, b_i^{D};\theta_{g|h}\bigr) p\bigl(z_{i,t}\big| \boldsymbol{\ell}_{i,0}, b_i^{Z};\theta_{g|h}\bigr) \bigr\} p\bigl(\boldsymbol{\ell}_{i,0};\theta_{g|h}\bigr)\Bigr\}}.
    \end{aligned}
    $}\]
Similarly, we derive the conditional densities:
\begin{itemize}
    \item $f_{H}\bigl(d_{i,T} \big| \mathbf{m}^{(1)}_{i,T-1}, \mathbf{m}^{(2)}_{i,T-1}, \mathbf{d}_{i,T-1}, \mathbf{z}_{i,T}, \boldsymbol{\ell}_{i,0}, \mathbf{b}_i\bigr)$,
    \item $f_{H}\bigl(m^{(1)}_{i,T-1}, m^{(2)}_{i,T-1} \big| \mathbf{m}^{(1)}_{i,T-2}, \mathbf{m}^{(2)}_{i,T-2}, \mathbf{d}_{i,T-1}, \mathbf{z}_{i,T-1}, \boldsymbol{\ell}_{i,0}, \mathbf{b}_i\bigr)$,
    \item $f_{H}\bigl(d_{i,T-1} \big| \mathbf{m}^{(1)}_{i,T-2}, \mathbf{m}^{(2)}_{i,T-2}, \mathbf{d}_{i,T-2}, \mathbf{z}_{i,T-1}, \boldsymbol{\ell}_{i,0}, \mathbf{b}_i\bigr)$, and so on.
\end{itemize}
These conditional densities appear in the nonparametric identification of the causal parameter $\theta(\mathbf{z}, \mathbf{z}_*)$ in Proposition~\ref{prop1} of the main text.

\section{EDPM truncation approximation}\label{sec:S5}

Let
\begin{equation}
\resizebox{\textwidth}{!}{$\displaystyle
    \begin{aligned}
       Y_i \big| \textbf{M}^{(2)}_{i,T}, \textbf{M}^{(1)}_{i,T}, \textbf{D}_{i,T}, \textbf{Z}_{i,T}, \textbf{L}_{i,0}; \boldsymbol{\beta}_i & \sim F_y\big(. \big| \textbf{m}^{(2)}_{i,T}, \textbf{m}^{(1)}_{i,T}, \textbf{d}_{i,T}, \textbf{z}_{i,T}, \boldsymbol{\ell}_{i,0}; \boldsymbol{\beta}_i \big) 
       \\
        \big(M^{(2)}_{i,t}, M^{(1)}_{i,t}\big)\big|\textbf{M}^{(2)}_{i,t-1}, \textbf{M}^{(1)}_{i,t-1}, \textbf{D}_{i,t}, \textbf{Z}_{i,t}, \textbf{L}_{i,0},b_i^{M^{(2)}}, b_i^{M^{(1)}}; \theta^{M^{(2)}}_{i,t}, \boldsymbol{\theta}^{M^{(1)}}_{i,t}
        & \sim F_{(m^{(2)}_{t}, m^{(1)}_{t})} \big(.\big|\textbf{m}^{(2)}_{i,t-1}, \textbf{m}^{(1)}_{i,t-1}, \textbf{d}_{i,t}, \textbf{z}_{i,t}, \boldsymbol{\ell}_{i,0}, b_i^{M^{(2)}}, b_i^{M^{(1)}}; \boldsymbol{\theta}^{M^{(2)}}_{i,t}, \theta^{M^{(1)}}_{i,t} \big) 
        \\
        D_{i,t}\big|\textbf{M}^{(2)}_{i,t-1}, \textbf{M}^{(1)}_{i,t-1}, \textbf{D}_{i,t-1}, \textbf{Z}_{i,t}, \textbf{L}_{i,0}, b_i^{D} ; \boldsymbol{\theta}^{D}_{i,t} &\sim F_{d_{t}}\big(. \big|\textbf{m}^{(2)}_{i,t-1}, \textbf{m}^{(1)}_{i,t-1}, \textbf{d}_{i,t-1}, \textbf{z}_{i,t}, \boldsymbol{\ell}_{i,0}, b_i^{D}; \boldsymbol{\theta}^{D}_{i,t} \big) 
        \\
        Z_{i,t}\big|\textbf{M}^{(2)}_{i,t-1}, \textbf{M}^{(1)}_{i,t-1}, \textbf{D}_{i,t-1}, \textbf{Z}_{i,t-1}, \textbf{L}_{i,0}, b_i^{Z}; \boldsymbol{\theta}^{Z}_{i,t} &\sim F_{z_{t}}\big(. \big|\textbf{m}^{(2)}_{i,t-1}, \textbf{m}^{(1)}_{i,t-1}, \textbf{d}_{i,t-1}, \textbf{z}_{i,t-1}, \boldsymbol{\ell}_{i,0}, b_i^{Z}; \boldsymbol{\theta}^{Z}_{i,t} \big) 
        \\
        \textbf{L}_{i,0}^{(k)};\boldsymbol{\theta}^{L_0^{(k)}}_{i} &\sim F_{\ell_{0}^{(k)}}\big(.\big| \boldsymbol{\theta}^{L_0^{(k)}}_{i} \big), \quad k = 1, \ldots, K \\
        (\boldsymbol{\beta}_i, \boldsymbol{\theta}_{i})|H &\sim H  \\
        H &\sim \mathcal{H}_{NM}
    \end{aligned}
$}
\end{equation}
In the specification above, $H \sim \mathcal{H}_{NM}$ implies
    $$ H = \sum_{r=1}^{N} \sum_{s=1}^{M}\xi_{r} \xi_{s|r} \delta_{\beta_r^{*}\theta_{s|r}^{*}}$$
    where:
    \begin{align*}
        \xi_r &= \xi_r^{'}\prod_{t <r} (1- \xi_t^{'}) \quad \quad \quad \xi_t^{'} \sim \text{Beta}(1, \alpha^{\beta}) \quad \quad \quad \beta^{*}_r \mathop{\sim}\limits^{iid} H_{0\beta} \\
        \xi_{s|r} &= \xi_{s|r}^{'}\prod_{t <s} (1- \xi_{t|r}^{'}) \quad \quad \quad \xi_{t|r}^{'} \sim \text{Beta}(1, \alpha_r^{\theta|\beta}) \quad \quad \quad \theta^{*}_{s|r} \mathop{\sim}\limits^{iid} H_{0\theta|\beta}.
    \end{align*}
    Like before, we can write the EDPM model specified above as an infinite mixture:
    \begin{equation}
\resizebox{\textwidth}{!}{$\displaystyle
        \begin{aligned}
            & f_{H}\bigl(y_i, \textbf{m}^{(1)}_{i,T},\textbf{m}^{(2)}_{i,T},\textbf{d}_{i,T},\textbf{z}_{i,T},\boldsymbol{\ell}_{i,0}, \textbf{b}_i\bigr)
            \\& = \sum_{r=1}^{N}\Bigl\{\xi_r p\bigl(y_i\big|\textbf{m}^{(1)}_{i,T},\textbf{m}^{(2)}_{i,T},  \textbf{d}_{i,T},\textbf{z}_{i,T}, \boldsymbol{\ell}_{i,0};\beta_r\bigr) \times \sum_{s=1}^{M} \xi_{s|r}
            \prod_{t = 1}^{T} \big\{
            p\bigl(m^{(1)}_{i,t},m^{(2)}_{i,t}\big| \textbf{z}_{i,t}, \textbf{d}_{i,t}, \boldsymbol{\ell}_{i,0}, b_i^{M^{(2)}}, b_i^{M^{(1)}};\theta_{s|r}\bigr)  \times
             p\bigl(d_{i,t}\big| \textbf{z}_{i,t}, \boldsymbol{\ell}_{i,0}, b_i^{D};\theta_{s|r}\bigr)
             \times  \\& p\bigl(z_{i,t}\big| \boldsymbol{\ell}_{i,0}, b_i^{Z};\theta_{s|r}\bigr) \big\} \times p\bigl( \boldsymbol{\ell}_{i,0};\theta_{s|r}\bigr)\Bigr\}\times f_{H}\big(\textbf{b}_i\big).
        \end{aligned}
$}
    \end{equation}
\cite{burns2023truncation} show that $\mathcal{H}_{NM}$ converges almost surely to an enriched Dirichlet process with base distribution $H_{0\beta} \times H_{0\theta|\beta}$ and precision parameters  $\alpha^{\beta}$ and $\{\alpha_r^{\theta|\beta}\}_{r=1}^{N}$, respectively. The term  $\beta$-cluster indicates top-level clusters based on the parameters of the outcome model. Similarly, the term  $\theta$-cluster denotes a subcluster (based on parameters other than the ones in the outcome model) nested within a $\beta$-cluster.

\subsection{Blocked Gibbs sampler}
The blocked Gibbs sampler described in this section is implemented by fitting the proposed models in \texttt{Nimble}. We define $V_{i} = (V_{i}^{\beta}, V_{i}^{\theta})$ to be the cluster membership assignment vector for subject $i$. Let $\textbf{V}^{\beta} = \big\{V_{i}^{\beta} \big\}_{i=1}^{n}$ and $\textbf{V}^{\theta}= \big\{V_{i}^{\theta} \big\}_{i=1}^{n}$ denote cluster membership assignment vectors across subjects at $\beta-$level and $\theta-$level, respectively.  At each iteration, we perform the following updates iteratively:
\begin{enumerate}
            \item Update cluster membership, i.e., sample from the conditional distribution of \\ $\textbf{V}^{\beta}, \textbf{V}^{\theta}\big|\boldsymbol{\xi}_{r}, \boldsymbol{\xi}_{s|r}, \boldsymbol{\beta}^{*}, \boldsymbol{\theta}^{*}, \textbf{b}^{M^{(1)}}, \textbf{b}^{M^{(2)}}, \textbf{b}^{D}, \textbf{b}^{Z}, \textbf{Y}, \textbf{M}^{(1)}, \textbf{M}^{(2)}, \textbf{D}, \textbf{Z}, \textbf{L}_0$: \\
            For each subject $i$, $V^{\beta}_i, V^{\theta}_i\big|\boldsymbol{\xi}_{r}, \boldsymbol{\xi}_{s|r}, \boldsymbol{\beta}^{*}, \boldsymbol{\theta}^{*}, b_i^{M^{(1)}}, b_i^{M^{(2)}}, b_i^{D}, b_i^{Z}, y_i, \textbf{m}_{i,T}^{(1)}, \textbf{m}_{i,T}^{(2)}, \textbf{d}_{i,T}, \textbf{z}_{i,T}, \boldsymbol{\ell}_{i0}$ is sampled from a multinomial  distribution with the  probability that subject $i$ is assigned to $\beta$-cluster $r$ (of the $N$ $\beta$-clusters) and $\theta$-cluster $s$ (of the $M$ $\theta$-clusters) is proportional to
            \begin{align*}
                p_{i,r, s|r}   &\propto \xi_r \xi_{s|r} p\bigl(y_i\big|\textbf{m}^{(1)}_{i,T},\textbf{m}^{(2)}_{i,T}, \textbf{d}_{i,T},\textbf{z}_{i,T}, \boldsymbol{\ell}_{i,0}; \boldsymbol{\beta}_r^{*}\bigr) \times \\& \Big\{\prod_{t=1}^{T} p\bigl(m^{(1)}_{i,t}, m^{(2)}_{i,t} \big|\textbf{z}_{i,t}, \textbf{d}_{i,t},\boldsymbol{\ell}_{i,0},b_i^{M^{(1)}}, b_i^{M^{(2)}}; \boldsymbol{\theta}_{s|r}^{M^{(1)},*}, \sigma^{2,M^{(1)},*}_{s|r}, \boldsymbol{\theta}_{s|r}^{M^{(2)},*},
                    \sigma^{2,M^{(2)},*}_{s|r}\bigr)
                    \\& p\bigl(d_{i,t} \big|\textbf{z}_{i,t},\boldsymbol{\ell}_{i,0},b_i^{D}; \boldsymbol{\theta}_{s|r}^{D,*}\bigr) p\bigl(z_{i,t} \big|\boldsymbol{\ell}_{i,0},b_i^{Z}; \boldsymbol{\theta}_{s|r}^{Z,*}\bigr)\Big\}
                \times p(\boldsymbol{\ell}_{i,0};  \theta_{s|r}^{L_0,*})
            \end{align*}
            \item Update regression parameters at the $\beta$-level cluster, i.e., sample from the conditional distribution of $\boldsymbol{\beta}^{*}\big|\textbf{V}^{\beta}, \alpha^{\beta},\textbf{Y}$: \\
            For each $r = 1, \ldots, N$,  draw an observation from the conditional distribution 
            \[\resizebox{\textwidth}{!}{$\displaystyle
            \begin{aligned}
                f\bigl(\boldsymbol{\beta}_{r}^{Y,*}, \sigma_{r}^{Y,2*}, \pi_{r}^{Y,*} \big| \textbf{V}^{\beta},\alpha^{\beta},\textbf{Y} \bigr) & \propto \underbrace{H_{0\beta}(d\beta_{1,r}^{Y,*}) \times \ldots \times H_{0\beta}(d\beta_{\text{num\_Y\_cov},r}^{Y,*}) \times H_{0\beta}(d\sigma_{r}^{Y,2*})}_{Y \text{ regression parameter priors}} \times \underbrace{H_{0\beta}(d\pi_{r}^{Y,*})}_{ Y \text{ hurdle probability prior}} \\&
                \underbrace{\prod_{i: V_i^{\beta} = r}f\big(y_i\big|\beta_{r}^{Y,*},\sigma_{r}^{Y,2*}, \pi_{r}^{Y,*}\big)}_{\text{likelihood contribution}},
            \end{aligned}
            $}\]
            where $\text{num\_Y\_cov}$ is the number of covariates in the outcome model within clusters. The parameters of the outcome model are updated using conjugate distributions: specifically, the mean, variance, and hurdle probability parameters are updated from their respective conjugate normal, inverse-gamma, and beta distributions.

            \item Update regression parameters at the $\theta$-level cluster, i.e., sample from the conditional distribution of  \\ $\boldsymbol{\theta}^{*}\big|\textbf{V}^{\beta}, \textbf{V}^{\theta}, \{\alpha_r^{\theta|\beta}\}_{r=1}^{N}, \textbf{b}^{M^{(1)}}, \textbf{b}^{M^{(2)}}, \textbf{b}^{D}, \textbf{b}^{Z}, \textbf{M}^{(1)}, \textbf{M}^{(2)}, \textbf{D}, \textbf{Z}, \textbf{L}_0$:
              \[\resizebox{\textwidth}{!}{$\displaystyle
              \begin{aligned}
                    & f\bigl( \boldsymbol{\theta}_{s|r,t}^{M^{(1)},*}, \sigma^{2,M^{(1)},*}_{s|r,t}, \pi_{s|r,t}^{M^{(1)},*}, \boldsymbol{\theta}_{s|r,t}^{M^{(2)},*},
                    \sigma^{2,M^{(2)},*}_{s|r,t},
                    \pi_{s|r,t}^{M^{(2)},*}, \boldsymbol{\theta}_{s|r,t}^{D,*},\boldsymbol{\theta}_{s|r,t}^{Z,*}, \boldsymbol{\theta}_{s|r}^{L_0,*} \big|   \textbf{V}^{\beta}, \textbf{V}^{\theta}, \alpha_r^{\theta|\beta},\textbf{b}^{M^{(1)}}, \textbf{b}^{M^{(2)}}, \textbf{b}^{D}, \textbf{b}^{Z},
                    \textbf{M}^{(1)}, \textbf{M}^{(2)}, \textbf{D}, \textbf{Z},\textbf{L}_0 \bigr)
                     \\&  \propto \underbrace{H_{0\theta|\beta}(d\theta_{s|r,t,1}^{M^{(1)},*}) \times \ldots \times H_{0\theta|\beta}(d\theta_{s|r,t,Q}^{M^{(1)},*}) \times H_{0\theta|\beta}(d\sigma_{s|r,t}^{2,M^{(1)},*})}_{M^{(1)}_t \text{ regression parameter priors}} \times
                     \underbrace{H_{0\theta|\beta}(d\pi_{s|r,t}^{M^{(1)},*})}_{M^{(1)}_t \text{ hurdle probability priors}} \times
                     \\& \underbrace{H_{0\theta|\beta}(d\theta_{s|r,t,1}^{M^{(2)},*}) \times \ldots \times H_{0\theta|\beta}(d\theta_{s|r,t,Q}^{M^{(2)},*}) \times H_{0\theta|\beta}(d\sigma_{s|r,t}^{2,M^{(2)},*})}_{M^{(2)}_t \text{ regression parameter priors}} \times
                     \underbrace{ H_{0\theta|\beta}(d\pi_{s|r,t}^{M^{(2)},*})}_{M^{(2)}_t \text{ hurdle probability priors}} \times
                     \\&
                     \underbrace{H_{0\theta|\beta}(d\theta_{s|r,t,1}^{D,*}) \times \ldots \times H_{0\theta|\beta}(d\theta_{s|r,t,Q}^{D,*}) }_{D_t \text{ regression parameter priors}} \times \underbrace{H_{0\theta|\beta}(d\theta_{s|r,t,1}^{Z,*}) \times \ldots \times H_{0\theta|\beta}(d\theta_{s|r,t,Q}^{Z,*}) }_{Z_t \text{ regression parameter priors}}  \times \underbrace{H_{0\theta|\beta}\big(d\theta_{s|r,1}^{L_{0},*}\big)  \times \ldots \times H_{0\theta|\beta}\big(d\theta_{s|r,K}^{L_{0},*}\big)}_{L_0 \text{ model parameter priors}}
                     \\&
                     \underbrace{\prod_{i: V_i^{\beta} = r,V_i^{\theta} = s} f\bigl(\textbf{m}^{(1)}_{i,T}, \textbf{m}^{(2)}_{i,T}, \textbf{d}_{i,T}, \textbf{z}_{i,T},\boldsymbol{\ell}_{i,0},b_i^{M^{(1)}}, b_i^{M^{(2)}},b_i^{D}, b_i^{Z}; \boldsymbol{\theta}_{s|r,t}^{M^{(1)},*}, \sigma^{2,M^{(1)},*}_{s|r,t},
                     \pi_{s|r,t}^{M^{(1)},*}, \boldsymbol{\theta}_{s|r,t}^{M^{(2)},*},
                    \sigma^{2,M^{(2)},*}_{ s|r,t},
                    \pi_{s|r,t}^{M^{(2)},*}, \boldsymbol{\theta}_{s|r,t}^{D,*},\boldsymbol{\theta}_{s|r,t}^{Z,*},\boldsymbol{\theta}_{s|r}^{L_0,*}\bigr)}_{\text{likelihood contribution}},
                \end{aligned}
                $}\]
                For each $r = 1, \ldots, N$ (indexing $\beta$-level parameters), $s = 1, \ldots, M$ (indexing $\theta$-level parameters within $r$), $t = 1, \ldots, T$ (indexing measurement times for subject $i$), the hurdle probabilities for both mediators, $\pi_{s|r,t}^{M^{(1)},*}$ and $\pi_{s|r,t}^{M^{(2)},*}$,  are updated from their respective conjugate beta distributions.
                For $k = 1, \ldots, K $ (indexing number of baseline covariates), the parameters in the regression models for the longitudinal variables are updated using a Gibbs sampler (conjugate update) if the corresponding data is continuous, or a Metropolis-Hastings step if the data is binary.

                For example, assuming, without loss of generality, that we have continuous time-varying mediators, the regression parameters $\big(\boldsymbol{\theta}_{s|r}^{M^{(1)}}, \boldsymbol{\theta}_{s|r}^{M^{(2)}}\big)$,  along with the variance parameters $\big(\boldsymbol{\sigma}_{s|r}^{2,M^{(1)}}, \boldsymbol{\sigma}_{s|r}^{2,M^{(2)}}\big)$  from the mediator models are updated from conjugate normal and inverse-gamma distributions, respectively. For binary time-varying treatment assignment and receipt statuses, the parameters from the corresponding probit regressions are updated using a random-walk Metropolis-Hastings algorithm with normal proposal distributions.

                The baseline covariate model parameters are updated using conjugate normal distributions for continuous baseline covariate model parameters and beta distributions for binary baseline covariate model parameters, with the posterior distribution proportional to:
                \begin{align*}
                    & f\bigl( \boldsymbol{\theta}_{s|r}^{L_{0},*}  \big|  \textbf{V}^{\beta}, \textbf{V}^{\theta}, \alpha_r^{\theta|\beta}, \textbf{L}_0 \bigr)
                      \propto    \underbrace{H_{0\theta|\beta}(d\boldsymbol{\theta}_{s|r}^{L_{0},*} )}_{\text{prior}}  \times
                      \underbrace{\prod_{i: V_i^{\beta} = r,V_i^{\theta} = s}
                     f\big(\boldsymbol{\ell}_{i,0}\big|  \theta_{s|r}^{L_0,*}\big)}_{\text{likelihood contribution}},
                \end{align*}
                where the corresponding likelihood contributions are from all subjects $i$ in outer cluster $r$ and inner cluster $s$.

            \item Update the weights at the $\beta$-level, i.e., sample from the conditional distribution of $\boldsymbol{\xi}_{r}|\textbf{V}^{\beta}, \alpha^{\beta}$  using $\xi_1= \xi_1^{'}, \xi_r= \xi^{'}_r\prod_{t =1}^{r-1} (1- \xi_t^{'}), r = 2, \ldots, N$ where 
            $$ \xi_{\tilde{r}}^{'}|\textbf{V}^{\beta}, \alpha^{\beta} \sim \text{Beta} \Bigl(n_{\tilde{r}}+1, \alpha^{\beta} + \sum_{w = \tilde{r}+1}^{N}n_w\Bigr) $$
             with $n_{\tilde{r}}$ denoting the number of subjects currently in the $\tilde{r}$th $\beta-$cluster for $\tilde{r} = 1, \ldots, N-1$ and $\xi_N^{'} = 1$.
             
            \item Update the weights at the $\theta$-level, i.e., sample from the conditional distribution of $\boldsymbol{\xi}_{s|r}|\textbf{V}^{\beta}, \textbf{V}^{\theta}, \alpha_r^{\theta|\beta}$ , using  $\xi_{1|r} = \xi_{1|r}^{'}, \xi_{s|r} = \xi_{s|r}^{'}\prod_{t =1}^{s-1} (1- \xi_{t|r}^{'}) ,s = 2,\ldots, M$  for every $r = 1,\ldots, N$ where
            $$\xi_{\tilde{s}|r}^{'}|\textbf{V}^{\theta}, \alpha_{r}^{\theta|\beta} \sim \text{Beta} \Bigl(n_{r\tilde{s}} +1, \alpha_{r}^{\theta|\beta} + \sum_{w = \tilde{s}+1}^{M} n_{rw} \Bigr) $$
            with $n_{r\tilde{s}}$ denoting the number of subjects currently in the $\tilde{s}$th $\theta-$cluster within the $r$th $\beta-$cluster for $\tilde{s} = 1, \ldots, M-1$ and $\xi_{M|r}^{'} = 1$.
            
            \item Update the concentration parameter at the $\beta$-level, i.e., for a $\text{Gamma}(a_{\beta},b_{\beta})$ prior on $\alpha^{\beta}$ (parameterized so that $b_{\beta}$ is a rate parameter), sample from the conditional distribution
            $$\alpha^{\beta}|\boldsymbol{\xi}_{r} \sim \text{Gamma}\Bigl(N + a_{\beta} -1, b_{\beta} - \sum_{r =1}^{N-1} \log(1-\xi_{r}^{'})\Bigr)$$
            
            \item Update the concentration parameter at the $\theta$-level, i.e., for a $\text{Gamma}(a_{\theta},b_{\theta})$ prior on $\alpha_r^{\theta|\beta}$ for each $r = 1,\ldots, N$, sample from the conditional distribution $$\alpha_r^{\theta|\beta}|\boldsymbol{\xi}_{s|r} \sim \text{Gamma}\Bigl(M + a_{\theta} -1, b_{\theta} - \sum_{s =1}^{M-1} \log(1-\xi_{s|r}^{'}) \Bigr).$$

            \item Update the random effects ($ b_i^{M^{(1)}}, b_i^{M^{(2)}}, b_i^{D},b_i^{Z}$): \\
             At each MCMC iteration, conditional on the data, the cluster memberships $(V_i^{\beta}, V_i^{\theta})$, and the regression parameters $\big(\boldsymbol{\theta}_{s|r}^{M^{(1)},*},  \boldsymbol{\sigma}_{s|r}^{2, M^{(1)},*}, \boldsymbol{\theta}_{s|r}^{M^{(2)},*},  \boldsymbol{\sigma}_{s|r}^{2, M^{(2)},*}\big)$,  the new random intercepts $\big(b_i^{M^{(1)},*}, b_i^{M^{(2)},*}\big)$ from the mediator (continuous data) models for subject $i$ are updated from the corresponding conjugate Normal distributions after taking the residuals from the current fit. Note that if we do not condition on the cluster memberships $(V_i^{\beta}, V_i^{\theta})$, the random intercepts $\big(b_i^{M^{(1)},*}, b_i^{M^{(2)},*}\big)$  must instead be updated using a Metropolis-Hastings step, since the mediators (data) are modeled as a (finite) mixture of normal distributions at the global level. The random intercept variance parameters $\sigma_{b}^{2, M^{(1)},*}$ and $\sigma_{b}^{2, M^{(2)},*}$ are updated from conjugate Inverse-Gamma distributions:
             \begin{align*}
                 \sigma_{b}^{2, M^{(1)},*} &\sim \text{Inverse-Gamma} \big(\text{shape} = \alpha_{\sigma_{b}^{2,M^{(1)}}} + \frac{n}{2}, \text{rate} = \beta_{\sigma_{b}^{2,M^{(1)}}} + \frac{1}{2} \sum_{i=1}^{n}b_i^{2,M^{(1)},*}\big),
                 \\
                 \sigma_{b}^{2, M^{(2)},*} &\sim \text{Inverse-Gamma} \big(\text{shape} = \alpha_{\sigma_{b}^{2,M^{(2)}}} + \frac{n}{2}, \text{rate} = \beta_{\sigma_{b}^{2,M^{(2)}}} + \frac{1}{2} \sum_{i=1}^{n}b_i^{2,M^{(2)},*}\big),
             \end{align*} 
              where, for $j = 1,2$, $\alpha_{\sigma_{b}^{2,M^{(j)}}}$ and $\beta_{\sigma_{b}^{2,M^{(j)}}}$ are the shape and rate parameters, respectively, from the prior distribution of $\sigma_{b}^{2,M^{(j)}}$.

            Similarly, for subject $i$, the new random intercepts $ b_i^{Z,*}$ and $b_i^{D,*}$ from the treatment assignment and receipt status (binary data) models, respectively,  are updated using a Metropolis-Hastings step, conditional on the data, the cluster memberships $(V_i^{\beta}, V_i^{\theta})$, and the regression parameters $(\boldsymbol{\theta}_{s|r}^{D}, \boldsymbol{\theta}_{s|r}^{Z})$. This is because the binary data are modeled locally using probit regressions. Given these random intercepts, 
            the variances $ \sigma_{b}^{2,D}$, and  $\sigma_{b}^{2,Z}$ are updated using a Metropolis-Hastings step as well.
 \end{enumerate}

\subsection{G-computation algorithm}\label{sec:Salg}

Algorithm~\ref{Alg1:Gcomp} describes the steps to draw a posterior sample for the parameter $\theta(\textbf{z},\textbf{z}_{*})$ using Monte Carlo integration. Given $Q$ posterior samples of EDPM parameters, the algorithm is repeated $Q$ times to obtain $Q$ posterior samples of $\theta(\textbf{z},\textbf{z}_{*})$. 

\clearpage

\begin{breakablealgorithm}
\caption{G-computation algorithm: compute one posterior sample for $\theta(\mathbf{z},\mathbf{z}_{*})$}
\label{Alg1:Gcomp}

{\singlespacing\small
\setlength{\abovedisplayskip}{3pt}
\setlength{\belowdisplayskip}{3pt}
\setlength{\abovedisplayshortskip}{2pt}
\setlength{\belowdisplayshortskip}{2pt}
\resizebox{\textwidth}{!}{%
\begin{minipage}{\textwidth}
    \begin{enumerate}\setlength{\itemsep}{2pt}\setlength{\parskip}{0pt}
            \item At the  $(B+q)^{\text{th}}$ iteration, randomly draw $C^{*}$ (row) vectors as samples for $\textbf{L}_{0}$, say $\boldsymbol{\ell}^{c}_0$, $c \in \{1,\ldots, C^{*} \}$, from
             $f_{L_0}\big(\boldsymbol{\ell}_0 \big|\beta^{(q)}, \theta^{(q)}\big) =  \sum_{r=1}^{N}\xi^{(q)}_r \sum_{s=1}^{M} \xi^{(q)}_{s|r} \times p\bigl( \boldsymbol{\ell}_0;\theta^{(q)}_{s|r}\bigr).$
             \item Given the $(B+q)^{\text{th}}$ posterior samples of the random effect variances,  $\sigma_{b}^{2,M^{(2)},(q)}$, $\sigma_{b}^{2,M^{(1)},(q)}$, $\sigma_{b}^{2,D,(q)}$, and $\sigma_{b}^{2,Z,(q)}$, randomly draw $C^{*}$ sets of random effects---corresponding to $C^{*}$ Monte Carlo samples---from $N\big(0, \sigma_{b}^{2,M^{(2)},(q)}\big), N\big(0, \sigma_{b}^{2,M^{(1)},(q)}\big),  N\big(0, \sigma_{b}^{2,D,(q)}\big),  \text{ and }  N\big(0, \sigma_{b}^{2,Z,(q)}\big)$ for $\big\{b^{M^{(2)},(q),c}\big\}_{c =1}^{C^{*}}, \big\{b^{M^{(1)},(q),c}\big\}_{c =1}^{C^{*}},  \big\{b^{D,(q),c}\big\}_{c =1}^{C^{*}},  \text{ and }  \big\{b^{Z,(q),c}\big\}_{c =1}^{C^{*}},$ respectively.
             \item Fix $U_1$ at $u_1 = (d_1,d_{*,1}) \in \big\{(1,1), (0,0), (1,0),(0,1)\big\}$.
            \item Repeat the following $C^{*}$ times:
                \begin{enumerate}\setlength{\itemsep}{2pt}\setlength{\parskip}{0pt}
                    \item  At $t = 1$,  for a fixed $Z_1 = z_{*,1} $, conditional on $\boldsymbol{\ell}_{0}^{c}$, $d_{*,1}$, $b^{M^{(2)},(q),c}$ and $b^{M^{(1)},(q),c}$, sample one observation of the mediators $\bigl(M^{(1)}_1, M^{(2)}_1\bigr)$, say $\bigl(m^{(1),c}_1,m^{(2),c}_1\bigr)$, jointly from
                    \begin{align*}
                     & f_M\big(m^{(1)}_1, m^{(2)}_1 \big|z_{*,1},d_{*,1},\boldsymbol{\ell}_{0}^{c},b^{M^{(2)},(q),c},b^{M^{(1)},(q),c};\beta^{(q)}, \theta^{(q)}\big)
                     \\& = \sum_{r=1}^{N} w_r\bigl(z_{*,1},d_{*,1},\boldsymbol{\ell}_{0}^{c}\bigr)\times p\bigl(m^{(1)}_1, m^{(2)}_1\big| z_{*,1},d_{*,1},\boldsymbol{\ell}_{0}^{c}, b^{M^{(2)},(q),c},b^{M^{(1)},(q),c};\theta^{(q)}_{s|r}\bigr).
                    \end{align*}
                    \item  For $t \in \{2, \ldots, T\}$:
                    \begin{enumerate}\setlength{\itemsep}{2pt}\setlength{\parskip}{0pt}
                        \item For a fixed regime $\textbf{Z}_{t} = \textbf{z}_t $, conditional on $\boldsymbol{\ell}_{0}^{c}$, $d_1$, $\textbf{d}^{c}_{2:t-1}$,  $\textbf{m}^{(1),c}_{t-1}, \textbf{m}^{(2),c}_{t-1}$, and $b^{D,(q),c}$, sample one observation of the treatment receipt status $D_t$, say $d^{c}_t$, from:
                        \begin{align*}
                            & f_D\big(d_t\big|\textbf{z}_t,d_{1},     \textbf{d}^{c}_{2:t-1}, \textbf{m}^{(1),c}_{t-1}, \textbf{m}^{(2),c}_{t-1},\boldsymbol{\ell}_{0}^{c},b^{D,(q),c};\beta^{(q)}, \theta^{(q)}\big)\\
                            & = \sum_{r=1}^{N} w_r\bigl(\textbf{z}_t,d_{1},\textbf{d}^{c}_{2:t-1}, \textbf{m}^{(1),c}_{t-1}, \textbf{m}^{(2),c}_{t-1},\boldsymbol{\ell}_{0}^{c}\bigr)\times p\bigl(d_t\big| \textbf{z}_t, \boldsymbol{\ell}_{0}^{c},b^{D,(q),c};\theta^{(q)}_{s|r}\bigr).
                        \end{align*}
                        Similarly, for a fixed regime $\textbf{Z}_{t} = \textbf{z}_{*,t} $, conditional on $\boldsymbol{\ell}_{0}^{c}$, $d_{*,1}$, $\textbf{d}^{c}_{*,2:t-1},$ $\textbf{m}^{(1),c}_{t-1}, \textbf{m}^{(2),c}_{t-1}$, and $b^{D,(q),c}$,   sample one observation of the treatment receipt status $D_t$, say $d^{c}_{*,t}$, from:
                        \begin{align*}
                            & f_D\big(d_t\big|\textbf{z}_{*,t},d_{*,1}, \textbf{d}^{c}_{*,2:t-1}, \textbf{m}^{(1),c}_{t-1},  \textbf{m}^{(2),c}_{t-1},\boldsymbol{\ell}_{0}^{c},b^{D,(q),c};\beta^{(q)}, \theta^{(q)}\big)\\
                            & = \sum_{r=1}^{N} w_r\bigl(\textbf{z}_{*,t},d_{*,1},\textbf{d}^{c}_{*,2:t-1}, \textbf{m}^{(1),c}_{t-1}, \textbf{m}^{(2),c}_{t-1},\boldsymbol{\ell}_{0}^{c}\bigr)\times p\bigl(d_t\big| \textbf{z}_{*,t}, \boldsymbol{\ell}_{0}^{c}, b^{D,(q),c};\theta^{(q)}_{s|r}\bigr).
                        \end{align*}
                        \item For a fixed regime $\textbf{Z}_{t} = \textbf{z}_{*,t} $, conditional on $\boldsymbol{\ell}_{0}^{c}$, $d_{*,1}$, $\textbf{d}^{c}_{*,2:t}$, $\textbf{m}^{(1),c}_{t-1}, \textbf{m}^{(2),c}_{t-1}$,   sample one observation of two mediators $\bigl(M^{(1)}_t, M^{(2)}_t\bigr)$, say $\bigl(m^{(1),c}_t,  m^{(2),c}_t\bigr)$, jointly from:
                        \begin{align*}
                             & f_M\big(m^{(1)}_t, m^{(2)}_t\big|\textbf{z}_{*,t},d_{*,1},\textbf{d}^{c}_{*,2:t}, \textbf{m}^{(1),c}_{t-1},  \textbf{m}^{(2),c}_{t-1},\boldsymbol{\ell}_{0}^{c},b^{M^{(2)},(q),c},b^{M^{(1)},(q),c};\beta^{(q)}, \theta^{(q)}\big)\\
                            & = \sum_{r=1}^{N} w_r\bigl(\textbf{z}_{*,t},d_{*,1},\textbf{d}^{c}_{*,2:t}, \textbf{m}^{(1),c}_{t-1},  \textbf{m}^{(2),c}_{t-1},\boldsymbol{\ell}_{0}^{c}\bigr)\times p\bigl(m^{(1)}_t, m^{(2)}_t\big|\textbf{z}_{*,t},d_{*,1},\textbf{d}^{c}_{*,2:t}, \boldsymbol{\ell}_{0}^{c},b^{M^{(2)},(q),c},b^{M^{(1)},(q),c};\theta^{(q)}_{s|r}\bigr).
                        \end{align*}
                    \end{enumerate}
                    \item At $t = T$, for a fixed regime $\textbf{Z}_{T} = \textbf{z}_{T} $, conditional on $\boldsymbol{\ell}_{0}^{c}$, $d_{1}$, $\textbf{d}^{c}_{2:T}$, $\textbf{m}^{(1),c}_{T}, \textbf{m}^{(2),c}_{T}$, compute the expectation  $E\big[Y_{T}\big|\textbf{z}_{T},d_{1},\textbf{d}^{c}_{2:T}, \textbf{m}^{(1),c}_{T},  \textbf{m}^{(2),c}_{T},\boldsymbol{\ell}_{0}^{c};\beta^{(q)}, \theta^{(q)}\big]$ using:
                    \begin{align*}
                        &  E\big[Y_{T}\big|\textbf{z}_{T},d_{1}, \textbf{d}^{c}_{2:T}, \textbf{m}^{(1),c}_{T}, \textbf{m}^{(2),c}_{T},\boldsymbol{\ell}_{0}^{c};\beta^{(q)}, \theta^{(q)}\big]\\
                        & = \sum_{r=1}^{N} w_r\bigl(\textbf{z}_{T},d_{1},\textbf{d}^{c}_{2:T}, \textbf{m}^{(1),c}_{T},  \textbf{m}^{(2),c}_{T},\boldsymbol{\ell}_{0}^{c}\bigr)\times E\big[Y_{T}\big|\textbf{z}_{T},d_{1}, \textbf{d}^{c}_{2:T}, \textbf{m}^{(1),c}_{T},  \textbf{m}^{(2),c}_{T},\boldsymbol{\ell}_{0}^{c};\beta^{(q)}\big]\\
                        & = \sum_{r=1}^{N} w_r\bigl(\textbf{z}_{T},d_{1},\textbf{d}^{c}_{2:T}, \textbf{m}^{(1),c}_{T},  \textbf{m}^{(2),c}_{T},\boldsymbol{\ell}_{0}^{c}\bigr)\times \big(1 - \pi_{r}^{Y, (q)}\big) \times \underbrace{E\big[Y_{T}\big|Y_{T} > 0, \textbf{z}_{T},d_{1}, \textbf{d}^{c}_{2:T}, \textbf{m}^{(1),c}_{T},  \textbf{m}^{(2),c}_{T}, \boldsymbol{\ell}_{0}^{c};\beta^{(q)}\big]}_{\text{mean of normal distribution truncated at 0 due to local outcome regression in Section~\ref{sec:4.1} of the main text}}.
                    \end{align*}
                \end{enumerate}
            \item Compute  $\theta(\textbf{z},\textbf{z}_{*})^{(q)} \approx  \frac{1}{C^{*}} \sum_{c=1}^{C^{*}} \Bigl\{E\big[Y_{T}\big|\textbf{z}_{T},d_{1},\textbf{d}^{c}_{2:T}, \textbf{m}^{(1),c}_{T}, \textbf{m}^{(2),c}_{T},\boldsymbol{\ell}_{0}^{c};\beta^{(q)}, \theta^{(q)}\big] \Bigr\}.$
        \end{enumerate}
\end{minipage}%
}%
}
\end{breakablealgorithm}

\newpage

\section{Details on the simulation study}\label{sec:S6}

\subsection{Simulation design}\label{sec:71}
\subsubsection{Data-generating mechanism}
To ensure that our simulations reflect realistic data characteristics, we design the data-generating process using parameter estimates obtained from fitting parametric models to the empirical dataset described in Section~\ref{sec:2} in the main text. Specifically, we fit the following parametric model:

\begin{equation}\label{Sim::ParamModel}
\resizebox{\textwidth}{!}{$\displaystyle
    \begin{aligned}
          & Y_{i}\big| M_{i,T}^{(1)},   M_{i,T}^{(2)}, D_{i,T}, Z_{i,T},L_{i0}   
         \begin{cases}
             = 0 \text{ with probability } \pi_Y \\
         \sim N\big(\gamma_{M^{(1)}}m^{(1)}_{i,T}+ \gamma_{M^{(2)}}m^{(J)}_{i,T}+ \gamma_{D}d_{i,T} + \gamma_{Z}z_{i,T}+ \gamma_{L_{0}}\ell_{i0}, \sigma^{2}_{Y}\big) \text{ with probability } 1- \pi_Y
         \end{cases}, 
         \\
        & M_{i,t}^{(2)} \big|M_{i,t}^{(1)}, M_{i,t-1}^{(2)}, D_{i,t}, Z_{i,t},L_{i0} 
         \begin{cases}
            = 0 \text{ with probability } \pi_{M_{t}^{(2)}} \\
         \sim N\big(\eta_{t, M^{(1)}} m^{(1)}_{i,t} + \eta_{t-1,M^{(2)}} m^{(2)}_{i,t-1} + \eta_{t,D}d_{i,t} + \eta_{t,Z}z_{i,t}+ \eta_{t, L_{0}}\ell_{i0}, \sigma^{2}_{M_{t}^{(2)}}\big) \text{ with probability } 1- \pi_{M_{t}^{(2)}}
        \end{cases}, 
        \\
        & M_{i,t}^{(1)} \big|M_{i,t-1}^{(1)},  M_{i,t-1}^{(2)}, D_{i,t}, Z_{i,t},L_{i0} 
         \begin{cases}
            = 0 \text{ with probability } \pi_{M_{t}^{(1)}} \\
         \sim N\big(\xi_{t,M^{(1)}}m^{(1)}_{i,t-1} + \xi_{t, M^{(2)}}m^{(2)}_{i,t-1} + \xi_{t,D}d_{i,t} + \xi_{t,Z}z_{i,t}+ \xi_{t,L_{0}}\ell_{i0}, \sigma^{2}_{M_{t}^{(J)}}\big) \text{ with probability } 1- \pi_{M_{t}^{(1)}}
        \end{cases}, 
        \\
        & D_{i,t}\big| M_{i,t-1}^{(1)}, \ldots,  M_{i,t-1}^{(j)}, D_{i,t-1}, Z_{i,t}, L_{i0}
        \sim Bern\big(\Phi(\beta_{t,M^{(1)}}m^{(1)}_{i,t-1}+  \beta_{t,M^{(2)}}m^{(2)}_{i,t-1}+ \beta_{t,D}d_{i,t-1} + \beta_{t,Z}z_{i,t}+ \beta_{t,L_0}\ell_{i0} )\big),
        \\
        & Z_{i,t}\big| M_{i,t-1}^{(1)}, \ldots,  M_{i,t-1}^{(j)}, D_{i,t-1}, Z_{i,t-1}, L_{i0}
        \sim Bern\big(\Phi(\psi_{t,M^{(1)}}m^{(1)}_{i,t-1}+  \psi_{t,M^{(2)}}m^{(2)}_{i,t-1}+ \psi_{t,D}d_{i,t-1} + \psi_{t,Z}z_{i,t-1}+ \psi_{t,L_0}\ell_{i0} )\big),
    \end{aligned}
$}
\end{equation}
where $t = 1,2,3$. Here, $\boldsymbol{\gamma}$, $\boldsymbol{\eta}_t$, $\boldsymbol{\xi}_t$, $\boldsymbol{\beta}_t$, and $\boldsymbol{\psi}_t$ denote the regression coefficient vectors for the outcome, second mediator, first mediator, treatment receipt status, and treatment assignment status models, respectively. The parameters $\pi_Y$, $\pi_{M_t^{(1)}}$, and $\pi_{M_t^{(2)}}$ represent the zero-inflation probabilities, and $\sigma^2_Y$, $\sigma^2_{M_t^{(1)}}$, and $\sigma^2_{M_t^{(2)}}$ denote the corresponding variance components. We assign weakly informative priors: $N(0, 10)$ for regression coefficients, $\text{Beta}(1, 1)$ for zero-inflation probabilities, and $\text{Inverse-Gamma}(3, 1)$ for variance parameters.

 The resulting data-generating process in  (\ref{Sim::ParamModel}) corresponds to a first-order autoregressive (AR(1)) longitudinal structural model, in which the conditional distribution of each time-varying variable at time $t$ depends on its own immediately preceding value and on contemporaneous longitudinal variables along with baseline covariates. The AR(1) specification is adopted to ensure numerical stability and to avoid multicollinearity. With $T=3$, including the multiple lagged time-varying covariates in each conditional model can induce severe multicollinearity due to the strong temporal dependence among variables. This, in turn, can lead to unstable parameter estimates, inflated posterior uncertainty, and poor mixing in Bayesian computation. The AR(1) gives a relatively complex dependence that provides a ground to evaluate the EDPM model and compare it against standard alternatives.

The parameter estimates from~(\ref{Sim::ParamModel}) serve two purposes. First, we use them to generate data replications for simulation. Second, using these estimates and the G-computation algorithm (Algorithm~\ref{Alg1:Gcomp}), we compute the "ground truth" principal interventional direct, indirect, and total effect estimates.


\subsubsection{Simulation scenarios}
We consider two simulation scenarios. Under each of these two simulation scenarios, we fit (i) a parametric model, which assumes a single-component specification as in~(\ref{Sim::ParamModel}) and therefore ignores the latent mixture structure, and (ii) the proposed EDPM model to each simulated dataset. The parametric model is correctly specified in the first scenario and misspecified in the second scenario. The EDPM model is a flexible nonparametric model. For each scenario, we generate $R = 500$ replicated datasets, each with sample sizes $n = 12{,}000$ and $15{,}000$ subjects observed at $T = 3$ time points, matching the temporal structure of the empirical study.

\begin{enumerate}[label=\textbf{Scenario \arabic*:}, leftmargin=*]
    \item \textbf{Correct Model Specification.} Data are generated from the parametric model in~(\ref{Sim::ParamModel}).

    \item \textbf{Parametric Model Misspecification.} Data are generated from a finite mixture model with $A=10$ mixture components. Specifically, each subject $i$ is first assigned to one of ten latent clusters $c_i \in \{1,\ldots,10\}$ with probabilities $(\pi_1, \ldots, \pi_{10})$, and conditional on $c_i = a$, the data for subject $i$ are generated according to the parametric model in~(\ref{Sim::ParamModel}) with cluster-specific parameters $\big(\boldsymbol{\gamma}^{(a)}, \boldsymbol{\eta}_t^{(a)}, \boldsymbol{\xi}_t^{(a)}, \boldsymbol{\beta}_t^{(a)}, \boldsymbol{\psi}_t^{(a)}, \pi_Y^{(a)}, \pi_{M_t^{(1)}}^{(a)}, \pi_{M_t^{(2)}}^{(a)}, \sigma^{2,(a)}_Y, \sigma^{2,(a)}_{M_t^{(1)}}, \sigma^{2,(a)}_{M_t^{(2)}}\big)$. The cluster-specific parameters are obtained by fitting the ten-component mixture model to the empirical data described in Section~\ref{sec:2} in the main text. The ground truth causal effects under this data-generating mechanism are computed by marginalizing over the mixture components. 
    This scenario evaluates the robustness of the EDPM model when the true data-generating mechanism exhibits substantial population-level heterogeneity that a single-component parametric model cannot capture. A secondary analysis under a three-component mixture data-generating mechanism is reported in Section~S7.
\end{enumerate}

 We note that in both simulation scenarios, neither the parametric model nor the EDPM model includes subject-specific random effects. Accordingly, causal effect estimation proceeds via the G-computation formula in Section~\ref{sec:3} rather than its random-effects counterpart in Section~\ref{sec:4} of the main text. The objective of the simulation study is to compare causal effect estimation under the parametric and EDPM frameworks. Including random effects would introduce additional computational burden, particularly in posterior sampling and numerical integration over subject-specific latent variables, without directly informing this comparison. The role of random effects within the proposed framework is explored in the case-study in Section~\ref{sec:5} of the main text, where the richer data structure and relevance of subject-specific variation justify their inclusion.

\subsubsection{Estimation procedures}

For the parametric models, we run a single Markov chain Monte Carlo (MCMC) chain with a total of $11{,}000$ iterations. We discard the first $10{,}000$ iterations as burn-in and retain the remaining $1{,}000$ posterior samples for inference. In contrast, for the EDPM models, we run $15{,}000$ MCMC iterations, discarding the first $10{,}000$ iterations as burn-in. From the remaining $5{,}000$ iterations, we retain every fifth sample, yielding $1{,}000$ posterior draws for inference.

  At each retained MCMC iteration, we implement the G-computation algorithm (Algorithm~\ref{Alg1:Gcomp}) with $M = 10{,}000$ Monte Carlo samples to compute the principal interventional effects across the four principal strata. This procedure yields $1{,}000$ posterior samples of each causal estimate per simulated dataset, from which we compute posterior means and 95\% Bayesian credible intervals. All model fitting and posterior computation are conducted using \texttt{Nimble} \citep{de2017programming, de2020nimble}.

  In the simulation studies, we run a single MCMC chain for each replicated dataset rather than four parallel chains, as used in Section~\ref{sec:5} of the main text. This choice is motivated by computational considerations: with $R = 500$ replicated datasets, each requiring posterior inference via MCMC followed by G-computation with $10{,}000$ Monte Carlo samples, running multiple chains per replication would substantially increase the already intensive computational burden.

\subsection{Additional simulation results}

\begin{figure}[!h]
\centering
\begin{subfigure}[t]{0.49\textwidth}
    \centering
    \includegraphics[width=\linewidth]{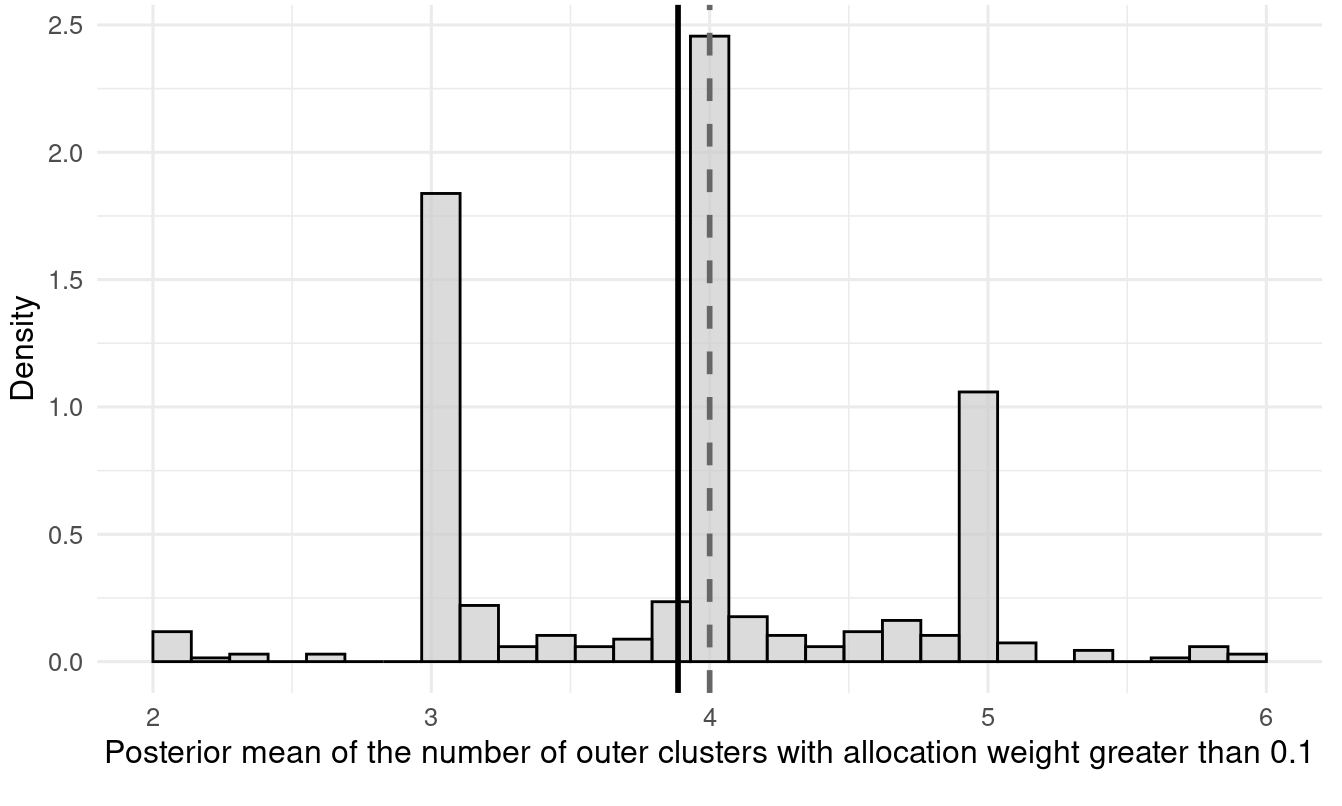}
    \caption{True parametric DGM: EDPM fitted model (n = 12000)}
    \label{Fig:Sim_TRUE_EDPM_num_outer_clusters_n12000}
\end{subfigure}
\hfill
\begin{subfigure}[t]{0.49\textwidth}
    \centering
    \includegraphics[width=\linewidth]{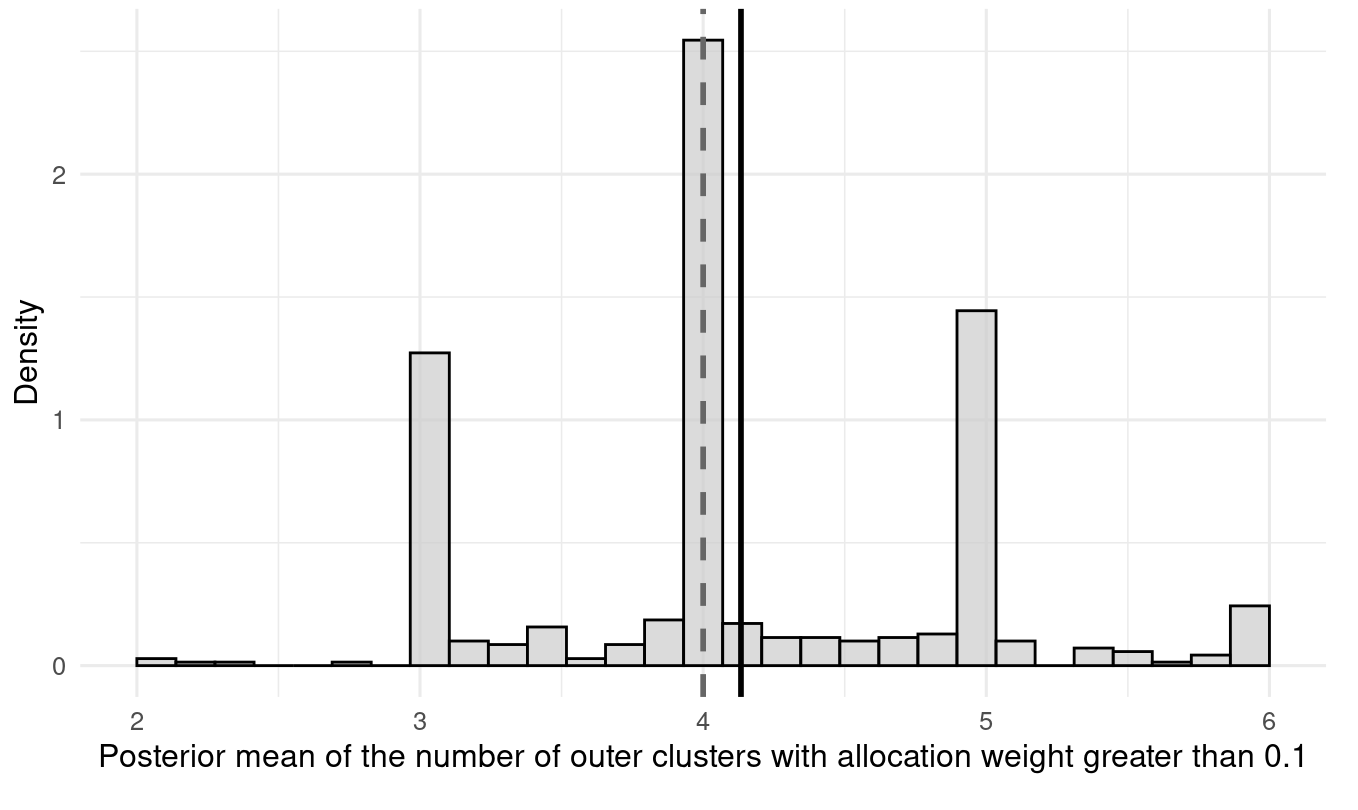}
    \caption{True parametric DGM: EDPM fitted model (n = 15000)}
    \label{Fig:Sim_TRUE_EDPM_num_outer_clusters_n15000}
\end{subfigure}
\caption{Histogram of the posterior mean of the number of significant outer clusters across $R=500$ replicated datasets. Solid black vertical lines denote the empirical mean (3.89 in the left panel and 4.13 in the right panel), while dashed gray vertical lines indicate the corresponding medians (4.00 in both left and right panels).}

\label{Fig:Sim_TRUE_EDPM_num_outer_clusters}
\end{figure}

\begin{figure}[!h]
\centering
\begin{subfigure}[t]{0.49\textwidth}
    \centering
    \includegraphics[width=\linewidth]{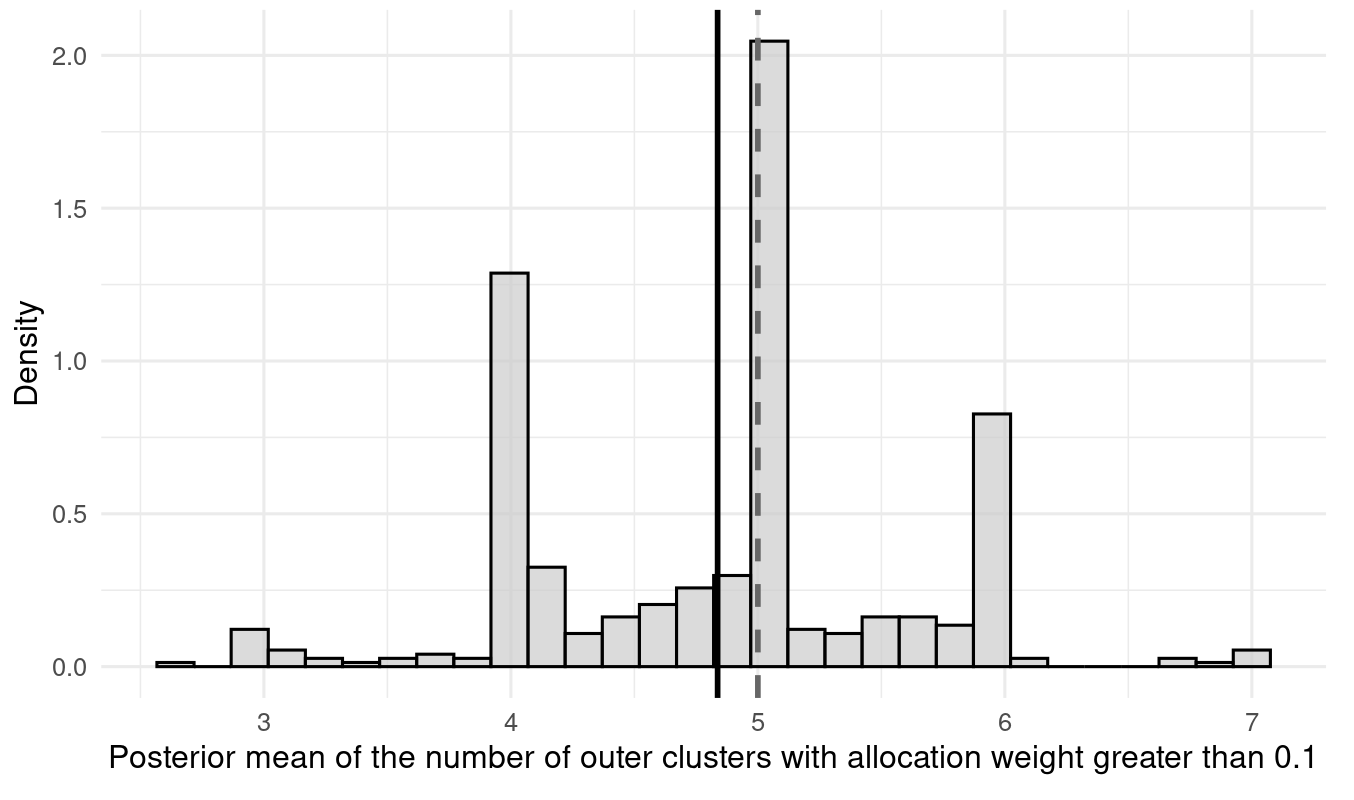}
    \caption{Misspecified parametric DGM: EDPM fitted model (n = 12000)}
    \label{Fig:Sim_MISSPECIFIED_EDPM_num_outer_clusters_n12000}
\end{subfigure}
\hfill
\begin{subfigure}[t]{0.49\textwidth}
    \centering
    \includegraphics[width=\linewidth]{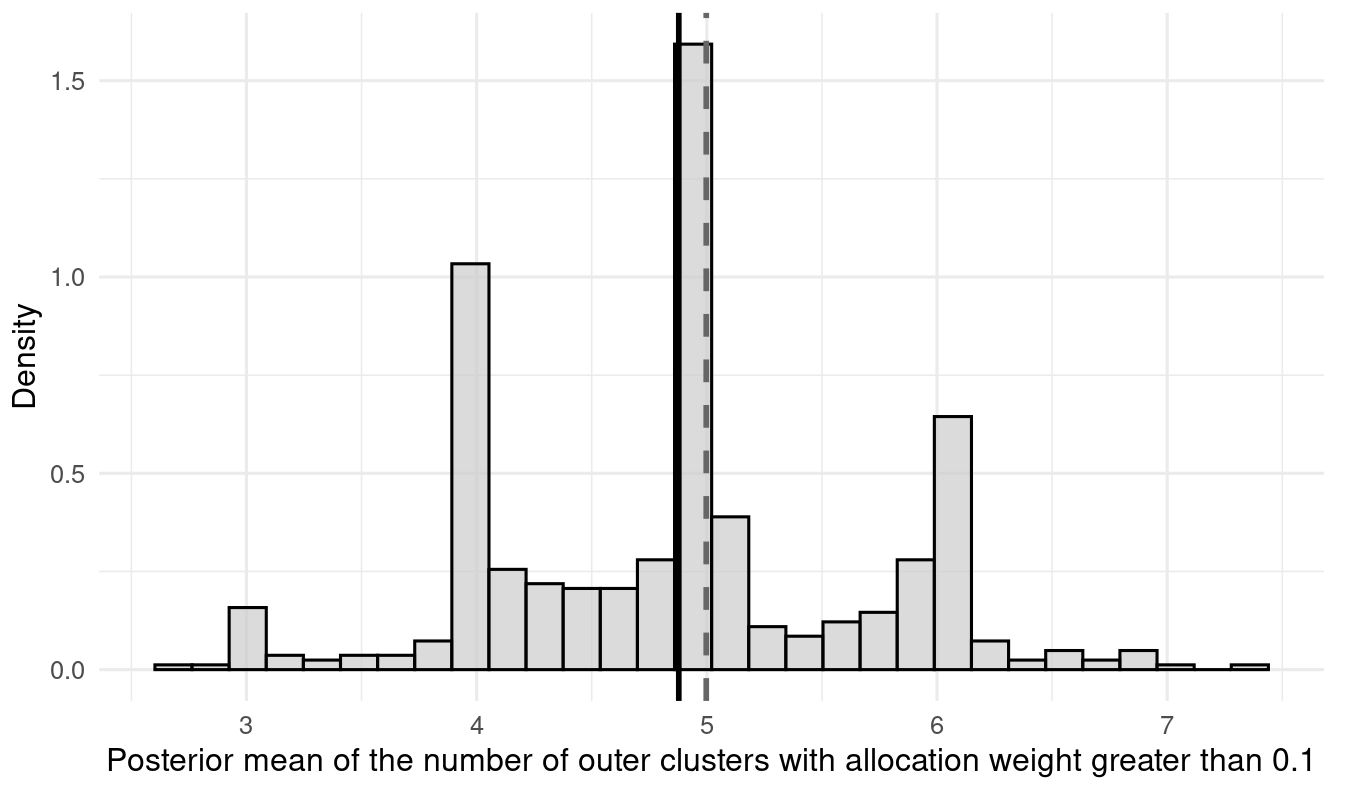}
    \caption{Misspecified parametric DGM: EDPM fitted model (n = 15000)}
    \label{Fig:Sim_MISSPECIFIED_EDPM_num_outer_clusters_n15000}
\end{subfigure}
\caption{Histogram of the posterior mean of the number of significant outer clusters across $R=500$ replicated datasets. Solid black vertical lines denote the empirical mean (4.84 in the left panel and 4.88 in the right panel), while dashed gray vertical lines indicate the corresponding medians (5.00 in both left and right panels).}

\label{Fig:Sim_MISSPECIFIED_EDPM_num_outer_clusters}
\end{figure}

\begin{table}[h]
\caption{Simulation results at sample size $n = 12{,}000$ for comparing the bias, MSE, and coverage of 95\% credible intervals of the \emph{true} parametric model versus EDPM model.}
\label{Table:SimTrueParamModel_n12000}
\resizebox{\textwidth}{!}{
\begin{tabular}{@{}lccccccccccc@{}}
\hline
Principal &  &  &  & \multicolumn{4}{c}{Bayesian parametric model} & \multicolumn{4}{c}{EDPM model} \\
\cline{5-8} \cline{9-12}
strata & n & Effects & Ground truth & Bias & $95\%$ & Average CI & MSE & Bias & $95\%$ & Average CI & MSE \\
at $t=1$ &  &  &  &  & coverage & width &  &  & coverage & width &  \\
\hline
Active & 12000 & Direct & 1.55 & 0.0008 & 0.99 & 0.57 & 0.0112 & 0.0288 &      0.94 &  0.82  &  0.0427  \\
customers &  & Indirect & -0.08 & 0.0443 & 1.00 & 1.04 & 0.0022 & 0.0213  &     0.93 &  0.14 &  0.0014 \\
 &  & Total & 1.47 & 0.0450 & 1.00 & 1.05 & 0.0135 &  0.0501 &      0.95  &  0.82 &   0.0440  \\
\cline{1-12}
Non-price & 12000 & Direct & 1.57 & -0.0003 & 0.99 & 0.58 & 0.0113 &  0.0077 &      0.95 &  0.96  &  0.0613  \\
value-attentive &  & Indirect & -0.08 & 0.0410 & 1.00 & 1.16 & 0.0021 &  0.0225  &     0.94   &  0.19  &   0.0026  \\
customers &  & Total & 1.49 & 0.0407 & 1.00 & 1.15 & 0.0133 & 0.0302  &     0.95  &   0.97 &   0.0610  \\
\cline{1-12}
Price-attentive & 12000 & Direct & 1.53 & 0.0008 & 0.99 & 0.57 & 0.0114 & 0.0483 &      0.95 &   0.98 &  0.0645  \\
customers &  & Indirect & -0.08 & 0.0476 & 1.00 & 1.05 & 0.0024 & 0.0232 &      0.93  &  0.17  &   0.0021  \\
 &  & Total & 1.45 & 0.0484 & 1.00 & 1.06 & 0.0140 &  0.0715 &      0.94  &   0.97  &  0.0674  \\
\cline{1-12}
Non-active & 12000 & Direct & 1.55 & -0.0004 & 0.99 & 0.57 & 0.0112 &  0.0262    &   0.95 &  0.82 &    0.0423  \\
customers &  & Indirect & -0.08 & 0.0448 & 1.00 & 1.03 & 0.0023 &  0.0258 &      0.91  &  0.14  &   0.0016  \\
 &  & Total & 1.47 & 0.0445 & 1.00 & 1.04 & 0.0135 & 0.0519 &   0.95   &    0.82  &  0.0437 \\
\hline
\end{tabular}}
\end{table}

\begin{table}[h]
\caption{Simulation results at sample size $n = 12{,}000$ for comparing the bias, MSE, and coverage of 95\% credible intervals of the \emph{misspecified} parametric model versus EDPM model under an $A=10$ component finite mixture data-generating mechanism.}
\label{Table:SimMisspecifiedParamModel_10comp_n12000}
\resizebox{\textwidth}{!}{
\begin{tabular}{@{}lccccccccccc@{}}
\hline
Principal &  &  &  & \multicolumn{4}{c}{Bayesian parametric model} & \multicolumn{4}{c}{EDPM model} \\
\cline{5-8} \cline{9-12}
strata & n & Effects & Ground truth & Bias & $95\%$ & Average CI & MSE & Bias & $95\%$ & Average CI & MSE \\
at $t=1$ &  &  &  &  & coverage & width &  &  & coverage & width &  \\
\hline
Active & 12000 & Direct & 1.73 & $-0.3418$ & 0.07 & 0.39 & 0.1273 & $0.2823$ & 0.69 & 0.82 & 0.1378 \\
customers &  & Indirect & 0.03 & $-0.1833$ & 0.00 & 0.10 & 0.0341 & $-0.1572$ & 0.64 & 0.48 & 0.0643 \\
 &  & Total & 1.76 & $-0.5251$ & 0.00 & 0.40 & 0.2864 & $0.1251$ & 0.84 & 0.88 & 0.1046 \\
\cline{1-12}
Non-price & 12000 & Direct & 1.74 & $-0.3419$ & 0.07 & 0.40 & 0.1275 & $0.2322$ & 0.81 & 0.91 & 0.1236 \\
value-attentive &  & Indirect & 0.04 & $-0.1963$ & 0.00 & 0.12 & 0.0394 & $-0.0913$ & 0.74 & 0.63 & 0.0815 \\
customers &  & Total & 1.78 & $-0.5381$ & 0.00 & 0.40 & 0.3005 & $0.1410$ & 0.83 & 1.00 & 0.1519 \\
\cline{1-12}
Price-attentive & 12000 & Direct & 1.71 & $-0.3403$ & 0.07 & 0.39 & 0.1264 & $0.3287$ & 0.71 & 0.93 & 0.1811 \\
customers &  & Indirect & 0.02 & $-0.1726$ & 0.00 & 0.09 & 0.0302 & $-0.1779$ & 0.62 & 0.49 & 0.0749 \\
 &  & Total & 1.74 & $-0.5129$ & 0.00 & 0.40 & 0.2739 & $0.1509$ & 0.85 & 1.00 & 0.1232 \\
\cline{1-12}
Non-active & 12000 & Direct & 1.73 & $-0.3423$ & 0.06 & 0.39 & 0.1276 & $0.2755$ & 0.68 & 0.81 & 0.1354 \\
customers &  & Indirect & 0.03 & $-0.1841$ & 0.00 & 0.10 & 0.0344 & $-0.1091$ & 0.74 & 0.57 & 0.0666 \\
 &  & Total & 1.76 & $-0.5264$ & 0.00 & 0.40 & 0.2878 & $0.1664$ & 0.81 & 0.90 & 0.1306 \\
\hline
\end{tabular}}
\end{table}

\clearpage
\section{Additional figures and tables}\label{sec:S7}

This section contains supporting figures and tables referenced in the main text. Table~S4 presents an additional model misspecification study comparing the proposed EDPM with a misspecified single-component Bayesian parametric model under a three-component finite-mixture data-generating mechanism. Figure~S3 shows the distributions of the longitudinal mediators and outcome, highlighting the substantial zero-inflation and right-skewness that motivate our hurdle-based EDPM specification. Figure~S4 reports the number of customers who opened promotional emails at each time point, illustrating the treatment non-compliance observed in the study. Figure~S5 displays posterior means and $95\%$ credible intervals for the principal interventional indirect effects across principal strata as the mediator values in the $\mathbf{Z}$ arm increase while those in the $\mathbf{Z}_{*}$ arm remain fixed at zero. Table~S5 reports causal effect estimates under a counterfactual intervention that sets treatment receipt to $D_{i2}=D_{i3}=1$ and both mediators to zero at all time points, thus isolating the marginal effect of treatment assignment. Finally, Table~S6 provides the numerical values of $\theta(\mathbf{z}_{T},\mathbf{z}_{*,T})$, $\theta(\mathbf{z}_{T},\mathbf{z}_{T})$, and $\theta(\mathbf{z}_{*,T},\mathbf{z}_{*,T})$ corresponding to Figure~\ref{fig:recipients_vs_openers_sampleSize} in the main text.

\begin{figure}[h]
    \centering
    \includegraphics[width=\linewidth]{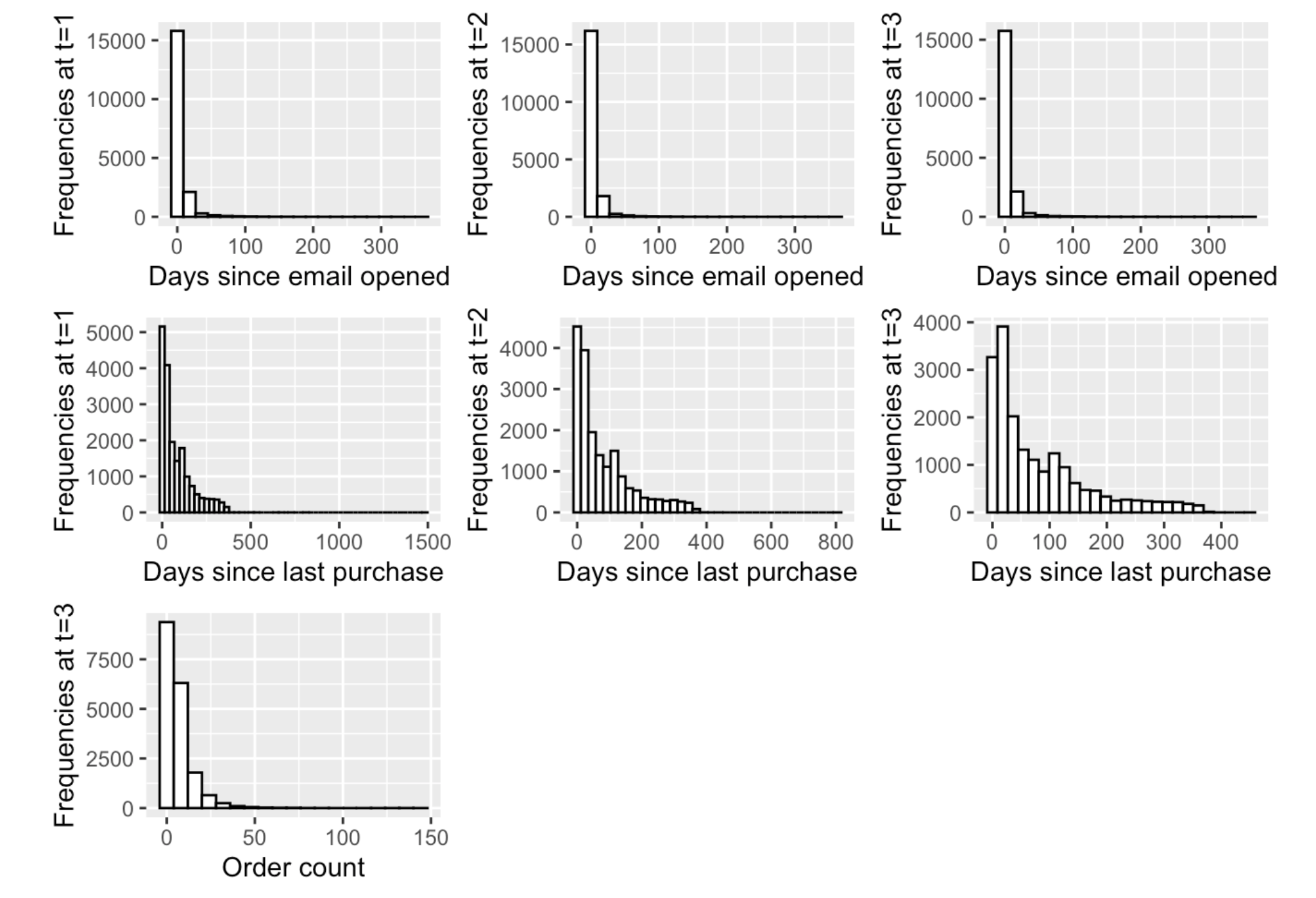}
    \caption{Histogram of observed longitudinal mediators at $t = 1,2,3$ and outcome at $t=3$.}
    \label{Fig1}
\end{figure}

\begin{figure}[h]
\centering
\includegraphics[width=0.8\linewidth]{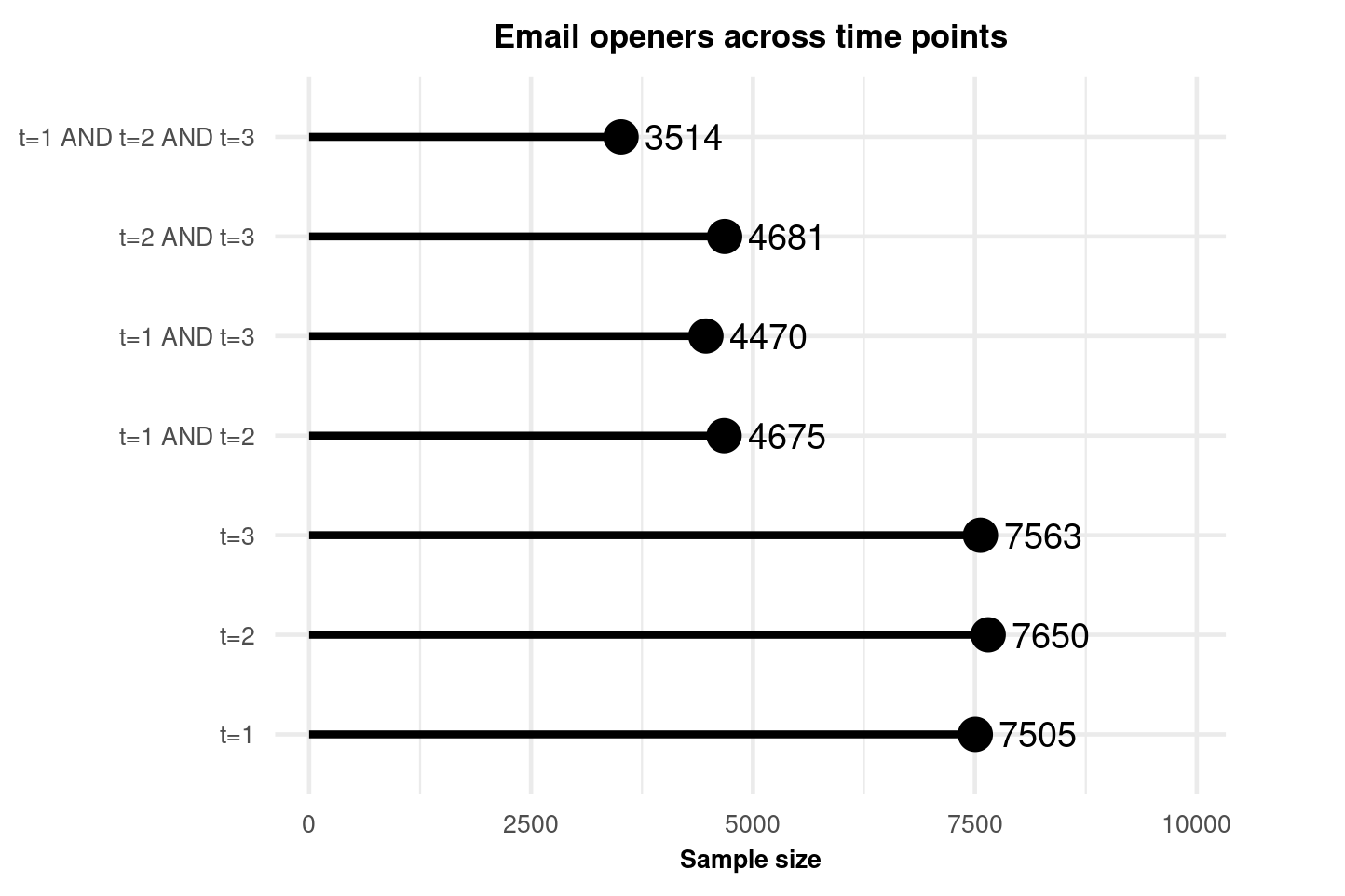}
\caption{Sample sizes of email openers at times $t = 1, 2, 3$ in the observed data.}
\label{fig:OpenerSubPopnAcrossTime}
\end{figure}

\begin{figure}[h]
    \centering
    \includegraphics[width=\linewidth]{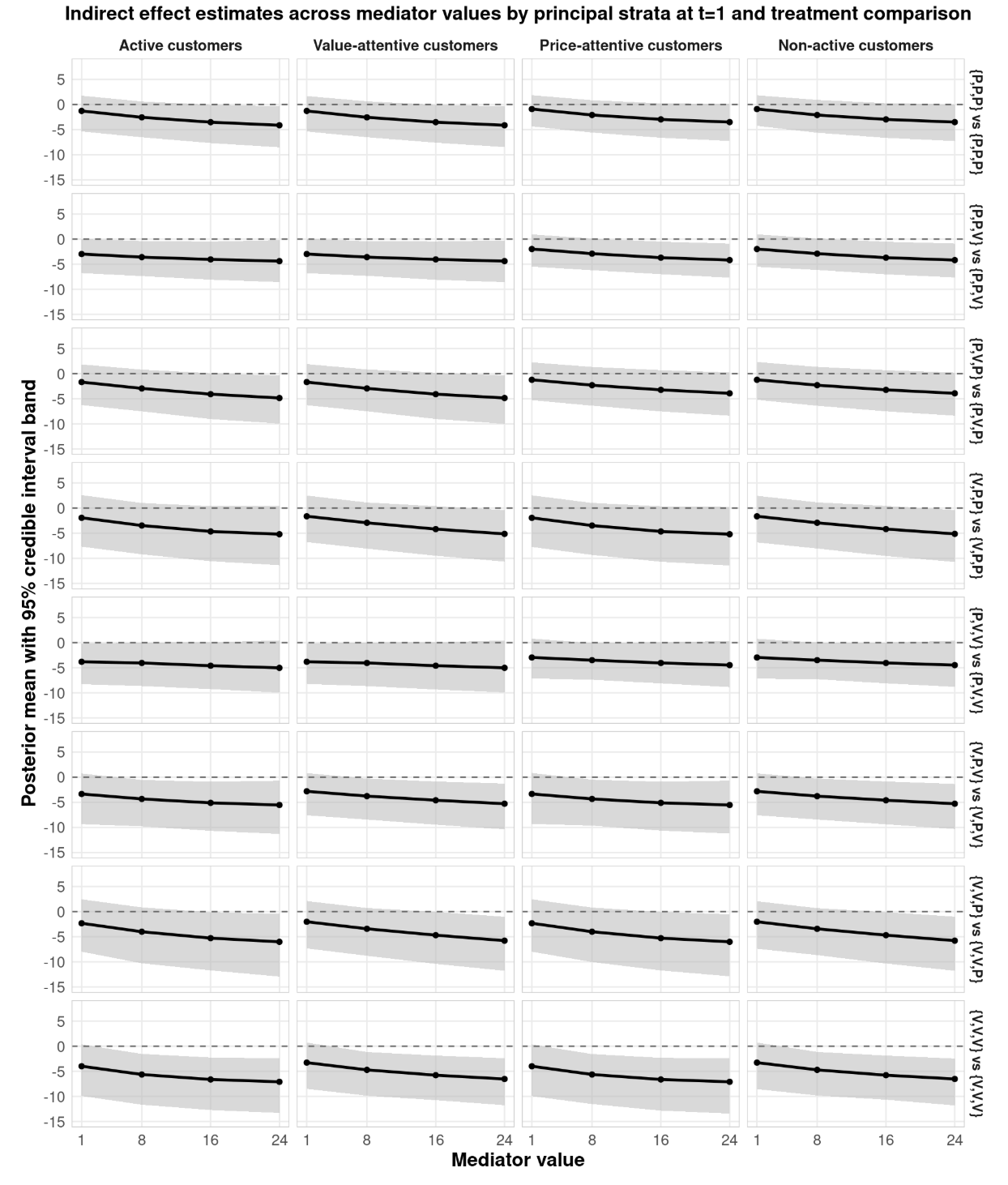}
    \caption{\small Comparison of indirect effects across principal strata at $t=1$ for different mediator values in the $\mathbf{Z}$ arm, with both mediators in the $\mathbf{Z}_{*}$ arm fixed at zero. Each plot displays the posterior mean (solid black curve) along with the corresponding $95\%$ Bayesian credible interval (grey band) for the effect estimates.}
    \label{fig:IndirectEffectEstimatesAcrossM181624.png}
\end{figure}

\begin{table}[h]
\caption{Simulation results for comparing the bias, MSE, and coverage of 95\% credible intervals of the \emph{misspecified} parametric model versus EDPM model under a three-component finite mixture data-generating mechanism (secondary model misspecification analysis).}
\label{Table:SimMisspecifiedParamModel}
\resizebox{\textwidth}{!}{
\begin{tabular}{@{}lccccccccccc@{}}
\hline
Principal &  &  &  & \multicolumn{4}{c}{Bayesian parametric model} & \multicolumn{4}{c}{EDPM model} \\
\cline{5-8} \cline{9-12}
strata & n & Effects & Ground truth & Bias & $95\%$ & Average CI & MSE & Bias & $95\%$ & Average CI & MSE \\
at $t=1$ &  &  &  &  & coverage & width &  &  & coverage & width &  \\
\hline
 & 12000 & Direct & 1.61 & -0.1153   &         0.83   &            0.45  &      0.0262 & 0.0376  &   0.91 &   0.86 &  0.0606 \\
 &  & Indirect & -0.04 & -0.0430    &        1.00 &              0.33  &      0.0020 & -0.1221  &     0.63 &                0.33   &   0.0227 \\
Active &  & Total & 1.58 & -0.1582   &         0.83  &             0.54    &    0.0376 & -0.0845 &      0.92 &                0.90   &   0.0708 \\
\cline{2-12}
customers & 15000 & Direct & 1.61 & -0.1204   &         0.81      &         0.41   &     0.0245 &  0.0579    &      0.91   &             0.78   &         0.0506  \\
 &  & Indirect & -0.04 & -0.0423   &         1.00 &              0.33 &       0.0019 &   -0.1202   &       0.57  &              0.30    &        0.0215  \\
 &  & Total & 1.58 & -0.1627   &         0.81  &             0.50  &      0.0363 &  -0.0623    &      0.91  &              0.82   &         0.0527  \\
\cline{1-12}
 & 12000 & Direct & 1.63 & -0.1132   &         0.84   &            0.46   &     0.0260 & -0.0220   &    0.93 &                0.97   &  0.0739 \\
 &  & Indirect & -0.03 & -0.0583  &          1.00   &            0.36  &      0.0037 & -0.1094   &    0.70 &                0.34   &  0.0215 \\
Non-price &  & Total & 1.60 & -0.1714  &          0.82 &              0.56   &     0.0421 & -0.1314 &      0.90  &    1.01 &  0.0950 \\
\cline{2-12}
value-attentive & 15000 & Direct & 1.63 & -0.1175 &           0.82   &            0.41  &      0.0240 &  -0.0011  &        0.94   &             0.88    &        0.0545  \\
customers &  & Indirect & -0.03 & -0.0571 &           1.00     &          0.36   &     0.0035 &  -0.1089  &        0.68   &             0.31      &      0.0202  \\
 &  & Total & 1.60 & -0.1747   &         0.81   &            0.52  &      0.0405 &  -0.1100  &        0.92  &              0.91   &         0.0700  \\
\cline{1-12}
 & 12000 & Direct & 1.60 & -0.1158   &         0.83  &             0.45   &     0.0265 & 0.0940  &     0.91  &               1.01 &   0.0886 \\
 &  & Indirect & -0.04 & -0.0297   &         1.00 &              0.34   &     0.0010 & -0.1271   &    0.61 &                0.33   &   0.0244 \\
Price-attentive &  & Total & 1.55 & -0.1455  &          0.86     &          0.55   &     0.0340 & -0.0331  &     0.92   &  1.05 & 0.0844 \\
\cline{2-12}
customers & 15000 & Direct & 1.60 & -0.1214  &          0.80     &          0.41     &   0.0248 &  0.1143  &     0.90    &            0.92   &         0.0806  \\
 &  & Indirect & -0.04 & -0.0294   &         1.00  &             0.34  &      0.0010 &  -0.1244  &     0.59   &             0.30   &         0.0226  \\
 &  & Total & 1.55 & -0.1508  &          0.85  &             0.51   &     0.0326 &  -0.0081   &    0.92  &              0.94   &         0.0673  \\
\cline{1-12}
 & 12000 & Direct & 1.62 & -0.1157   &         0.83  &             0.45   &     0.0263 & 0.0322  &     0.92 &      0.85   &    0.0586 \\
 &  & Indirect & -0.04 & -0.0436  &          1.00   &            0.32  &      0.0021 & -0.1129   &    0.65 &                0.32    &        0.0210 \\
Non-active &  & Total & 1.58 & -0.1593   &         0.84            &   0.54  &      0.0380 & -0.0807  &     0.91  &    0.89  &  0.0673 \\
\cline{2-12}
customers & 15000 & Direct & 1.62 & -0.1206  &          0.81    &           0.41  &      0.0245 &  0.0525  &     0.91    &               0.77      &         0.0488  \\
 &  & Indirect & -0.04 & -0.0429   &         1.00   &            0.32   &     0.0020 &  -0.1091  &     0.64     &              0.29     &          0.0190  \\
 &  & Total & 1.58 & -0.1635  &          0.81   &            0.50   &     0.0366 &  -0.0566    &   0.92  &                 0.81      &         0.0507 \\
\hline
\end{tabular}}
\end{table}

\begin{sidewaystable}[h]
\caption{Causal effect estimates (posterior means and $95\%$ Bayesian credible intervals) setting $D_{i2} = D_{i3} = 1$ and fixing \\ both mediators $M_{i,t}^{(1)} = M_{i,t}^{(2)} = 0$  for all $t$, with various treatment assignment regimes $\textbf{Z}_{T}$ against the reference arm $\textbf{Z}_{*,T} = \{0,0,0\}$}
\label{Table:AppendixDt2=Dt3=1fixedMt=0}
\resizebox{\textheight}{!}{%
\begin{tabular}{@{}lcrcrrrrr@{}}
\hline
 Principal strata & Effects & $\textbf{Z}_{T} = \{1,0,0\}$ & $\textbf{Z}_{T} = \{0,1,0\}$ & $\textbf{Z}_{T} = \{0,0,1\}$ & $\textbf{Z}_{T} = \{1,1,0\}$ & $\textbf{Z}_{T} = \{1,0,1\}$ &  $\textbf{Z}_{T} = \{0,1,1\}$ & $\textbf{Z}_{T} = \{1,1,1\}$  \\
at $t=1$ &  &
 & & & & & \\
\hline
{Active}  & {Direct} & 6.28($-$0.29, 13.12) & 1.79($-$0.77, 4.52) & 4.02($-$1.93, 10.03) & 6.14($-$0.76, 12.47) & 5.06($-$3.19, 13.11) & 4.12($-$2.44, 10.47) & 7.50($-$1.03, 16.12) \\
{customers}        & {Indirect}  & 0.00($-$0.16, 0.17) & 0.00($-$0.11, 0.13) & 0.00($-$0.16, 0.16) & 0.01($-$0.17, 0.18) & 0.00($-$0.15, 0.15) & 0.00($-$0.16, 0.17) & 0.00($-$0.16, 0.17) \\
          & {Total}   & 6.28($-$0.28, 13.11) & 1.79($-$0.69, 4.53) & 4.03($-$1.97, 10.01) & 6.14($-$0.58, 12.45) & 5.06($-$3.25, 13.20) & 4.12($-$2.43, 10.48) & 7.50($-$1.00, 16.20) \\[6pt]
{Non-price}   & {Direct} & 5.32($-$1.52, 11.93) & 1.79($-$0.72, 4.55) & 2.95($-$2.85, 8.27) & 5.02($-$2.00, 10.85) & 4.10($-$3.96, 11.31) & 3.16($-$3.14, 9.35) & 6.53($-$1.70, 14.94) \\
{value-attentive}  & {Indirect}  & 0.00($-$0.16, 0.16) & 0.00($-$0.12, 0.15) & 0.00($-$0.14, 0.17) & 0.00($-$0.15, 0.19) & 0.00($-$0.14, 0.14) & 0.00($-$0.16, 0.16) & 0.00($-$0.16, 0.16) \\
{customers}    & {Total}   & 5.32($-$1.47, 11.98) & 1.79($-$0.77, 4.56) & 2.95($-$2.95, 8.36) & 5.02($-$2.01, 10.76) & 4.10($-$3.93, 11.25) & 3.16($-$3.16, 9.32) & 6.53($-$1.65, 14.96) \\[6pt]
{Price-attentive} & {Direct } & 6.99($-$0.07, 13.68) & 1.73($-$0.62, 4.38) & 4.78($-$1.23, 11.24) & 6.89(0.15, 13.64) & 5.82($-$1.54, 14.00) & 4.87($-$1.70, 11.43) & 8.25(0.37, 16.83) \\
 {customers}      & {Indirect}  & 0.00($-$0.17, 0.17) & 0.00($-$0.11, 0.11) & 0.00($-$0.17, 0.17) & 0.00($-$0.18, 0.19) & $-$0.01($-$0.16, 0.14) & 0.00($-$0.17, 0.17) & 0.00($-$0.17, 0.17) \\
          & {Total}   & 6.98($-$0.12, 13.62) & 1.73($-$0.59, 4.34) & 4.78($-$1.21, 11.21) & 6.89(0.14, 13.62) & 5.81($-$1.58, 13.88) & 4.87($-$1.68, 11.51) & 8.25(0.34, 16.75) \\[6pt]
{Non-active}  & {Direct} & 6.02($-$0.36, 12.56) & 1.73($-$0.56, 4.29) & 3.70($-$1.68, 8.99) & 5.77($-$0.56, 11.56) & 4.85($-$2.08, 11.93) & 3.90($-$2.48, 9.93) & 7.29($-$0.37, 15.19) \\
{customers}    & {Indirect}  & 0.00($-$0.15, 0.15) & 0.00($-$0.11, 0.10) & 0.00($-$0.13, 0.16) & 0.00($-$0.15, 0.15) & 0.00($-$0.14, 0.14) & 0.00($-$0.15, 0.15) & 0.00($-$0.15, 0.15) \\
          & {Total}   & 6.02($-$0.36, 12.52) & 1.73($-$0.58, 4.28) & 3.70($-$1.69, 9.00) & 5.77($-$0.61, 11.64) & 4.86($-$2.07, 12.04) & 3.90($-$2.50, 10.03) & 7.29($-$0.35, 15.26) \\[6pt]
\hline
\end{tabular}%
}%
\end{sidewaystable}

\begin{sidewaystable}[h]
\caption{Posterior means and $95\%$ Bayesian credible intervals for parameters $\theta(\textbf{z}_{T}, \textbf{z}_{*,T})$, $\theta(\textbf{z}_{T}, \textbf{z}_{T})$, and $\theta(\textbf{z}_{*,T}, \textbf{z}_{*,T})$,  with various treatment \\ assignment regimes $\textbf{z}_{T}$ against the reference arm $\textbf{z}_{*,T} = \{0,0,0\}$}
\label{Table:AppendixThetaZZstar}
\resizebox{\textheight}{!}{%
\begin{tabular}{@{}lcrcrrrrr@{}}
\hline
 Principal strata & Parameter & $\textbf{z}_{T} = \{1,0,0\}$ & $\textbf{z}_{T} = \{0,1,0\}$ & $\textbf{z}_{T} = \{0,0,1\}$ & $\textbf{z}_{T} = \{1,1,0\}$ & $\textbf{z}_{T} = \{1,0,1\}$ &  $\textbf{z}_{T} = \{0,1,1\}$ & $\textbf{z}_{T} = \{1,1,1\}$  \\
at $t=1$ &  &  
 & & & & & \\
\hline
{Active}  & {$\theta(\textbf{z}_{T}, \textbf{z}_{*,T})$} & 7.36(6.22, 9.45) & 6.65(5.75, 8.37) & 6.07(5.44, 7.04) & 8.72(7.48, 10.96) & 7.88(7.13, 9.28) & 7.36(6.62, 8.49) & 9.40(8.56, 10.78) \\
{customers} & {$\theta(\textbf{z}_{T}, \textbf{z}_{T})$} & 7.32(6.16, 9.80) & 6.26(5.35, 7.87) & 5.95(5.28, 7.01) & 8.17(6.94, 10.58) & 7.31(6.92, 7.88) & 6.73(6.05, 7.92) & 8.25(7.69, 8.95) \\
        & {$\theta(\textbf{z}_{*,T}, \textbf{z}_{*,T})$}   & 5.50(4.68, 6.87) & 5.50(4.72, 6.89) & 5.50(4.71, 6.91) & 5.50(4.73, 6.90) & 5.50(4.70, 6.89) & 5.50(4.70, 6.88) & 5.50(4.69, 6.91) \\[6pt]
{Non-price} & {$\theta(\textbf{z}_{T}, \textbf{z}_{*,T})$} & 7.33(6.34, 9.03) & 6.65(5.74, 8.33) & 6.07(5.44, 7.07) & 8.70(7.62, 10.56) & 7.94(7.23, 9.14) & 7.36(6.63, 8.46) & 9.47(8.70, 10.69) \\
{value-attentive} &{$\theta(\textbf{z}_{T}, \textbf{z}_{T})$}  & 7.06(6.00, 9.43) & 6.25(5.35, 7.81) & 5.95(5.28, 7.05) & 7.94(6.80, 10.18) & 7.18(6.82, 7.73) & 6.73(6.04, 7.90) & 8.05(7.56, 8.70) \\
{customers}    & {$\theta(\textbf{z}_{*,T}, \textbf{z}_{*,T})$}   & 5.50(4.69, 6.84) & 5.50(4.68, 6.91) & 5.50(4.70, 6.95) & 5.50(4.70, 6.86) & 5.50(4.72, 6.91) & 5.50(4.69, 6.88) & 5.50(4.70, 6.91) \\[6pt]
{Price-attentive} & {$\theta(\textbf{z}_{T}, \textbf{z}_{*,T})$} & 7.48(6.29, 9.93) & 6.76(5.83, 8.34) & 6.20(5.55, 7.22) & 8.84(7.54, 11.34) & 7.96(7.19, 9.29) & 7.50(6.75, 8.60) & 9.48(8.62, 10.83) \\
 {customers}   & {$\theta(\textbf{z}_{T}, \textbf{z}_{T})$}  & 7.32(6.17, 9.79) & 6.35(5.46, 7.79) & 6.05(5.40, 7.00) & 8.17(6.97, 10.65) & 7.31(6.91, 7.91) & 6.82(6.15, 7.81) & 8.25(7.68, 8.93) \\
          & {$\theta(\textbf{z}_{*,T}, \textbf{z}_{*,T})$}   & 5.61(4.76, 7.01) & 5.61(4.78, 6.99) & 5.60(4.76, 6.95) & 5.61(4.78, 6.97) & 5.60(4.77, 6.97) & 5.60(4.76, 6.98) & 5.60(4.77, 6.97) \\[6pt]
{Non-active}  & {$\theta(\textbf{z}_{T}, \textbf{z}_{*,T})$} & 7.45(6.39, 9.53) & 6.76(5.83, 8.34) & 6.20(5.54, 7.18) & 8.82(7.68, 11.11) & 8.02(7.29, 9.12) & 7.50(6.75, 8.60) & 9.54(8.76, 10.74) \\
{customers}    & {$\theta(\textbf{z}_{T}, \textbf{z}_{T})$}  & 7.06(6.02, 9.40) & 6.35(5.48, 7.86) & 6.05(5.40, 7.02) & 7.94(6.81, 10.30) & 7.18(6.82, 7.73) & 6.82(6.14, 7.85) & 8.05(7.57, 8.69) \\
          & {$\theta(\textbf{z}_{*,T}, \textbf{z}_{*,T})$}   & 5.60(4.79, 7.00) & 5.60(4.78, 6.93) & 5.60(4.79, 6.97) & 5.60(4.77, 6.99) & 5.60(4.78, 7.02) & 5.60(4.76, 6.97) & 5.60(4.78, 6.98) \\[6pt]
\hline
\end{tabular}%
}
\end{sidewaystable}

\end{document}